\documentclass{article}
\usepackage[utf8]{inputenc}
\usepackage{fullpage}
\usepackage{amsmath,amssymb,amsthm,mathrsfs}
\usepackage{braket,physics}
\usepackage{complexity}
\usepackage[ruled,linesnumbered,noend]{algorithm2e}
\usepackage{setspace}
\usepackage[noend]{algpseudocode}
\usepackage{parskip}
\usepackage[colorlinks,linkcolor=Red,citecolor=Blue]{hyperref}
\usepackage[capitalize]{cleveref}
\usepackage[dvipsnames]{xcolor}
\usepackage{subcaption}
\usepackage{tikz}

\title{Approaching the Quantum Singleton Bound\\with Approximate Error Correction}
\author{Thiago Bergamaschi\thanks{UC Berkeley. Email: \texttt{thiagob@berkeley.edu}.}, Louis Golowich\thanks{UC Berkeley. Email: \texttt{lgolowich@berkeley.edu}. Supported by an NSF Graduate Research Fellowship.}, Sam Gunn\thanks{UC Berkeley. Email: \texttt{gunn@berkeley.edu}. Supported by a Google PhD Fellowship.}}
\date{\today}

\newif\ifnotes
\notestrue
\newcommand{\Enc}{\mathsf{Enc}}

\newcommand{\Dec}{\mathsf{Dec}}
\newcommand{\tDec}{\widetilde{\mathsf{Dec}}}
\newcommand{\Id}{\mathbb{I}}

\newtheorem{theorem}{Theorem}[section]
\newtheorem{proposition}[theorem]{Proposition}
\newtheorem{lemma}[theorem]{Lemma}
\newtheorem{corollary}[theorem]{Corollary}
\newtheorem{claim}[theorem]{Claim}
\newtheorem{definition}{Definition}[section]
\newtheorem{remark}{Remark}[section]
\newtheorem{fact}{Fact}[section]

\newcommand{\comment}[1]{}

\newcommand{\from}{\leftarrow}

\newcommand{\RSS}{\mathsf{RSS}}
\newcommand{\Share}{\mathsf{Share}}
\newcommand{\Reconstruct}{\mathsf{Reconstruct}}

\renewcommand{\epsilon}{\varepsilon}

\SetKwInput{KwInput}{Input}

\begin{document}

\maketitle

\begin{abstract}
    It is well known that no quantum error correcting code of rate $R$ can correct adversarial errors on more than a $(1-R)/4$ fraction of symbols. But what if we only require our codes to approximately recover the message?

In this work, we construct efficiently-decodable approximate quantum codes against \emph{adversarial} error rates approaching the quantum Singleton bound of $(1-R)/2$, for any constant rate $R$. Specifically, for every $R \in (0,1)$ and $\gamma>0$, we construct codes of rate $R$, message length $k$, and alphabet size $2^{O(1/\gamma^5)}$, that are efficiently decodable against a $(1-R-\gamma)/2$ fraction of adversarial errors and recover the message up to inverse-exponential error $2^{-\Omega(k)}$.

At a technical level, we use classical robust secret sharing and quantum purity testing to reduce approximate quantum error correction to a suitable notion of quantum list decoding. We then instantiate our notion of quantum list decoding by (i) defining and constructing folded quantum Reed-Solomon codes, and (ii) applying a new, quantum version of distance amplification.


\end{abstract}

\newpage

\thispagestyle{empty}
\newpage
\pagenumbering{roman}
\setcounter{tocdepth}{2}
{\small\tableofcontents}
\newpage
\pagenumbering{arabic}

\section{Introduction}
There are several models of noise in quantum error correction. One of the weakest noise models is erasures, where some known subset of symbols are removed. The number of erasure errors that a code can tolerate is the \emph{distance} of the code. At the other extreme there are adversarial errors, where an adversary is allowed to tamper arbitrarily with a subset of symbols (and the decoder doesn't know which symbols were corrupted). The number of symbols that an adversary can be allowed to corrupt without sacrificing correctness of the decoded message is the \emph{decoding radius}.

The quantum Singleton bound imposes an information-theoretic limit on the distance of a quantum code: No code of rate $R$ can correct from more than a $(1-R)/2$ fraction of erasures. But how far can we push the decoding radius? For the case of exact quantum error correction, it is a well-known consequence of the KL conditions \cite{Knill1997TheoryOQ} that any code of distance $d$ has decoding radius at most $d/2$. Combined with the quantum Singleton bound, it follows that quantum codes cannot achieve decoding radius beyond a $(1-R)/4$ fraction of the block length.

However, in \cite{CGS05} it was shown that, by relaxing the definition of quantum error correction to allow even \emph{exponentially-small} error in the recovered state, one can design certain codes where the decoding radius essentially matches the quantum Singleton bound for very small (asymptotically decaying) rates. Unfortunately, that construction requires an exponentially large alphabet and is therefore unlikely to be useful for communicating over a noisy channel.

Still, the codes of \cite{CGS05} indicate that the distance and the decoding radius might not be inherently far apart in approximate quantum error correction! This is in stark contrast to the classical setting, where it is impossible to build codes with decoding radius nearly matching the classical Singleton bound of $1-R$ for small rates\footnote{Indeed, given half of a codeword, an adversary could always substitute it with half of a fresh encoding of an arbitrary message. Therefore, decoding from adversarial errors up to the classical Singleton bound is impossible for $R < 1/2$.}. Whether the quantum phenomenon was an artefact of the non-standard exponential alphabet used by \cite{CGS05}, or an exciting possibility for practical quantum error correcting codes, has essentially remained open.

In this work we construct the first efficiently-decodable approximate quantum error correcting codes over constant-sized alphabets with decoding radius approaching the quantum Singleton bound for all rates. Our codes have error that decays exponentially with the message size, and conclusively show that even a slight relaxation of standard quantum error correction permits the design of codes where the the distance and decoding radius are nearly identical.

\subsection{Our Contributions}
\paragraph{Main results.}
Our main results are constructions of efficient approximate quantum error-correcting codes (AQECCs) approaching the quantum Singleton bound. We present two constructions, one in \Cref{cor:aqeccmain} and the other in \Cref{corollary:aqeccswoLR}, based on different approaches and with qualitatively slightly different guarantees. Either construction yields the following main result. We use the notation $[[n,k]]_q$ to refer to a quantum code which has block length $n$, message length $k$, and alphabet size (i.e., local dimension) $q$.

\begin{theorem}[Restatement of \Cref{cor:aqeccmain} and \Cref{corollary:aqeccswoLR}]
\label{thm:aqeccmaininf}
    For any $\gamma>0$ and $0<R<1$, there exists a family of approximate quantum $[[n,R \cdot n]]_q$ codes with constant alphabet size $q=q(\gamma)$ that correct errors acting on $(1-R-\gamma) \cdot n/2$ registers, up to a recovery error of $2^{-\Omega(n)}$.
\end{theorem}

All of the codes we construct come with efficient algorithms for encoding and decoding. In \Cref{sec:AQECC} we present a general compiler that builds AQECCs from
\begin{itemize}
    \item quantum list-decodable codes,
    \item purity testing codes, and
    \item classical robust secret sharing schemes.
\end{itemize}
The purity testing codes of \cite{barnum2002authentication} and robust secret sharing schemes of \cite{CDD+15} are sufficient for our purposes. Constructing good quantum list-decodable codes is the main technical part of this work. We define a quantum stabilizer code to be $(\tau,L)$ quantum list-decodable if each error syndrome is compatible with at most $L$ logically-equivalent Pauli errors of weight $\leq\tau n$ (\Cref{def:QLD}).

We construct CSS codes which are efficiently list-decodable up to the quantum Singleton bound by introducing \emph{folded quantum Reed-Solomon codes}, a quantum version of the codes from \cite{Guruswami2008ExplicitCA}.

\begin{theorem}[Restatement of \Cref{theorem:fqrs}]
\label{thm:fqrslargealph}
    For any $\gamma>0$ and $0<R<1$, there exists a family of quantum $[[n, R \cdot n]]_q$ stabilizer codes with alphabet size $q = n^{O(1/\gamma^2)}$ which is $((1-R-\gamma)/2, n^{O(1/\gamma)})$ quantum list-decodable.
\end{theorem}

As in \cite{Guruswami2008ExplicitCA}, the alphabet size of our folded codes scales with the block-length $n$. To reduce the alphabet size to a constant, we introduce a quantum version of the expander-based distance amplification and alphabet reduction techniques of Alon, Edmonds, and Luby \cite{AEL95}. This yields the following theorem.

\begin{theorem}[Restatement of \Cref{thm:frsael}]
\label{theorem:contrib-main-list-decodable}
    For any $\gamma>0$ and $0<R<1$, there exists a family of quantum $[[n, R \cdot n]]_q$ stabilizer codes with constant alphabet size $q=2^{O(1/\gamma^5)}$ which is $((1-R-\gamma)/2,n^{O(1/\gamma^3)})$ quantum list-decodable.
\end{theorem}

Together with our compiler, \Cref{theorem:contrib-main-list-decodable} allows us to obtain our main result \Cref{thm:aqeccmaininf}.

\paragraph{Additional results.} We find our quantum version of distance amplification useful for three additional purposes. First, using recent constructions of good quantum LDPC codes \cite{panteleev2022asymptotically,leverrier2022quantum,leverrier2022efficient,dinur2022good,gu2022efficient,leverrier2022parallel} we obtain quantum LDPC codes with distance approaching the quantum Singleton bound:

\begin{theorem}[Restatement of \Cref{cor:qLDPC}]
\label{thm:qLDPCinf}
    For any $\gamma>0$ and $0<R<1$, there exists a family of quantum LDPC $[[n, R \cdot n,(1-R-\gamma) \cdot n/2]]_q$ stabilizer codes with constant alphabet size $q=q(\gamma)$.
\end{theorem}

Note that these codes are exact quantum error correcting codes, and are therefore only unique-decodable from errors on $(1-R-\gamma)/4$ fraction of registers.

Second, applying a more involved construction using AEL and quantum list decoding, we obtain the following LDPC version of \Cref{thm:aqeccmaininf}, assuming access to a classical side-channel.

\begin{theorem}[Restatement of \Cref{cor:outercodes}]
\label{thm:privateLDPCinf}
    For any $\gamma>0$ and $0<R<1$, there exists a family of classical side-channel assisted approximate quantum LDPC $[[n, R \cdot n]]_q$ stabilizer codes with constant alphabet size $q=q(\gamma)$ that can correct errors acting on $(1-R-\gamma) \cdot n/2$ registers, up to a recovery error of $2^{-\Omega(n)}$.
\end{theorem}

Third, we obtain alternative codes satisfying \Cref{thm:aqeccmaininf} in \Cref{corollary:aqeccswoLR}.




\subsection{Technical Overview}
We begin by defining AQECCs in \Cref{sec:techo-aqec}. Then we present our definition of quantum list-decoding for stabilizer codes in \Cref{sec:techo-def-qld}, along with constructions in \Cref{sec:techo-construct-qld}. In \Cref{sec:techo-construct-aqec} we tie these ideas together to construct AQECCs from quantum list-decodable codes, purity testing codes, and robust secret sharing schemes.

We then turn to our quantum version of distance amplification in \Cref{sec:techo-list-recovery}. Finally, we outline our use of these techniques to construct private LDPC codes in \Cref{sec:techo-ldpc}.

The discussion below assumes familiarity with basic notions of quantum error correction. The reader is referred to \Cref{sec:prelim} or \cite{Gottesman1997StabilizerCA} for a review.

\subsubsection{Approximate Quantum Error Correction}
\label{sec:techo-aqec}
An approximate quantum error correcting code should encode message states into (possibly mixed) code-states, with the guarantee that the original messages are recoverable up to a high fidelity even when an adversary is allowed to tamper with a bounded number of registers. First, we define the syntax of a quantum code.

\begin{definition}
    An $[[n, k]]_q$ quantum code is a pair $(\Enc, \Dec)$ of quantum channels on qudits of local dimension $q$ such that 
    \begin{enumerate}
        \item $\Enc$ encodes a $k$ qudit message register $M$ into an $n$ qudit code-state register $C$, and
        \item $\Dec$ decodes a received code-state register $\hat{C}$ into a $k$ qudit message register $\hat{M}$.
    \end{enumerate}
\end{definition}

The code is said to be \textit{efficient} if $(\Enc, \Dec)$ can be implemented in polynomial time in $n$. If the code-state register $C$ is defined on $n$ qudits, then the \emph{weight} of an error channel acting on $C$ is the number of qudits of $C$ which the channel acts non-trivially on.

\begin{definition}[\cite{CGS05}]
\label{def:techo-aqec}
    A quantum code $(\Enc, \Dec)$ is said to be a $(\delta, \varepsilon)$-AQECC (Approximate Quantum Error Correcting Code) if, for all adversarial channels $\mathcal{A}$ of weight $\delta\cdot n$,
    $$\norm{\Dec \circ \mathcal{A} \circ \Enc - \mathbb{I}}_\diamond \leq \varepsilon.$$
\end{definition}
We refer the reader to \Cref{sec:prelim} for a review of distance measures between quantum states and channels.

While the ordinary quantum Singleton bound only applies to exact quantum error correcting codes, more robust versions for AQECCs also hold. Specifically, \cite[Theorem~5]{Mamindlapally2022SingletonBF} provides a general quantum Singleton bound for hybrid quantum/classical/entanglement-assisted codes with approximate or exact recovery.

\begin{theorem}[\cite{Mamindlapally2022SingletonBF}]
\label{thm:MWqsing}
If a $(\delta,\epsilon)$-AQECC of local dimension $q$ has rate $R$, then
\begin{equation*}
    \delta \le \frac{1-R}{2} + O(\varepsilon \cdot c)
\end{equation*}
where $c = 1+\log_q(1/\varepsilon)/n$.
\end{theorem}
This bound also holds for AQECCs with side-information, or ``private AQECCs'' (\Cref{def:privateaqecc}).

\subsubsection{Defining Quantum List-Decodable Codes}
\label{sec:techo-def-qld}
Recall that stabilizer codes \cite{Gottesman1997StabilizerCA} are a subspace of $n$-qudit states, defined as the joint eigenspace of a set of commuting operators called the \emph{stabilizer generators}. Given a corrupted code-state, the syndrome measurement is the projective measurement of the eigenvalues of the stabilizer generators. The measurement outcome (a list of eigenvalues) is referred to as the syndrome, and is the quantum analogue of a classical parity check syndrome.

In classical coding theory, a code $C\subset \Sigma^n$ is said to be $(\tau, L)$ list-decodable \cite{woz1958, Elias1957ListDF} if there are at most $L$ codewords of $C$ in any hamming ball of radius $\tau\cdot n$ in $\Sigma^n$. To understand our quantum analogue, observe that for linear codes this classical notion is exactly equivalent to requiring that there are at most $L$ errors of weight at most $\tau\cdot n$ consistent with each possible parity check syndrome.

Perhaps the most direct quantum translation would be the following: A stabilizer code $Q$ could be $(\tau,L)$ quantum list-decodable if for every possible syndrome $s$, there are at most $L$ Pauli operators of weight at most $\tau\cdot n$ that are consistent with the syndrome $s$. However, if the code has many low-weight stabilizers (such as concatenated codes) then the list of errors consistent with a given syndrome could contain exponentially many logically-equivalent Pauli operators. Since such logically-equivalent errors by definition have the same action on codewords, we instead use a slightly relaxed definition where we only count logically-distinct operators in the list.


\begin{definition}[Restatement of \Cref{def:QLD}]
A stabilizer code $Q$ is $(\tau,L)$ quantum list-decodable (QLD) if for every possible syndrome $s$, there are at most $L$ logically-distinct Pauli operators of weight at most $\tau\cdot n$ that are consistent with the syndrome $s$.
\end{definition}

Therefore, after measuring the syndrome, quantum list-decoding is the entirely classical task of computing the list of possible corrections for the given syndrome.

\subsubsection{Constructions of Quantum List-Decodable Codes}
\label{sec:techo-construct-qld}
Our main tool for designing these quantum list-decodable codes is a reduction to classical list-decoding via CSS codes. 

\begin{claim} \label{claim:CSSintro}
    If $C_1, C_2\subset (\mathbb{F}_q^m)^n$ are two $\mathbb{F}_q$-linear classical codes which are $(\tau,L)$ list-decodable and $C_2^\perp\subset C_1$, then there exists a CSS code CSS$(C_1, C_2)$ which is $(\tau,L^2)$ quantum list-decodable. Moreover, if the classical codes are efficiently list-decodable, then the CSS code is as well.
\end{claim}

This follows from the fact that correcting errors on a code-state of a CSS code simply consists of correcting the $X$ and $Z$ errors.\footnote{Technically we consider generalized Paulis, the ``shift'' and ``phase'' operators. We refer the reader to \Cref{sec:prelim} or \cite{Ashikhmin2001NonbinaryQS} for a review of stabilizer codes over finite fields.}

The central insight behind our quantum code construction --- and later the quantum distance amplification in \Cref{sec:techo-list-recovery} --- is that \emph{syntactic} or \emph{physical} operations on the registers of a quantum code are not only well-defined, but behave similarly to their classical counterparts. For instance, folding a quantum CSS code (i.e., bundling together $m$ consecutive qudits into a larger qudit) induces folded component classical codes:

\begin{claim}[Informal statement of \Cref{claim:foldingcsscodes}]
\label{claim:techo-folding}
    If $Q$ is the CSS code of two linear classical codes $C_1, C_2\subset \mathbb{F}_q^n$, then the $m$-folded quantum code $Q^{(m)}$ is a CSS code of the two $\mathbb{F}_q$-linear $m$-folded classical codes $C_1^{(m)}, C_2^{(m)}\subset (\mathbb{F}_q^m)^{n/m}$.
\end{claim}

This follows from the fact that folding operation automatically bundles the qudits \emph{in both bases}, so error correction in each basis can leverage folding. Instantiating \Cref{claim:techo-folding} with the folded Reed-Solomon codes of \cite{Guruswami2008ExplicitCA}, we obtain explicit quantum list-decodable codes tolerating errors up to the quantum Singleton bound:

\begin{theorem}[Restatement of \Cref{theorem:fqrs}]
    For any $\gamma>0$ and $0<R<1$, there exists a family of quantum $[[n, R \cdot n]]_q$ stabilizer codes with alphabet size $q = n^{O(1/\gamma^2)}$ which is $((1-R-\gamma)/2, n^{O(1/\gamma)})$ quantum list-decodable.
\end{theorem}

Unfortunately, these folded quantum codes have alphabet size $n^{O(1/\gamma^2)}$ rather than constant. \cite{Guruswami2008ExplicitCA} reduce the alphabet size by showing that their codes can be used for the related task of \textit{list-recovery}, and combining their ideas with the distance amplification and alphabet reduction techniques of \cite{AEL95}. The main technical component of this work is to show that an analogous quantum list-recovery task can be used in combination with a quantum analog of the distance amplification techniques by \cite{AEL95}. We outline quantum distance amplification in \Cref{sec:techo-list-recovery}.

Using quantum distance amplification, we present a randomized construction of quantum stabilizer codes on constant sized alphabets, with list decoding radius approaching the quantum Singleton bound.
\begin{theorem}[Restatement of \Cref{thm:frsael}]
    For any $\gamma>0$ and $0<R<1$, there exists a family of quantum $[[n, R \cdot n]]_q$ stabilizer codes with constant alphabet size $q=2^{O(1/\gamma^5)}$ which is $((1-R-\gamma)/2,n^{O(1/\gamma^3)})$ quantum list-decodable. Moreover, there exists an efficient randomized algorithm to construct a classical description of the stabilizers of the code, which succeeds with high probability.
\end{theorem}

\subsubsection{Construction of AQECCs from Quantum List-Decodable Codes}
\label{sec:techo-construct-aqec}
Quantum list decoding enables us to find a short list of possible errors which match a syndrome measurement. If we had some way of tagging valid messages, then we could filter through the list for the true error and obtain \emph{unique decoding}.

We find that purity testing codes (PTCs) \cite{barnum2002authentication} are well-suited for this task. A PTC is a collection $\{Q_k\}_{k \in K}$ of stabilizer codes which is able to faithfully detect a Pauli error with high probability over keys $k \from K$. They were originally introduced to construct quantum authentication schemes, but it turns out that we do not need full quantum authentication for our purposes.

\begin{definition} [Purity testing codes, \cite{barnum2002authentication}] \label{def:PTC}
    A stabilizer purity testing code with error $\varepsilon$ ($\varepsilon$-PTC) is a set of stabilizer codes $\{Q_k\}_{k \in K}$ such that, for all Pauli errors $E$,
    $$\Pr_k\big[ E\in N(Q_k) \setminus S(Q_k)] \leq \varepsilon$$
    where $S(Q_k)$ is the stabilizer group of $Q_k$ and $N(Q_k)$ is the normalizer group.
\end{definition}

We construct private AQECCs by first encoding the message in a PTC $Q_k$, and subsequently encoding in a quantum list-decodable code $Q_{LD}$. As long as the keys are kept secret from the channel, the probability that an undetectable error appears in the list is at most $\varepsilon \cdot L$ by a union bound. Therefore a recipient who knows $k$ can uniquely recover the message with probability $1-\varepsilon \cdot L$.

Of course, this construction gives a collection of codes $\{Q_{LD} \circ Q_k\}_{k \in K}$ rather than a single code. Since the channel cannot be allowed to see the keys $k$, we only obtain a relaxed version of AQECCs where the sender and receiver have shared private randomness (unknown to the error channel). We call such codes private AQECCs in analogy to the classical private codes of \cite{1366252, Guruswami2003ListDW}.

\begin{definition} \label{def:privateaqecc}
    A set of quantum codes $\{(\Enc_k, \Dec_k)\}_{k\in K}$ is said to be a $(\delta, \varepsilon)$ private AQECC with keys $K$ if, for all adversarial channels $\mathcal{A}$ of weight $\delta\cdot n$,
    $$\norm{\mathbb{E}_{k \from K}\left[\Dec_k\circ \mathcal{A} \circ \Enc_k\right] - \mathbb{I}}_\diamond \leq \varepsilon.$$
\end{definition}

By using $k \from K$ as the shared private randomness, we obtain the following theorem.

\begin{lemma} [Restatement of \Cref{lemma:private-aqec}]
    Let $\{Q_k\}$ be an $\varepsilon$-PTC where each $Q_k$ is an $[[n_{PTC}, m]]_q$ stabilizer code, and let $Q_{LD}$ be a $[[n, n_{PTC}]]_q$ stabilizer code which is $(\delta, L)$ quantum list-decodable. Then $Q_{LD} \circ Q_{k}$ is a $(\delta, 2\cdot L\cdot \varepsilon)$ private AQECC.
\end{lemma}

In this construction, it becomes important that the adversary cannot choose the error according to the key $k$. The local indistinguishability of $Q_{LD}$ protects against such attacks, by ensuring that the adversary's view consisting of fewer than $\text{dist}(Q_{LD})$ components of a code-state of $Q_{LD}\circ Q_k$ is the same regardless of the value of $k$.

To convert a private AQECC into an AQECC, we incorporate the classical key information into the quantum registers. The key tool for this conversion is a classical robust secret sharing schemes (RSSs) \cite{Cramer1999EfficientMC, Cramer2001OnTC, Cramer2008DetectionOA, CDD+15}\footnote{Also known as an Error Tolerant Secret Sharing scheme (ETSS).}, which ensures that the hidden secret key looks uniformly random to the adversary while still allowing the receiver to reconstruct the key with high probability. \cite{CGS05} used a combination of classical secret sharing and classical authentication from \cite{Cramer2001OnTC} to construct such an RSS. We need to be more careful with our choice of classical secret sharing scheme, in order to ensure that the resulting quantum code still lies roughly on the quantum Singleton bound and maintains constant alphabet. In \cref{sec:AQECC}, we use a more recent RSS construction of \cite{CDD+15} to prove the following theorem.
\begin{theorem}[Informal statement of \Cref{theorem:removing-privacy}]
For any $[[n,k]]_q$ code that is a $(\delta, \varepsilon)$ private AQECC with sufficiently short keys, there is a $[[n,(1-O\left(\frac{1}{\log(q)}\right)) \cdot k]]_{O(q)}$ code that is a $(\delta, \varepsilon+2^{-\Omega(n)})$ AQECC.
\end{theorem}

\subsubsection{Distance Amplification for Quantum Codes}
\label{sec:techo-list-recovery}
In order to decrease the alphabet size of our quantum code constructions and improve on their encoding/decoding efficiency, we develop a quantum analog of the seminal distance amplification techniques for classical error correcting codes by Alon, Edmonds and Luby \cite{AEL95}. Their ideas hinged on using code concatenation and a careful redistribution (or permutation) of the symbols of the code, in order define a new code with a larger distance and/or smaller alphabet (and slightly lower rate). These techniques have found extensive use in classical coding theory, such as the construction of linear time encodable/decodable codes \cite{Guruswami2002NearoptimalLC}, the construction of list-decodable codes from list-recoverable codes \cite{guruswami2003linear, Guruswami2008ExplicitCA}, high-rate locally testable codes \cite{Kopparty2015HighrateLA, gopi2018locally,hemenway2019local}, and many other applications. 

The simplest formulation of the quantum analog we study can be understood as physically permuting and re-grouping the qudits of a concatenated quantum code \cite{Knill1996ConcatenatedQC}\footnote{See \Cref{def:codeconcatenation} for a review of quantum code concatenation.}. Starting with an outer $[[n, k,n\cdot \Delta_{out}]]_{q^m}$ stabilizer code $Q_{out}$, and an inner $[[n', m, n'\cdot \Delta_{in}]]_q$ stabilizer code $Q_{in}$, their quantum code concatenation is the $[[n\cdot n', k\cdot m]]_q$ stabilizer code $Q_\diamond = Q_{out}\diamond Q_{in}$. Consider a physical depiction of $Q_\diamond$ as a $n \times n'$ table of $q$-ary qudits. The final step is to use an $n'$-biregular expander graph $G = (L\cup R, E)$ on two partitions of $|L|=|R|=n$ nodes, to redistribute the qudits of $Q_\diamond$ into another $[[n, km/n']]_{q^{n'}}$ stabilizer code $Q_G$. To do so, if the $i$th edge incident to $u\in L$ equals the $j$th edge incident to $v\in R$, then the $i$th $q$-ary qudit of the $u$th inner block of $Q_\diamond$ is placed into the $j$th $q$-ary qudit of the $v$th inner block of $Q_G$. In other words, the table associated to $Q_\diamond$ is permuted into another $n$ by $n'$ table for $Q_G$, and each of its rows represents a $q^{n'}$-ary qudit. Qualitatively, this qudit redistribution ensures that any $q^{n'}$-ary Pauli $E$ (represented as a tensor product of $q$-ary Pauli's) supported on roughly a $< \Delta_{in}/2$ fraction of the rows of $Q_G$, after unpermuting $G$, is then almost uniformly spread out among the inner blocks/rows of $Q_\diamond$ such that most inner blocks have less than $\Delta_{in}/2$ fraction of errors. If the fraction of the `high weight' inner blocks is less than $\Delta_{out}/2$, then $E$ is a perfectly correctable error for $Q_G$, and thus $Q_G$ has relative distance roughly $\Delta_{in}$. 

\begin{theorem} [Informal statement of \Cref{prop:aelbasic}]
    If $G$ is an $\varepsilon$-pseudorandom expander graph (\Cref{def:expander}), then $Q_G$ has relative distance $\geq \Delta_{in} - 2\cdot \varepsilon \cdot \sqrt{\Delta_{in}/\Delta_{out}}$.
\end{theorem}

In \Cref{sec:ael} we argue that the insights from \cite{AEL95} translate remarkably well into the quantum setting. We analogously adapt these ideas, and the intuition above to decrease the alphabet of our quantum codes, construct quantum list decodable codes from a suitable analog of quantum list recoverable codes, and to construct private AQECC's which are LDPC and encodable/decodable in linear time.

\subsubsection{Private LDPC AQECCs with Linear-Time Decoding}
\label{sec:techo-ldpc}
Recall that the construction of efficient private AQECCs described in \Cref{sec:techo-construct-aqec} first constructs efficient quantum list-decodable codes using AEL amplification with inefficient but small quantum list-decodable inner codes, and then composes these AEL-amplified codes with PTCs to obtain the desired private AQECCs. Therefore while this construction provides a general recipe for AQECCs, as mentioned previously the outer code of the AEL concatenation must satisfy an appropriate form of quantum list-recoverability, which we only know how to achieve with folded quantum Reed-Solomon codes.

In \Cref{sec:aqeccdirect}, we present a construction of efficient private AQECCs that permits more flexibility in the code properties by combining the same ingredients described above in an alternative manner. Specifically, if we instead apply AEL amplification with a certain class of inefficient but small private AQEC inner codes, the outer code may be chosen to be any asymptotically good efficient quantum code $Q_{\text{out}}$. The resulting construction will be an efficient private AQECC that inherents some of the properties of $Q_{\text{out}}$, such as its locality and decoding time. Thus choosing the outer code from a family of good quantum LDPC codes \cite{panteleev2022asymptotically,leverrier2022quantum,leverrier2022efficient,dinur2022good,gu2022efficient,leverrier2022parallel}, we obtain the efficient private LDPC AQECCs of \Cref{thm:privateLDPCinf}.

\subsection{Related Work}
\label{sec:relwork}
The idea of using approximate quantum error correction to build better quantum codes goes back to \cite{LNCY97}, who construct a 4-qubit AQECC encoding one qubit that tolerates a specific noise model.

A more closely related result is that of \cite{CGS05}, who construct AQECCs with a decoding radius of $\lfloor (n-1)/2 \rfloor$ for codes of block length $n$. Their construction is based on a combination of quantum error correction, quantum message authentication, and classical secret sharing. However, the code in \cite{CGS05} has asymptotically decaying rate, and is defined on qudits with exponentially large local dimension. To go beyond the exponential-alphabet regime we use quantum list decoding.

A related notion of quantum list decoding was introduced by \cite{LS08}, who show that random stabilizer codes on qubits are quantum list-decodable with high probability. They use these codes to construct private AQECCs on qubits with rates up to $1-H(\delta)-\delta\log 3$ (where $\delta \cdot n$ is the decoding radius).

Compared to this work, \cite{LS08} have a slightly different definition of quantum list decoding and do not construct codes with efficient decoding procedures. In fact, we expect that an efficient decoding algorithm for random stabilizer codes should not exist. Moreover, their codes do not approach the quantum Singleton bound because they restrict to binary alphabets. While extending their results to larger (but still constant) alphabets would yield codes that do approach the quantum Singleton bound, the problem of efficient decoding would remain.

Several works have studied conditions in which approximate quantum error correction is possible \cite{Barnum2000ReversingQD, SW02, Devetak2005ThePC, Kretschmann2008TheIT, Bny2009ConditionsFT, MN10, MN12, BO10, Hayden2017ApproximateQE}. For instance, \cite{SW02,  Kretschmann2008TheIT, Bny2009ConditionsFT, BO10} were interested in characterizing when information can be sent over a noisy quantum channel, with or without prior knowledge of the input state and the noise model. \cite{MN10, MN12} studied the efficacy and optimality of the \textit{transpose channel} as a recovery map, when the noise process is known. \cite{Devetak2005ThePC, Hayden2017ApproximateQE} studied applications of approximate quantum error correction to quantum and classical communication, including to the tasks of state merging, entanglement distillation, quantum identification and remote state preparation.

\cite{Brun2006CorrectingQE} introduced the concept of an entanglement-assisted quantum error correcting code (EAQECC), when the sender and receiver have access to an unbounded supply of pre-shared EPR pairs. A number of results \cite{Lai2013DualityIE, Lai2013EntanglementIT, Wilde2011EntanglementAssistedQT} studied properties of EAQECCs, and discussed their construction from classical linear codes. While no pure quantum code can bypass the quantum Singleton bound, quite remarkably,  recently EAQECCs were shown to be able to transmit quantum information up to the classical Singleton bound \cite{Grassl2020EntanglementAssistedQC} in certain parameter regimes. \cite{Devetak2003TheCO, Hsieh2008EntanglementAssistedCO, grassl2022entropic, Mamindlapally2022SingletonBF} proved Singleton bounds and other upper bounds on the capacity of entanglement-assisted and private quantum codes.

\subsection{Discussion and Open Problems}
The folded Reed-Solomon codes \cite{Guruswami2008ExplicitCA} are a family of deterministic, polynomial-time constructible classical codes which approach the information-theoretic list-decoding limit. Since their result, a number of (explicit) codes have improved on their alphabet and list sizes \cite{Guruswami2010CyclotomicFF, Guruswami2020OptimalRL,Guo2022EfficientLW}, without the need of expander-based distance amplification \cite{AEL95}. However, it seems unclear how to use these codes to design quantum codes, as their duals aren't necessarily explicit, or they require pre-encoding with certain \textit{subspace-evasive} sets \cite{Dvir2012SubspaceES, https://doi.org/10.48550/arxiv.1704.05992} to reduce their list sizes. We leave for future work the design of explicit quantum codes approaching the list-decoding capacity at small qudit local dimensions. 

We note that the techniques in this work can be used to formalize a quantum analog of \textit{list-decoding from erasures}, where the receiver is tasked with producing a list of correction operators to a reduced density matrix of some code-state. This task has been extensively studied in the classical setting \cite{guruswami2001expander, Guruswami2003ListDF, guruswami2003linear, Guruswami2004LinearTimeLD, erasureldexpander}, where it is well known that random binary linear codes are efficiently erasure list-decodable even when all but a small fraction of symbols have been erased. One could ask an analogous question for quantum codes on qubits: For fixed rate $R$, what is the maximum number of erasures for which one can still efficiently list-decode (or approximately unique-decode) the quantum state?


\subsection{Organization}
We organize the rest of this work as follows. In \cref{sec:prelim} we review the necessary background on quantum error correction. In \cref{section:QLD} we present our definition of quantum list decoding and our results on list-decodable CSS codes, including our explicit construction of folded quantum Reed-Solomon codes. In \cref{sec:AQECC} we show how to reduce approximate quantum error correction to quantum list decoding. In \cref{sec:ael} we describe our quantum distance amplification and alphabet reduction techniques. We apply this AEL amplification in \Cref{sec:frsael} to construct efficient quantum list-decodable codes approaching the quantum Singleton bound over constant alphabets. We apply AEL in an alternative manner in \Cref{sec:aqeccdirect} to construct private LDPC AQECCs with linear-time decoders, which as a byproduct gives another construction of efficient AQECCs. 

In \Cref{section:randomCSSproperties} we discuss the list decodability and other simple properties of random CSS codes, which follow easily from their classical counterparts. For completeness, in \Cref{sec:listrecoverablequantumcodes} we introduce a quantum notion of list recoverability and a ``fully-quantum" proof of the construction of quantum list decodable codes on smaller alphabets from quantum list recoverable codes and distance amplification.

\section{Preliminaries}
\label{sec:prelim}

\subsection{Notation}

\textbf{Norms and Distances} We use three notions of distance between quantum states and quantum channels. The trace distance $\|\rho-\sigma\|_1$ between two mixed states $\rho, \sigma \in \mathcal{D}(\mathcal{H})$ captures their distinguishability, and is defined in terms of the trace norm (or Shatten 1-norm) $\|M\|_1 = \text{Tr}[\sqrt{M^\dagger M}]$. The fidelity $F(\rho, \sigma) = \|\sqrt{\rho}\sqrt{\sigma}\|_1^2$, and can be written simply as $\bra{\psi}\rho\ket{\psi}$ when $\sigma = \ket{\psi}\bra{\psi}$ is a pure state.

Finally, the diamond norm distance between two quantum channels $\mathcal{M}, \mathcal{N}:\mathcal{L}(\mathcal{H}_A)\rightarrow \mathcal{L}(\mathcal{H}_B)$, quantifies their distinguishability even in the presence of entanglement 
\begin{equation}
    \|\mathcal{N}-\mathcal{M}\|_\diamond = \sup_n \max_{\rho_{AE}}\| (\mathbb{I}_n\otimes \mathcal{N})(\rho_{AE}) - (\mathbb{I}_n\otimes \mathcal{M})(\rho_{AE})\|_1
\end{equation}

where $\rho_{AE}\in \mathcal{D}(\mathcal{H}_A\otimes \mathcal{H}_E)$.

\textbf{Finite Fields} Let $q =  p^m$ be a power of a prime $p$, and denote by $\mathbb{F}_{q}$ to be the Galois field of $p^m$ elements. We refer to the $\mathbb{F}_p$ functional $\text{tr}_{\mathbb{F}_{q}/\mathbb{F}_{p}}:\mathbb{F}_{q}\rightarrow \mathbb{F}_p$ as the trace function, where $\text{tr}_{\mathbb{F}_{q}/\mathbb{F}_{p}}(a) = \sum_{i=0}^{m-1}a^{p^i}$. If $\alpha_1, \cdots, \alpha_m$ is a basis of $\mathbb{F}_q$ over $\mathbb{F}_p$, then one can express $a\in \mathbb{F}_{q} $ as $a = \sum_{i=1}^m a_i \alpha_i$ for $a_i\in \mathbb{F}_p$. We refer to a pair of bases $\alpha = \alpha_1, \cdots, \alpha_m$, $\beta = \beta_1\cdots \beta_m$ of $\mathbb{F}_q$ as dual bases if $\text{tr}_{\mathbb{F}_{q}/\mathbb{F}_{p}}(\alpha_i\beta_j) = \delta_{i,j}$. If $a, b\in \mathbb{F}_{q}$ is expressed as $(a_1, \cdots, a_m)$, $(b_1, \cdots, b_m)$ in the dual bases $\alpha, \beta$ respectively, then the inner product over the basis representation becomes the trace:
\begin{equation}
   \big\langle a,  b\big\rangle =\sum_{i = 1}^{m} a_i b_i  = \sum_{i, j = 1}^{m} a_i b_j \text{tr}_{\mathbb{F}_{q}/\mathbb{F}_{p}}(\alpha_i\beta_j) = \text{tr}_{\mathbb{F}_{q}/\mathbb{F}_{p}}(ab)
\end{equation}

\textbf{Error Basis} To connect these notions to quantum codes, we define an explicit error basis for $q=p^m$-ary quantum systems. Fix $\omega = e^{2\pi i/p}$. Let $T, R$ be the `shift' and `phase' operators on $\mathbb{C}_p$, defined by

\begin{equation}
    T  = \sum_{x\in \mathbb{F}_p}\ket{x}\bra{x+1} \text{ and }R =  \sum_{x\in \mathbb{F}_p} \omega^x \ket{x}\bra{x}
\end{equation}

The operators $T^iR^j$, $i, j\in \mathbb{F}_p$ form the Weyl-Heisenberg operators, an orthonormal basis of operators over $\mathbb{C}_p$. If $a, b\in \mathbb{F}_q$, with representations $(a_1, \cdots a_m)$, $(b_1, \cdots b_m)$ in the dual bases $\alpha, \beta$ respectively, then one can define an o.n. basis of operators over $\mathbb{C}_{q}$: 

\begin{equation}
    E_{a, b} = X^a Z^b  = \bigotimes_{i\in [m]} T^{a_i} R^{b_i} \text{ and thus }  E_{a, b} E_{a', b'} = \omega^{\langle a, b'\rangle -\langle a', b\rangle   } E_{a', b'}E_{a, b}.
\end{equation}

Finally, if $\textbf{a} = (a^{(1)}, a^{(2)}\cdots a^{(n)}),  \textbf{b} = (b^{(1)}, b^{(2)}\cdots b^{(n)}) \in \mathbb{F}_{q}^n$, then one can define operators acting on $\mathbb{C}_q^{\otimes n}$ via $E_{\textbf{a}, \textbf{b}} = \otimes_{j\in [n]} E_{a^{(j)}, b^{(j)}}$. The error group $\mathcal{P}_{q}^n$ is the group generated by the $E_{\textbf{a}, \textbf{b}}$ and the phase $\omega \cdot \mathbb{I}_{q^n\times q^n}$. The weight $wt(E_{\textbf{a}, \textbf{b}})$ is the number of locations $j\in [n]$ where either $a^{(j)}, b^{(j)}$ are non-zero, and $\mathcal{P}_{q, \delta}^n\subset \mathcal{P}_{q}^n$ is the set of operators in the error group of weight less than $\delta\cdot n$.

\subsection{Quantum Error Correction}

\begin{definition} [\cite{Knill1997TheoryOQ}]
    An $[[n, k, d]]_q$ QECC (quantum error correcting code) is a subspace $Q\subset (\mathbb{C}^{q})^{\otimes n}$ of dimension $q^k$. Let $\Pi$ be the projection onto $Q$. If for all operators $E$ of weight $\leq d$, we have 
    \begin{equation}
    \Pi E\Pi = \eta_E \Pi 
    \end{equation}
\noindent where $\eta_E \in \mathbb{C}$ only depends on $E$, then $Q$ is said to have distance $d$.
\end{definition}

Quantum stabilizer codes are a class of QECCs defined as the joint eigenspace of a set $S$ of commuting operators, the stabilizers. We discuss a presentation of stabilizer codes over finite fields by \cite{Ashikhmin2001NonbinaryQS}. 

\begin{definition} [Stabilizer Codes over Finite Fields, \cite{Ashikhmin2001NonbinaryQS}]
    Let $q = p^m$ be a prime power. Let the set of operators $\{S_1, \cdots, S_{r}\}\subset \mathcal{P}_q^n$ generate a commutative subgroup $S$ of $\mathcal{P}_q^n$. Then the subspace $Q\subset \mathbb{C}_q^{\otimes n}$ defined by
    \begin{equation}
        Q = \bigg\{ |\psi\rangle\in \mathbb{C}_q^{\otimes n} : S_i \ket{\psi} = \ket{\psi} \text{ for all }i\in [r] \bigg\}
    \end{equation}
    forms a $[[n, n-\frac{r}{m}]]_q$ QECC. 
\end{definition}

We refer to the stabilizer group of a code $Q$ as $S(Q)$. Associated to a stabilizer group is the normalizer group $N(Q)$, the set of operators in the error group $\mathcal{P}_q^n$ which commute with all the elements of $S(Q)$. The elements of $N(Q) - S(Q)$ \footnote{By which we mean the set difference, to distinguish from the quotient group.} are the `undetectable errors' on $Q$, as they map code-states to distinct code-states, and the quotient group $N(Q)/S(Q)$ are the logical operators on $Q$. The distance of $Q$ is the lowest weight of any operator in $N(Q) - S(Q)$. The class of `detectable' errors of $Q$ are the operators outside $N(Q)$, which we detect by measuring the syndrome vector:

\begin{definition}
    For any operator $E\in \mathcal{P}_q^n$, we refer to the syndrome $s_E = (s_1, s_2\cdots, s_r)\in \mathbb{F}_p^{r}$ of the operator $E$ on the stabilizer code $Q$ as the phases $s_i$ defined by $S_i E = \omega^{s_i} ES_i$ for $i\in [r]$. 
\end{definition}

Note that if $E\in  N(Q)$, then the syndrome is 0, $s_E = 0^r$. Moreover, the syndrome is additive: if $E, E'\in \mathcal{P}_q^n$, then $s_{EE'} = s_E+s_{E'}$, and $s_{E^\dagger} = -s_E\mod p$.

We instantiate most of our constructions in this work using the CSS codes introduced by \cite{Calderbank1996GoodQE, Steane1996SimpleQE}, and their extensions to non-binary fields \cite{Ketkar2006NonbinarySC, Rtteler2004OnQM, Kim2008NonbinaryQE}. 

\begin{theorem} [Galois-qudit CSS codes] \label{theorem:galoiscss} 
    Fix two linear classical codes $C_1, C_2 \subset \mathbb{F}_q^n$ of dimension $k$, where $C_2^\perp \subset C_1$ and $C_1, C_2$ both have distance at least $d$. Let CSS$(C_1, C_2)\subset \mathbb{C}_q^{\otimes n}$ be the stabilizer code defined by the stabilizers $E_{\textbf{a}, \textbf{b}} \in \mathcal{P}_q^n$, $a\in (C_2)^\perp, b\in (C_1)^\perp$. Then CSS$(C_1, C_2)$ is a $[[n, 2k-n, d]]_q$ QECC. 
\end{theorem}

Moreover, the normalizer group of CSS$(C_1, C_2)$ are (up to a global phase) the operators $E_{\textbf{a}, \textbf{b}} \in \mathcal{P}_q^n$,  $a\in C_1, b\in C_2$. While a seemingly minor modification, it will later be relevant to discuss CSS codes over vector spaces as well. As the authors haven't found such a pre-existing definition in the literature, we discuss a short self-contained proof below, which follows immediately from the techniques by \cite{Ashikhmin2001NonbinaryQS}.

\begin{theorem} \label{theorem:vectorspacecss}
    Fix two $\mathbb{F}_q$-linear $[n, k]_{q^m}$ classical codes $C_1, C_2 \subset (\mathbb{F}_{q}^m)^{n}$ of dimension $k$, where $C_2^\perp \subset C_1$ and $C_1, C_2$ both have distance at least $d$. Let CSS$(C_1, C_2)\subset (\mathbb{C}_q^{\otimes m})^{\otimes n}$ be the stabilizer code defined by the stabilizers $E_{\textbf{a}, \textbf{b}} \in \mathcal{P}_q^{n\cdot m}$, $a\in (C_2)^\perp, b\in (C_1)^\perp$. Then CSS$(C_1, C_2)$ is a $[[n, 2k-n, d]]_{q^m}$ QECC. 
\end{theorem}

\begin{proof}
    We view the codewords of $a\in C_1^\perp, b\in C_2^\perp$ as elements of $(\mathbb{F}_{q})^{n\cdot m}$, and define the set of operators $E_{\textbf{a}, \textbf{b}}\in \mathcal{P}_q^{n\cdot m}$ using the same expansion as the Galois-qudit construction. By design, these operators are commuting and thus define a stabilizer code $Q$ over $\mathbb{C}_q^{\otimes n\cdot m}$, which we view as a code $Q'$ over $(\mathbb{C}_q^{\otimes m})^{\otimes n}$.
\end{proof}

\subsection{Composition and Concatenation of Quantum Codes}
\label{sec:concat}

If $Q$ is a $[[n, k]]_q$ stabilizer code, then let $\Enc_Q: \mathbb{C}_q^{\otimes k}\rightarrow \mathbb{C}_q^{\otimes n}$ be the isometry which encodes into $Q$. Recall that $\Enc_Q$ is a Clifford circuit, such that if $X\in \mathcal{P}_q^k$ is a operator in the Pauli group on $k$ qudits, then we let
\begin{equation}
    \bar{X} = \Enc_QX \Enc_Q^\dagger \in  \mathcal{P}_q^n,
\end{equation}

refer to the encoding of $X$ on $Q$. 

\begin{definition} \label{def:codecomposition}
    Let $Q_1$ be an $[[n, m]]_q$ stabilizer code, and $Q_2$ an $[[m, k]]_q$ stabilizer code. Then the composition $Q_\circ  = Q_1\circ Q_2$ of $Q_1, Q_2$ is the $[[n,k]]_q$ stabilizer code defined by encoding any message $\ket{\psi}\in \mathbb{C}_q^{\otimes k}$ as
    \begin{equation}
        \Enc_\circ(\psi) = \Enc_{Q_1}(\Enc_{Q_2}(\ket{\psi}))
    \end{equation}
\end{definition}

Since $Q_1, Q_2$ are stabilizer codes, $Q_\circ$ can alternatively be characterized in terms of its stabilizers:

\begin{fact}\label{fact:composedstabilizers}
    If $\{S^{1}_i\}_{i\in [r_1]}\subset \mathcal{P}_q^n$ and $\{S^{2}_j\}_{j\in [r_2]}\subset \mathcal{P}_q^m$ are the stabilizer generators of $Q_1, Q_2$ respectively, then the stabilizer generators of $Q_\circ$ are those of $Q_1$ together with those of $Q_2$ encoded into $Q_1$:
    \begin{equation}
        \{S^{1}_i\}_{i\in [r_1]}\cup \{\bar{S}^{2}_i\}_{j\in [r_2]} = \{S^{1}_i\}_{i\in [r_1]}\cup \{\Enc_QS^{2}_j \Enc_Q^\dagger\}_{j\in [r_2]}
    \end{equation}
\end{fact}

The notion of concatenated (classical) codes was developed in a seminal work by Forney in 1966 \cite{Forney2009ConcatenatedC}, and later introduced into a quantum setting by Knill and Laflamme \cite{Knill1996ConcatenatedQC} and Gottesman \cite{Gottesman1997StabilizerCA}. 

\begin{definition}\label{def:codeconcatenation}
    Let $Q_{out}$ be an $[[n,k]]_{q^m}$ `outer' stabilizer code, and $Q_{in}$ be a $[[n', m]]_q$ `inner' stabilizer code. Then the quantum code concatenation $Q_{out} \diamond Q_{in}$ is the $[[n\cdot n', k\cdot m]]_q$ stabilizer code defined by encoding any message $\ket{\psi}\in \mathbb{C}_{q^m}^{\otimes k}$ as
    \begin{equation}
       \Enc_\diamond(\psi)=  \bigg(\bigotimes_{i=1}^n \Enc_{in}\bigg) \Enc_{out}(\ket{\psi})
    \end{equation}
\end{definition}

That is, each qudit of the outer code is encoded into the inner code. Much like the case of code composition, the stabilizer generators of a concatenated quantum code can be expressed in terms of the stabilizers of the inner code $Q_{in}$, and the encoding of stabilizers of the outer code $Q_{out}$. 

\begin{fact}\label{fact:concatstabilizers}
    If $\{S^{out}_i\}_{i\in [r_{out}]}\subset \mathcal{P}_{q^m}^n$ and $\{S^{in}_j\}_{j\in [r_{in}]}\subset \mathcal{P}_q^{n'}$ are the stabilizer generators of $Q_{out}, Q_{in}$ respectively, then the stabilizer group of $Q_\diamond$ is generated by the following set of operators:

    \begin{equation}
        \bigg\{ \big(\bigotimes_{i=1}^n \Enc_{in}\big) S^{out}_j \big(\bigotimes_{i=1}^n \Enc_{in}^\dagger \big) \bigg\}_{j\in [r_{out}]} \bigcup  \bigg\{ \big(S^{in}_j\big)_i \otimes \mathbb{I}^{\otimes [n]\setminus \{i\}}_{q^{n'}}\bigg\}_{i\in  [n], j\in [r_{in}]} \subset \mathcal{P}_q^{n\cdot n'}
    \end{equation}
\end{fact}

where the second set above corresponds to applying the $j$th stabilizer of $S_{in}$ to the $i$th inner block of the concatenated code.

CSS codes have a particular nicely behaved form of concatenation. Recall that concatenation requires specifying an encoding function for the inner code; we now describe how this can be done while preserving the CSS structure. Let $Q^{\text{out}}=\text{CSS}(C_1^{\text{out}},C_2^{\text{out}})$ and $Q^{\text{in}}=\text{CSS}(C_1^{\text{in}},C_2^{\text{in}})$ be $[[n_{\text{out}},k_{\text{out}}]]_{q_{\text{out}}}$ and $[[n_{\text{in}},k_{\text{in}}]]_{q_{\text{out}}}$ CSS codes respectively. Assume that $q_{\text{out}}=q_{\text{in}}^{k_{\text{in}}}$, and that all four classical codes $C_1^{\text{out}},C_2^{\text{out}},C_1^{\text{in}},C_2^{\text{in}}$ are $\mathbb{F}_q$-linear for some $q$. Letting $q_{\text{in}}=q^m$, so that $q_{\text{out}}=q^{mk_{\text{in}}}$, then observe that $(\mathbb{F}_q^{mk_{\text{in}}},\mathbb{F}_q^{mk_{\text{in}}})$ and $(C_1^{\text{in}}/{C_2^{\text{in}}}^\perp,C_2^{\text{in}}/{C_1^{\text{in}}}^\perp)$ are both dual pairs of $mk_{\text{in}}$-dimensional vector spaces over $\mathbb{F}_q$. Thus there exists a pair of $\mathbb{F}_q$-linear isomorphisms 
\begin{equation*}
    (\Enc_1^{\text{in}}:\mathbb{F}_q^{mk_{\text{in}}}\rightarrow C_1^{\text{in}}/{C_2^{\text{in}}}^\perp,\; \Enc_2^{\text{in}}:\mathbb{F}_q^{mk_{\text{in}}}\rightarrow C_2^{\text{in}}/{C_1^{\text{in}}}^\perp)
\end{equation*} that are duality-preserving in the sense that $\langle x,y\rangle=\langle\Enc_1^{\text{in}}(x),\Enc_2^{\text{in}}(y)\rangle$ for all $x,y$.

\begin{claim}
\label{claim:concatcss}
Let $Q^{\text{out}}=\text{CSS}(C_1^{\text{out}},C_2^{\text{out}})$ and $Q^{\text{in}}=\text{CSS}(C_1^{\text{in}},C_2^{\text{in}})$ be as above. Then there is a well defined $[[n_{\text{out}}n_{\text{in}},k_{\text{out}}k_{\text{in}}]]_{q_{\text{in}}}$ CSS concatenation given by
\begin{equation*}
    Q^{\text{out}}\diamond Q^{\text{in}} = \text{CSS}(C_1^{\text{out}}\diamond (C_1^{\text{in}}/{C_2^{\text{in}}}^\perp),\; C_2^{\text{out}}\diamond(C_2^{\text{in}}/{C_1^{\text{in}}}^\perp)),
\end{equation*}
where $C_1^{\text{out}}\diamond (C_1^{\text{in}}/{C_2^{\text{in}}}^\perp)$ is the union of all the cosets in $(\Enc_1^{\text{in}})^{\oplus n_{\text{out}}}(C_1^{\text{out}})\subseteq(C_1^{\text{in}})^{\oplus n_{\text{out}}}/({C_2^{\text{in}}}^\perp)^{\oplus n_{\text{out}}}$, and $C_2^{\text{out}}\diamond(C_2^{\text{in}}/{C_1^{\text{in}}}^\perp)$ is the union of all the cosets in $(\Enc_2^{\text{in}})^{\oplus n_{\text{out}}}(C_2^{\text{out}})\subseteq(C_2^{\text{in}})^{\oplus n_{\text{out}}}/({C_1^{\text{in}}}^\perp)^{\oplus n_{\text{out}}}$.
\end{claim}
\begin{proof}
Let $C_1=C_1^{\text{out}}\diamond (C_1^{\text{in}}/{C_2^{\text{in}}}^\perp)$ and $C_2=C_2^{\text{out}}\diamond(C_2^{\text{in}}/{C_1^{\text{in}}}^\perp)$. It suffices to show that
\begin{align*}
    C_1^\perp &= {C_1^{\text{out}}}^\perp\diamond(C_2^{\text{in}}/{C_1^{\text{in}}}^\perp) \\
    C_2^\perp &= {C_2^{\text{out}}}^\perp\diamond(C_1^{\text{in}}/{C_2^{\text{in}}}^\perp).
\end{align*}
We will show the first equality above; the proof of the second is analogous. By definition, if $c\in {C_1^{\text{out}}}^\perp\diamond(C_2^{\text{in}}/{C_1^{\text{in}}}^\perp)$, then for all $c'\in C_1$, 
\begin{equation*}
    \langle c',c\rangle = \langle ((\Enc_1^{\text{in}})^{-1})^{\oplus n_{\text{out}}}(c'), ((\Enc_2^{\text{in}})^{-1})^{\oplus n_{\text{out}}}(c)\rangle = 0,
\end{equation*}
so $c\in C_1^\perp$. For the converse, assume that $c\in C_1^\perp$. Because $C_1\supseteq({C_2^{\text{in}}}^\perp)^{\oplus n_{\text{out}}}$, we must have $C_1^\perp\subseteq({C_2^{\text{in}}})^{\oplus n_{\text{out}}}$. Thus $((\Enc_2^{\text{in}})^{-1})^{\oplus n_{\text{out}}}(c)$ is well defined, so for all $c'\in C_1$,
\begin{equation*}
    \langle ((\Enc_1^{\text{in}})^{-1})^{\oplus n_{\text{out}}}(c'), ((\Enc_2^{\text{in}})^{-1})^{\oplus n_{\text{out}}}(c)\rangle = \langle c',c\rangle = 0.
\end{equation*}
Therefore $c$ belongs to the (coset specified by the) encoding under $(\Enc_2^{\text{in}})^{\oplus n_{\text{out}}}$ of a vector that is orthogonal to all elements of $C_1$, that is, $c\in (\Enc_2^{\text{in}})^{\oplus n_{\text{out}}}({C_1^{\text{out}}}^\perp)={C_1^{\text{out}}}^\perp\diamond(C_2^{\text{in}}/{C_1^{\text{in}}}^\perp)$.
\end{proof}


\section{Quantum List Decodable Codes}
\label{section:QLD}

Let $Q$ be a stabilizer code of distance $d$, defined with respect to a basis of operators on single-qudits $\mathcal{P}_q$. Consider any adversarial quantum channel $\mathcal{A}$ acting on a (unknown) set of $\tau\cdot n$ qudits of a code-state $\ket{\psi}$ of $Q$. If the support of the attack has size $\tau\cdot n < d$, then, by measuring the syndrome of $Q$, one collapses $\mathcal{A}(\psi)$ into a mixture of single errors in the basis:

\begin{equation}
     \mathcal{A}(\psi) \rightarrow \sum_s \ket{s}\bra{s}\otimes \Pi_s \mathcal{A}(\psi) \Pi_s = \sum_{s} |c_s|^2 \ket{s}\bra{s}\otimes  \sigma_s \ket{\psi}\bra{\psi} \sigma_s^\dagger
\end{equation}

Informally, this occurs since if any two distinct $n$-qudit operators $\sigma, \sigma'\in \mathcal{P}_q^n$ had the same small support $<d$ and the same syndrome, then their product $\sigma^\dagger \sigma'$ must be a stabilizer of $Q$. But if $\sigma^\dagger \sigma'$ is a stabilizer, then $\sigma, \sigma'$ corrupt code-states in essentially the same way: For any $\ket{\psi}\in Q$, $\sigma\ket{\psi} =\eta \sigma'\ket{\psi}$ for some phase $\eta\in \mathbb{C}$.

In this fashion, to decode $\mathcal{A}(\psi)$ it suffices to identify $\sigma_s$ from the syndrome measurement $s$. When the support of $\mathcal{A}$ is $\leq d/2$, then there in fact exists a single $\sigma$ (up to a stabilizer of $Q$) which matches the syndrome, and one can unique-decode. When the support of $\mathcal{A}$ is larger, all bets are off: In this sense, the notion of list-decoding for stabilizer codes we consider is simply the number of operators which match the measured syndrome $s$ - and are all distinct up to a stabilizer of $Q$.
\begin{definition}
    If $\mathcal{S}\subset \mathcal{P}_q^n$ is a stabilizer group, then any pair of operators $O, O'\in  \mathcal{P}_q^n$ is said to be $S$-\textit{stabilizer equivalent} if $O = O' S$ for some $S\in \mathcal{S}$, and \textit{stabilizer-distinct} otherwise.  
\end{definition}

In other words, $O$ is in the coset $O'\mathcal{S}$. When the set of stabilizers is implicit, we simply refer to the pair of operators as `stabilizer-equivalent' or `stabilizer-distinct'. We note that this relation defines equivalence classes over the Pauli group $\mathcal{P}_q^n$, since if $A, B$ and $B, C$ are both stabilizer equivalent, then so are $A$ and $C$. More importantly, if $Q$ is the quantum code stabilized by $\mathcal{S}$, then stabilizer-equivalent operators acting on the code-space give rise to the same state:
\begin{equation}
    O\ket{\psi} = O' S\ket{\psi} = \eta_S O' \ket{\psi}, \text{ for all }\ket{\psi}\in Q,
\end{equation}

\noindent where $\eta_S\in \mathbb{C}$ is some global phase. Moreover, given a description of the stabilizer generators and $O, O'$, one can test stabilizer-equivalence efficiently via row reduction. 

Equipped with this notion, we discuss the following definition of list decoding for stabilizer codes. Conceptually, after measuring the syndrome of the quantum code, we consider enumerating all the operators $E_1, E_2\cdots $ in a given class of errors $\mathcal{E}$ which agree with the syndrome measurement. However, note that many of these operators could be stabilizer-equivalent, corresponding essentially to the same state. The definition we study filters out this degeneracy, by considering the size of any list $\mathcal{L}_s$ of stabilizer-distinct operators which agrees with $s$.

\begin{definition}  \label{def:QLD}
    Let $Q$ be an $[[n,k,d]]_q$ stabilizer QECC with stabilizer group $\mathcal{S}$, and let $\mathcal{E}\subset \mathcal{P}_q^n$ be a set of errors. $Q$ is said to be $\ell$-QLD (Quantum List Decodable) for $\mathcal{E}$ if, for every error $E \in \mathcal{E}$, there exists at most $\ell$ stabilizer-distinct operators which agree with the syndrome $s$ of $E$.
\end{definition}

Of particular attention is the class of qudit Pauli errors $\mathcal{E}= \mathcal{P}_{q, \tau}^n$ of bounded weight $\tau\cdot n$, and we refer to a quantum code as $(\tau,\ell)$-QLD if it is an $\ell$-QLD for $\mathcal{P}_{q, \tau}^n$.

The immediate reason that this particular definition of list-decoding for quantum codes is useful is that the list $\mathcal{L}_s$ can be viewed as a list of ``correction operators.'' For instance, if $\ket{\psi}\in Q$ is a code state of an $\ell$-QLD $Q$, and an adversary corrupts it with an error $E\in \mathcal{E}$ of syndrome $s$, let the list of operators $F_1, \cdots,  F_\ell \subset \mathcal{E}$ be those guaranteed by \cref{def:QLD}. We note that by applying $F_1^\dagger$ on the corrupted code state $E\ket{\psi}$, we return it to the code space $Q$: $F_1^\dagger E \ket{\psi} \in Q$. In fact, we can produce a next candidate code state by applying $(F_2^\dagger F_1) F_1^\dagger E \ket{\psi} = F_2^\dagger E \ket{\psi}\in Q$, and so on. Naturally, we won't know which code state we are in, but by definition at least one correction $F_i^\dagger E \ket{\psi}\propto \ket{\psi}$ recovers the original state. 

This description of list-decoding in terms of the syndrome measurements can alternatively be characterized through the normalizers and logical operators on the stabilizer code.

\begin{claim}
    If $Q$ is a stabilizer code with stabilizer group $S(Q)$ and normalizer group $N(Q)$, then for any $E\in \mathcal{E}$ of syndrome $s$, any list $\mathcal{L}_s$ of operators in $\mathcal{E}$ which agrees with $s$ satisfies
    \begin{equation}
        \mathcal{L}_s \subset \big\{ E N: N\in N(Q)\} \cap \mathcal{E}
    \end{equation}
\end{claim}

\begin{proof}
    If two operators $A, E$ have the same syndrome, then $A^\dagger E$ has null syndrome vector and thus $A^\dagger E\in N(Q)$. In this manner, $A$ can be written as the product of $E$ and an element of the normalizer.
\end{proof}

While the characterization above is a structural constraint on the stabilizers and normalizers of the quantum error correcting code, perhaps its main advantage is that one can cast the list-decoding \textit{task} as an entirely classical process, without any modification to the encoding of the quantum code. As we will see shortly, this description allows us to draw a connection between CSS codes which are QLD and classical codes which are (classically) LD. Beforehand, we need to describe what it means to \textit{efficiently} list-decode these stabilizer codes.

We refer to an $[[n, k, d]]_q$ $\ell$-QLD code $Q$ as \textit{efficiently} list decodable for $\mathcal{E}$ if, given the syndrome $s$, one can produce in time polynomial in $n, \log q, l$ a classical description of a list $\mathcal{L}_s\subset \mathcal{P}_q^n$ of $\leq \ell$ operators with the following stipulations: $\mathcal{L}_s$ contains only stabilizer distinct operators with syndrome $s$; and moreover, for each element of $\mathcal{E}$ of syndrome $s$, there is an operator in $\mathcal{L}_s$ stabilizer-equivalent to it. That is, the list-decoding algorithm is allowed to output extraneous list elements that are not stabilizer-equivalent to any error in $\mathcal{E}$, as long as they still act as corrections, and the overall list size is $\leq \ell$. We emphasize that these ``unpruned'' lists arise because it can be intractable to determine if a given Pauli operator is stabilizer-equivalent to an element of $\mathcal{E}$. 

The main contribution of this section is a recipe to construct quantum list-decodable codes from classical list-decodable codes. For conciseness and to fit our applications, we present a discussion for CSS codes based on the Galois-qudit construction, which we review in \Cref{theorem:galoiscss}.

\begin{definition}
\label{def:cosetld}
    Given an $\mathbb{F}_q$-linear classical code $C\subseteq(\mathbb{F}_q^m)^n$ and a subcode $C'\subseteq C$, let the quotient $C/C'$ be $(\tau,L)$ list-decodable if for every $x\in(\mathbb{F}_q^m)^n$, there exist at most $L$ cosets in $C/C'$ that intersect the Hamming ball $B_{\tau n}(x)$ of radius $\tau n$ centered at $x$.
    
    We say that $C/C'$ is efficiently $(\tau,L)$ list-decodable if there exists a $\poly(nm\log q)$ time algorithm that outputs a list consisting of a representative of each coset in $C/C'$ that intersects $B_{\tau n}(x)$.
\end{definition}

By definition if $C$ is (efficiently) $(\tau,L)$ list-decodable then so is $C/C'$, as we may remove redundant list elements differing by an element of $C'$ using row reduction.

\begin{theorem} \label{thm:CSS}
    Let $Q=\text{CSS}(C_1,C_2)$ be a CSS codes. If $C_1/C_2^\perp$ and $C_2/C_1^\perp$ are both (efficiently) $(\tau,L)$ list-decodable, then $Q$ is (efficiently) $(\tau,L^2)$ list-decodable.
\end{theorem}
\begin{proof}
    Fix an error $E_{\textbf{a}, \textbf{b}}$ on the CSS code, of syndrome $s = (s_x, s_z)$. The normalizers of the CSS code are the operators $E_{\textbf{c}_1, \textbf{c}_2}$, $c_1\in C_1, c_2\in  C_2$. Recall from the definition in \Cref{theorem:vectorspacecss} that we implicitly view $x\in (\mathbb{F}_q^m)^n$ as an element $\textbf{x}\in \mathbb{F}_q^{nm}$. From the previous claim, any list $\mathcal{L}_s$ is contained in the set of operators $F = E_{\textbf{a'}, \textbf{b'}}$ where $a-a'\in C_1$ and $b-b'\in C_2$. Furthermore, $E_{\textbf{a}',\textbf{b}'}$ and $E_{\textbf{a}'',\textbf{b}''}$ are stabilizer equivalent if and only if $a'-a''\in C_2^\perp$ and $b'-b''\in C_1^\perp$. Because $C_1/C_2^\perp$ is $(\tau, L)$-LD, the number of cosets $a'+C_2^\perp$ with $|a'|\leq \tau \cdot n$ and $a-a'\in C_1$ is $\leq L$. Similarly, because $C_2/C_1^\perp$ is $(\tau, L)$-LD, the number of cosets $b'+C_1^\perp$ with $|b'|\leq \tau \cdot n$ and $b-b'\in C_1$ is $\leq L$. It follows that the total number of pairs $(a', b')$ in the list is at most $\leq L^2$. 

    It remains to be shown that efficiency of the classical list-decoding for $C_1/C_2^\perp$ and $C_2/C_1^\perp$ implies efficiency for the quantum list-decoding of $Q$. Let $E' = E_{e_x, e_z}$ be any Pauli operator with the syndrome $s = (s_x, s_z)$, so that $s_x$ and $s_z$ are the parity check syndromes of $e_x$ and $e_z$ for $C_1$ and $C_2$ respectively. One can find such a pair $e_x, e_z\in (\mathbb{F}_q^m)^n$ by finding any solution to the underlying linear system in polynomial time via Gaussian elimination (note that here, we don't restrict $e_x, e_z$ to be low weight). We call the classical list decoding algorithm for $C_1$ on $e_x$, and for $C_2$ on $e_z$, obtaining lists $a_1\cdots a_L$ and $b_1\cdots b_L\in (\mathbb{F}_q^m)^n$ of classical errors of weight $\leq \tau n$. We claim that any stabilizer distinct operator of weight $\leq\tau n$ that matches $s$ is stabilizer equivalent to an operator $E_{a_i, b_j}$, $i, j\in [L]$, output by this algorithm. Indeed, for every operator $E_{a', b'}$ of weight $\leq \tau n$ and of syndrome $s$, then $a'$ and $b'$ both have weight $\leq\tau n$ and have parity check syndromes $s_x$ and $s_z$ for $C_1$ and $C_2$ respectively, so by definition the list-decoding algorithm must output an operator that is stabilizer-equivalent to $E_{a',b'}$.
\end{proof}

\subsection{Explicit Constructions and Folded Quantum Codes}
\label{sec:rsfolding}

In this subsection we discuss \textit{explicit} constructions of quantum list-decodable codes from classical list decodable codes, using \Cref{thm:CSS}. We use the well known classical Reed-Solomon codes, and the Folded Reed-Solomon codes from the breakthrough result of \cite{Guruswami2008ExplicitCA}, to construct quantum codes which are efficiently list decodable up to the quantum Singleton bound (\Cref{theorem:fqrs}). While in principle we could approach the quantum Singleton bound using even random CSS codes, unfortunately these wouldn't be explicit, nor efficiently decodable. Thus for conciseness, we dedicate \Cref{section:randomCSSproperties} to a collection of their properties which we later use in our concatenated constructions. Before diving into the explicit constructions, let us briefly overview the necessary background on these classical codes and related quantum codes. 

\begin{definition} [Generalized Reed-Solomon Codes\label{def:GRS}] Let $q$ be a prime power. Fix $k<n<q$, a primitive element $\gamma$ of $\mathbb{F}^*_q$, and a set of `multipliers' $U = (u_0\cdots u_{n-1})\subset \mathbb{F}^*_q$. The degree $k-1$ GRS code is the set

\begin{equation}
    GRS_{n, k, q}(\gamma, U) = \bigg\{(u_0 f(1),\cdots, u_{n-1}  f(\gamma^{n-1})) : f \in \mathbb{F}_q[X]/(x^k-1)\bigg\}
\end{equation}

\end{definition}

Each codeword of a Reed-Solomon code corresponds to the evaluation of a distinct polynomial over $\mathbb{F}_q$. The codewords of the \textit{Generalized} RS code defined above, are RS codewords multiplied element-wise by scalars. The following properties don't depend on the choice of primitive $\gamma$ or subset $U$ of $\mathbb{F}^*_q$:

\begin{fact}
    $GRS_{n, k, q}(\gamma, U)$ is a $[n,k, n-k+1]_q$ code. 
\end{fact}

\begin{fact} \label{fact:grsdual}
    The dual code $(GRS_{n, k, q}(\gamma, U))^\perp = GRS_{n, n-k, q}(\gamma, V)$, for certain choice of multipliers $V$. When $u_i=1, \forall i\in [n]$, $V = \{1,\gamma, \cdots, \gamma^{n-1}\}$.
\end{fact}

\begin{lemma} [\cite{Guruswami1998ImprovedDO}]
    $GRS_{n, R\cdot n, q}(\gamma, U)$ is efficiently $(1-\sqrt{R}-\varepsilon, n^{O(1)})$-LD.
\end{lemma}

The Quantum GRS codes, or Galois-qudit RS codes \cite{Aharonov2008FaultTolerantQC, Jin2014ACO}, are quantum codes based on the Galois-Qudit CSS construction of the classical codes above. They are a class of quantum codes which meets the quantum Singleton bound, but at the cost of a large local qudit dimension:

\begin{theorem}[Quantum GRS codes \label{def:QRS} \cite{Li2008QuantumGR}]
    Fix a constant $R\in (0, 1)$. Let $C_2^\perp \subset C_1\subset  \mathbb{F}_q^n$ be two GRS codes of block length $n<q$ and dimension $k_2 = n(1-R)/2, k_1 = n(1+R)/2$ respectively, with the same evaluation points $S = \{1, \gamma, \cdots, \gamma^{n-1}\} \subset \mathbb{F}_q$. The Quantum GRS code $Q_{RS} =$ CSS$(C_1, C_2)$ is a $[[n, R n, (1-R)/2 \cdot n]]_q$ QECC.
\end{theorem}

    One can interpret the QGRS encoding as a superposition of polynomial evaluations, where the high degree coefficients are defined by the message. By combining our insights in the previous section with the result by \cite{Guruswami1998ImprovedDO}, we obtain a quantum adaptation of a well known family of classically list decodable codes:

\begin{corollary}
    The Quantum GRS codes of \Cref{def:QRS} of rate $R$ and block-length $n$ are $(1-\sqrt{\frac{1+R}{2}}, n^{O(1)})$-QLD. For small $R$, this approaches radius $\approx.29-1.4\cdot R$.
\end{corollary}

\begin{proof}
    Note that the classical codes $C_1, C_2$ in \Cref{def:QRS} both have rate $R' = (1+R)/2$, and are efficiently LD up to a radius $1-\sqrt{R'}$ \cite{Guruswami1998ImprovedDO}. \Cref{thm:CSS} concludes the proof. 
\end{proof}

In their seminal work Guruswami and Rudra \cite{Guruswami2008ExplicitCA} introduced Folded Reed-Solomon codes, which are efficiently list decodable up to almost the Singleton bound. Informally, they are defined by `bundling' or `folding' together groups of $m$ consecutive symbols of the RS code, into a symbol of a code on a larger alphabet. 

\begin{definition} [$m$-folded codes]
    If $C\subset \mathbb{F}_q^n$ is an $[n, k]_q$ linear code, then the $m$-folded code $C^{(m)} \subset (\mathbb{F}_q^m)^{n/m}$ is an $[n/m, k/m]_{q^m}$ $\mathbb{F}_q$-linear code defined by the codewords
    \begin{equation}
        C^{(m)} = \bigg\{ (c_1, \cdots, c_m), \cdots, (c_{n-m+1}, \cdots, c_n )): c = (c_1,\cdots, c_n)\in C\bigg\}
    \end{equation}
\end{definition}

Indeed, note that the code itself hasn't truly changed, we are simply viewing it as a code over a larger alphabet. Moreover, observe $(C^{(m)})^\perp = (C^\perp)^{(m)}$.

\begin{definition} 
    [$m$-folded Reed-Solomon codes \label{def:FRS} \cite{Guruswami2008ExplicitCA}] Fix a primitive element $\gamma$ of $\mathbb{F}^*_q$, integers $n, m, k$ where $k<n\leq q-1$ and $m$ divides $n$. The $m$-folded RS code $FRS_{n,k, q}^m$ is an $[n/m, k/m, d/m]_{q^m}$ $\mathbb{F}_q$-linear code defined over alphabet $\mathbb{F}_{q}^m$ which encodes a polynomial $f\in F_q[x]$ of degree $k-1$ as
    \begin{equation}
        \begin{pmatrix}
        \begin{bmatrix}
        f(1) \\
        f(\gamma) \\
        \vdots \\
        f(\gamma^{m-1})
        \end{bmatrix},  \begin{bmatrix}
        f(\gamma^{m}) \\
        f(\gamma ^{m+1}) \\
        \vdots \\
        f(\gamma^{2m-1})
        \end{bmatrix}, \cdots, 
        \begin{bmatrix}
        f(\gamma^{n-m}) \\
        f(\gamma ^{n-m+1}) \\
        \vdots \\
        f(\gamma^{n-1})
        \end{bmatrix}
        \end{pmatrix}
    \end{equation}
\end{definition}

\begin{theorem}[\cite{Guruswami2008ExplicitCA}\label{theorem:FRS}]
    For any $q > n$, there exists a choice of $m = O(1/\gamma^2)$ such that the $m$-folded RS code of rate $R$ is efficiently $(1-R-\gamma, q^{O(1/\gamma)})$-LD.
\end{theorem}

By multiplying element-wise the RS codes above by scalars $u_i\in \mathbb{F}_q^*, i\in [n]$, one defines Folded Generalized RS codes. Our starting point for efficiently list decodable quantum codes is an analogous bundling for the physical qudits of any stabilizer code. 

\begin{definition}
    Let $S_1, \cdots, S_r\in \mathcal{P}_q^n$ define a set of generators for an $[[n, k]]_q$ stabilizer code $Q\subset \mathbb{C}_q^{\otimes n}$. Then the $m$-folded quantum code $Q^{(m)}\subset (\mathbb{C}_q^{m})^{\otimes n/m}$ is the $[[n/m, k/m]]_{q^m}$ stabilizer code defined by the stabilizer generators $S^{(m)}_1\cdots S^{(m)}_r$, where
    \begin{equation}
        S^{(m)}_i =  \bigotimes_{j\in [n/m]} \big( \otimes_{l\in [m]} S_{i, j\cdot m+ l}\big) 
    \end{equation}
    where each generator $S_i = \otimes_{j\in [n]} S_{ij}$ is expressed as a tensor product of operators acting on $\mathbb{C}_q$. 
\end{definition}

Analogously to the classical construction, we are folding the physical qudits together, and we view the stabilizers as operators acting on a larger local dimension. We emphasize however that they are the same stabilizers, and the syndromes of operators acting on the code are measured in the same way as before. The key distinction is that now the notion of the `weight' of an operator $E$ acting on $Q^{(m)}$, is the number of distinct `bundles' that $E$ acts non-trivially on. A simple claim allows us to use the list decoding algorithms of folded classical codes on folded quantum codes, when $Q$ is a CSS code.

\begin{claim}\label{claim:foldingcsscodes}
    If $Q$ is the Galois-qudit CSS code of two linear classical codes $C_1, C_2\subset \mathbb{F}_q^n$, then $Q^{(m)}$ is the vector-space CSS code of the two $\mathbb{F}_q$-linear folded classical codes $C_1^{(m)}, C_2^{(m)}$. 
\end{claim}

See the definitions in \Cref{theorem:galoiscss} and \Cref{theorem:vectorspacecss} for the distinctions between these CSS constructions. 

\begin{proof}
    As defined in \Cref{theorem:vectorspacecss}, the vector space CSS code of $C_1^{(m)}, C_2^{(m)}$ is the subspace of $(\mathbb{C}_q^{\otimes m})^{\otimes n/m}$ stabilized by the operators $E_{\textbf{a}, \textbf{b}}$, where $a \in (C_1^{(m)})^\perp, b\in (C_2^{(m)})^\perp \subset (\mathbb{F}_{q}^m)^{n/m}$ are viewed as elements $\textbf{a}, \textbf{b}$ of $\mathbb{F}_q^n$. Since $a \in (C_1^{(m)})^\perp$ as an element of  $\mathbb{F}_q^n$ is the `unfolded' element $\textbf{a}\in C_1^\perp$, we arrive at precisely the definition of $Q^{(m)}$.
\end{proof}

Equipped with this notion of folding, we can now consider folding the Quantum GRS codes of \Cref{def:QRS}.

\begin{definition} [Folded Quantum Reed-Solomon Code\label{def:FQRS}]
    For $0<R<1$, fix a primitive element $\gamma\in \mathbb{F}_q^*$ and integer block length $n<q-1$, and define integers dimension $d  = n-(1+R) n/2$ and $m$ with $m|n$. The $m$-folded Quantum RS code is an $[[n/m, R \cdot n/m, d/m]]_{q^m}$ quantum error correcting code defined over qudits of local dimension $q^m$, which encodes a computational basis state $|f\rangle$, $f\in \mathbb{F}_q^{R n}$ through the superposition $Enc(|f\rangle) = |\bar{f}\rangle$:

    \begin{gather}
        \frac{1 }{q^{d/2}} \sum_{deg(f')\leq d}
        \bigg|
        \begin{bmatrix}
        (f+f')(1) \\
        (f+f')(\gamma) \\
        \vdots \\
        (f+f')(\gamma^{m-1})
        \end{bmatrix} \bigg\rangle \otimes \cdots \otimes \bigg|
        \begin{bmatrix}
        (f+f')(\gamma^{n-m}) \\
        (f+f')(\gamma ^{n-m+1}) \\
        \vdots \\
        (f+f')(\gamma^{n-1})
        \end{bmatrix} \bigg\rangle
    \end{gather}
\end{definition}

Where in the above $(f+f')(x)$ corresponds to the evaluation at $x\in \mathbb{F}_q$ of the polynomial defined by $f'$ as the low degree coefficients $< d$, and $f$ as the high degree coefficients within $d$ and $d+R\cdot n$. The bracket notation corresponds to the tensor product state in $\mathbb{C}_q^{\otimes m}$

\begin{equation}
    \bigg| \begin{bmatrix}
        (f+f')(1) \\
        (f+f')(\gamma) \\
        \vdots \\
        (f+f')(\gamma^{m-1})
        \end{bmatrix}  \bigg\rangle  = \otimes_i^{m-1} \ket{(f+f')(\gamma^i)}
\end{equation}

The main theorem of this section is that these $m$-folded QRS codes are QLD up to essentially the distance of the quantum code. 

\begin{theorem} \label{theorem:fqrs}
    The FQRS of \Cref{def:FQRS} is $(\frac{1-R}{2}-\gamma,  q^{O(1/\gamma)})$-QLD for a choice of $m= O(1/\gamma^2)$.
\end{theorem}

\begin{proof}
    From \cref{claim:foldingcsscodes}, we understand the FQRS codes as the vector space CSS codes (\cref{theorem:vectorspacecss}) of two classical folded GRS codes. These are classically $(\frac{1-R}{2}-\gamma,  q^{O(1/\gamma)})$-LD from \cref{theorem:FRS}, \cite{Guruswami2008ExplicitCA}. Finally, since the folded GRS codes are $\mathbb{F}_q$-linear, from \Cref{thm:CSS} we conclude the FQRS codes are $(\frac{1-R}{2}-\gamma,  q^{O(1/\gamma)})$-QLD as well.
\end{proof}

\section{Approximate Quantum Error Correction}
\label{sec:AQECC}
In this section we construct approximate quantum error correcting codes (AQECCs) that approach the quantum Singleton bound. In \Cref{subsec:private} we construct private AQECCs --- which use a private classical side-channel --- from quantum list decodable codes and purity testing codes. In \Cref{subsec:RSS} we show how to use a classical robust secret sharing scheme to remove the need for the classical side-channel, at the cost of slightly increasing the alphabet size.

\subsection{Private AQEC from Quantum List Decoding and Purity Testing}
\label{subsec:private}

Recall our definition of a private AQECC, where the sender and receiver have access to a private classical side-channel.

\begin{definition} [\Cref{def:privateaqecc}, restatement]
    A set of pairs $\{(\Enc_k, \Dec_k)\}_{k\in K}$ of quantum channels is said to be a $(\delta, \varepsilon)$-private AQECC with keys $K$ if for all adversarial channels $\mathcal{A}$ of weight $\delta\cdot n$,
    $$\norm{\mathbb{E}_{k \from K}[\Dec_k\circ \mathcal{A} \circ \Enc_k] - \Id}_\diamond \le \varepsilon.$$
\end{definition}
We now describe a construction of private AQECCs from list decodable quantum codes (QLDs) and purity testing codes (PTCs). Recall that a $(\delta,L)$-QLD stabilizer code produces a list of possible corrections of size $L$, given any codestate that has been adversarially corrupted on up to a $\delta$ fraction of symbols (see \Cref{def:QLD}). Roughly, an $\varepsilon$-PTC detects any fixed Pauli error with probability $1-\varepsilon$ (see \Cref{def:PTC}).

Given an $[[n, n_{PTC}]]_q$ QLD $Q_{LD}$ and $[[n_{PTC}, m]]_q$ PTC $\{Q_k\}$, our private AQECC $\{Q_{LD} \circ Q_k\}$ is defined by channels $\Enc_k$ that encode a message state $\rho$ into $Q_k$ and then into $Q_{LD}$:
\begin{equation}
    \Enc_k(\rho) = \Enc_{Q_{LD}} \circ \Enc_{Q_k}(\rho).
\end{equation}

To decode the quantum code composition $Q_{LD} \circ Q_k$\footnote{See \Cref{def:codecomposition} and \Cref{fact:concatstabilizers} for the properties of the composition of quanutm codes.}, we use \Cref{alg:algorithm1}.

\begin{algorithm}[h]
    \caption{$\Dec_k$ for $Q_{LD} \circ Q_k$}
    \label{alg:algorithm1}
    \KwInput{A corrupted code-state $\mathcal{A}(\Enc_k(\rho))$}
    Measure the syndromes $s_{QLD}$ of $Q_{LD}$ and $s_{PTC}$ of $Q_k$\;
    Apply list decodability to $s_{QLD}$ to obtain a list $E_1, \dots, E_L$ of potential errors\;
    Compute the $Q_k$-syndromes $s_1, \dots, s_L$ of $E_1, \dots, E_L$, and choose any $i \in [L]$ such that $s_i = s_{PTC}$\;
    Apply $(E_i \circ \Enc_k)^\dagger$ to the corrupted code-state and output the resulting message\;
\end{algorithm}

The main result of this subsection is the following lemma, which states that the composition above forms a private AQECC.

\begin{lemma} \label{lemma:private-aqec}
    Let $\{Q_k\}$ be an $\varepsilon$-PTC, where each $Q_k$ is an $[[n_{PTC}, m]]_q$ stabilizer code, and let $Q_{LD}$ be a $[[n, n_{PTC}, d]]_q$ stabilizer code which is $(\delta, L)$-QLD and has distance $d>\delta\cdot n$. Then $Q_{LD} \circ Q_{k}$ is an ensemble of $[[n, m]]_q$ stabilizer codes, and a $(\delta, 2\cdot L\cdot \varepsilon)$ private AQECC.
\end{lemma}

We instantiate \Cref{lemma:private-aqec} with the family of $[[n, R\cdot n]]_q$, $((1-R-\gamma)/2, n^{O_\gamma(1)})$-QLD codes on alphabets of size $q = 2^{\Theta(1/\gamma^5)}$ guaranteed by \Cref{thm:frsael} or \Cref{theorem:listrecoveryconstruction}, together with a family of particularly randomness-efficient PTC's from the work of \cite{barnum2002authentication}.

\begin{theorem}[\cite{barnum2002authentication}]
\label{thm:ptc}
    For every prime power $q$ and every $\lambda,n_{PTC}\in\mathbb{N}$ with $\lambda|n_{PTC}$, there exists a stabilizer $\varepsilon$-PTC $\{Q_k\}$ on qudits of local dimension $q$ and key size $\lambda \cdot \log q$ bits, encoding $n_{PTC}-\lambda$ qudits into $n_{PTC}$ qudits, where $\varepsilon \leq 2 \cdot n_{PTC} \cdot q^{-\lambda}/\lambda$.
\end{theorem}

With $\lambda =  n/\log q$ and $n_{PTC} = (R+\gamma/2)\cdot n$, we obtain:

\begin{corollary}
\label{cor:private-aqec}
    For any $\gamma > 0$ and $R \in (0,1)$, there exists a family of $((1-R-\gamma)/2, 2^{-\Omega(n)})$ private AQECCs of rate $R$, block length $n$, key size $n$ bits, and local dimension $2^{\Theta(1/\gamma^5)}$. Furthermore, there is an efficient randomized algorithm to construct these codes with failure probability $2^{-\Omega(n)}$, as well as efficient encoding and decoding algorithms.
\end{corollary}

We leverage two key ideas to prove \Cref{lemma:private-aqec}. First, recall the discussion in \Cref{section:QLD}: If $Q_{LD}$ has distance $d$, then the syndrome measurement in step 1 of \Cref{alg:algorithm1} of any adversarial channel $\mathcal{A}(\cdot)$ of support at most $\delta\cdot n < d$ collapses the corruption on the code state down to a single Pauli error $E$. Moreover, the local indistinguishability of the stabilizer code guarantees that an adversary can't learn any information about the secret key $k$, and thus $E$ is independent of the random choice of $k$. We formalize this idea in \Cref{fact:syndromelocind}, and for conciseness, defer to \Cref{claim:LI-composable-draft} in \Cref{section:proof-claim-LI-composable} a proof of a more general version even in the presence of adversaries with entangled side information.

Next, the list decodability of $Q_{LD}$ enables us to write down a short list $F_1, \cdots, F_L$ of errors associated to the syndrome measurement $s_{QLD}$ of $E$. We argue in \Cref{claim:ptcsyndrome-draft} that if any candidate error $F_i$ has the same PTC syndrome as the syndrome measurement $s_{PTC}$, then with high probability it exactly corrects the error $E$. Finally, in \Cref{claim:lemma41-proof} we tie these ideas together and conclude the proof of \Cref{lemma:private-aqec}. 

\begin{fact}
\label{fact:syndromelocind}
    Let $Q$ be be a $[[n,m,d]]_q$ stabilizer code. Let $\mathcal{A}$ be an arbitrary adversarial error channel acting on some subset $S\subseteq[n]$ of $|S|<d$ code components. For an arbitrary corrupted code state $\rho$, let $s_\rho$ denote the random variable for the syndrome of $\rho$. Then the distribution of $s_{\mathcal{A}(\psi)}$ is the same for all $\ket{\psi}\in Q$.
\end{fact}

In this manner, if we pre-encode a message $\psi$ with an `inner' PTC code $Q_k$ (for a random choice of key $k$) before encoding it into a stabilizer code $Q_{LD}$, then the syndrome measurement $s_{QLD}$ must be independent of $k$.

\Cref{claim:ptcsyndrome-draft} below reasons that if an operator $E$ is an undetectable error of $Q_{LD}$ (i.e., $E\in N(Q_{LD}) \setminus S(Q_{LD})$), then with high probability over the choice of $k$, $E$ is not an undetectable error of the composition $Q_{LD}\circ Q_k$. Later, we use this claim to show the PTC indeed eliminates all but one list element (up to stabilizer-equivalence) with high probability.

\begin{claim} \label{claim:ptcsyndrome-draft}
    Let $E\in \mathcal{P}_q^n$ satisfy $E\in N(Q_{LD}) \setminus S(Q_{LD})$. If $\{Q_k\}$ forms an $\varepsilon$-PTC, then
    \begin{equation}
        \mathbb{P}_k\big[ E \in N(Q_{LD}\circ Q_k) \setminus S(Q_{LD}\circ Q_k) \big] \leq \varepsilon
    \end{equation}
\end{claim}

\begin{proof}
    Since $E\in N(Q_{LD}) \setminus S(Q_{LD})$, $E$ is equivalent up to a stabilizer (of $Q_{LD}$) to some encoded operator $\bar{O}$ on $Q_{LD}$. Let $O$ be the operator acting on the message space of $Q_{LD}$ corresponding to $\bar{O}$. If we could prove that $\bar{O} \in N( Q_{LD}\circ Q_k) \setminus S(Q_{LD}\circ Q_k)$ iff $O\in N(Q_k)\setminus S(Q_k)$, we would be done, as therefore since $O$ is fixed, 
    \begin{equation}
        \mathbb{P}_k\big[ E \in N( Q_{LD}\circ Q_k) \setminus S(Q_{LD}\circ Q_k) \big]  = \mathbb{P}_k\big[ O \in N(Q_k) \setminus S(Q_k) \big] \leq \varepsilon
    \end{equation}

    by definition of the $\varepsilon$-PTC. To prove our claim, we consider two cases on $O$: If $O\in S(Q_k)$, then $\bar{O}\in S(Q_{LD}\circ Q_k)$ is simply a stabilizer of the composed code. If $O\notin N(Q_k)$, then $O$ doesn't commute with some stabilizer $S$ of $Q_k$. That implies their encodings in the composed code $\bar{O}$, $\bar{S}$ also don't commute. Thus $\bar{O}\notin N(Q_{LD}\circ Q_k)$, and therefore $E \notin N(Q_{LD}\circ Q_k)$. We conclude $E \in N( Q_{LD}\circ Q_k) \setminus S(Q_{LD}\circ Q_k)$ iff $O\in N(Q_k)\setminus S(Q_k)$ as intended.
\end{proof}

Let us now tie these ideas together, and prove that \Cref{alg:algorithm1} approximately decodes the code composition. \Cref{claim:lemma41-proof} below reasons that the density matrix produced at the end of the decoding \Cref{alg:algorithm1} has uniquely decoded the adversarial error with high probability over the random choice of key $k$, and the outcome of the syndrome measurement. 

\begin{claim}\label{claim:lemma41-proof}
    For any adversarial channel $\mathcal{A}$ acting on $|S|\leq \delta\cdot n < d$ code components, and any $Q_{LD}$ syndrome measurement $s=s_{QLD}$, there exists a large subset of PTC keys $\text{Good}_{\mathcal{A},s}\subset K$, $|\text{Good}_{\mathcal{A},s}|\geq |K|(1-L\cdot \varepsilon)$, such that the density matrix describing the outcome of step 4 in \Cref{alg:algorithm1} in expectation over $k\in K$ can be written as
    \begin{equation}
        \rho_4 = \bigg(\sum_{s, k\in \text{Good}_{\mathcal{A},s}} \frac{p_{\mathcal{A}, s}}{|K|} \ket{k}\bra{k}\otimes \ket{s}\bra{s}\bigg)\otimes \rho +  \sum_{s, k\in K\setminus \text{Good}_{\mathcal{A},s}} \frac{p_{\mathcal{A}, s}}{|K|} \ket{k}\bra{k}\otimes \ket{s}\bra{s} \otimes \rho_{k,s}.
    \end{equation}

    where $p_{\mathcal{A}, s}$ denotes the probability of the syndrome measurement $s$, and $\rho_{k, s}$ are arbitrary $m$ qudit density matrices. Thus, \Cref{alg:algorithm1} uniquely recovers the encoded state $\rho$ with probability (and fidelity) $\geq 1-L\cdot \varepsilon$, and moreover $Q_{LD}\circ Q_k$ forms a $(\delta, 2\varepsilon L)$ private AQECC.
\end{claim}

\begin{proof}
    From \Cref{fact:syndromelocind}, we know that the syndrome measurement of $Q_{LD}$ in step 1 of \Cref{alg:algorithm1} collapses the state into a mixture of single Pauli errors

    \begin{equation}
        \rho_1 = \sum_{s_{QLD}, k} \frac{p_{\mathcal{A}, s}}{|K|} \ket{k}\bra{k}\otimes \ket{s_{QLD}}\bra{s_{QLD}} \otimes \sigma_{\mathcal{A}, s_{QLD}} \Enc_k(\rho) \sigma_{\mathcal{A}, s_{QLD}}^\dagger
    \end{equation}

    where each Pauli $\sigma_{\mathcal{A}, s} \in \mathcal{P}_q^n$ is a function only of $s=s_{QLD}$ and the support of $\mathcal{A}$. Note that the syndrome measurement of $s_{PTC}$ does not further collapse the state. 
    
    Given $(s_{QLD}, s_{PTC})$, if $E_1\cdots E_L$ is the list of candidate corrections to $\sigma_{\mathcal{A}, s_{QLD}}$ produced in step 3, then at least some $j\in [L]$ corrects $\sigma_{\mathcal{A}, s_{QLD}}$, i.e.  $E_j^\dagger \sigma_{\mathcal{A}, s_{QLD}}\in S(Q_{LD}\circ Q_k)$. Moreover, the decoding algorithm fails to produce $\rho$ if there exists some other $E_i$ of PTC syndrome $s_{PTC}$ which fails to correct $\sigma_{\mathcal{A}, s_{QLD}}$, i.e.
    \begin{equation}
       \exists i\in [L] \text{ s.t. } E_i^\dagger\sigma_{\mathcal{A}, s_{QLD}} \in N(Q_{LD}\circ Q_k) \setminus S(Q_{LD}\circ Q_k)
    \end{equation}

    From \Cref{claim:ptcsyndrome-draft}, the probability over $k$ this occurs is at most $L\cdot \varepsilon$ by a union bound. Thereby, for every $Q_{LD}$ syndrome measurement $s=s_{QLD}$, there exists a subset of PTC keys $\text{Good}_{\mathcal{A},s}\subset K$, $|\text{Good}_{\mathcal{A},s}|\geq |K|(1-L\cdot \varepsilon)$ where any $E_j$ in the list with PTC syndrome $s_{PTC}$ exactly recovers $\rho$. If $k\in K\setminus \text{Good}_{\mathcal{A},s}$, we have no guarantees on the returned state $\rho_{k, s}$. 

    To conclude the proof, let us trace out the syndrome measurement and choice of key $k$ in $\rho_4$. For any encoded state $\rho$, $\text{Tr}_{s,k}[\rho_4] = \rho \cdot p_{success}  + \rho_{fail}$, where $p_{success} = \sum_s p_{\mathcal{A},s}\text{Good}_{\mathcal{A},s} / |K|\geq (1-L\cdot \varepsilon)$ and $\|\rho_{fail}\|_1 = 1-p_{success} \leq L\varepsilon$. Thus,
    \begin{equation*}
    \norm{\mathbb{E}_k[\Dec_k\circ \mathcal{A} \circ \Enc_k(\rho)] - \rho}_1 \leq 2\cdot L\cdot \varepsilon. \qedhere
    \end{equation*}
\end{proof}

\subsection{AQECC from Private AQEC and Robust Secret Sharing}
\label{subsec:RSS}
Recall the definition of an AQECC.
\begin{definition}
\label{def:aqec}
    A pair $(\Enc, \Dec)$ of quantum channels is said to be a $(\delta, \varepsilon)$-AQECC if, for all adversarial channels $\mathcal{A}$ of weight $\delta \cdot n$,
    $$\norm{\Dec \circ \mathcal{A} \circ \Enc - \mathbb{I}}_\diamond \leq \varepsilon.$$
\end{definition}

The following definition of robust secret sharing is taken from \cite{CDD+15}, except that we do not restrict to linear schemes and we do not include separate parameters for privacy and reconstructability. We also require the algorithms in the scheme to be efficient.
\begin{definition}[Robust secret sharing]
    For integers $n,s,d,$ and $a$, an $[n,s]_a$ $(d,\epsilon)$-robust secret sharing scheme $\RSS$ consists of two efficient algorithms $\RSS.\Share$ and $\RSS.\Reconstruct$. For every $r \in [a]^s$, $\RSS.\Share(r)$ outputs a vector of shares $v \in [a]^n$. We require the following two properties:
    \begin{itemize}
        \item Privacy: For all $r,r' \in [a]^s$ and every $A \subseteq [n]$ of size $\abs{A} \le d$, the restrictions $v_A$ and $v_A'$ of $v = \RSS.\Share(r)$ and $v' = \RSS.\Share(r')$ to the coordinates in $A$ have the same probability distribution.
        \item Reconstructability: For every $r \in [a]^s$ and every $A \subseteq [n]$ of size $\abs{A} \le d$, if $v = \RSS.\Share(r)$ and $\tilde{v}$ is such that $\tilde{v}_{\bar{A}} = v_{\bar{A}}$ and $\tilde{v}$ only depends on $v_A$, then $\RSS.\Reconstruct(\tilde{v}) = r$ except with probability $\epsilon$.
    \end{itemize}
\end{definition}

The following result is from \cite{CDD+15}, with the explicit choice of alphabet size $a$ from \cite{Guruswami2014OptimalRL}. Concretely, we use the fact that $2^{1/\gamma}< a = 2^{\tilde{O}(\gamma^{-2})}$.

\begin{theorem}[\cite{CDD+15}, Theorem 7]
\label{theorem:rss-construction}
For any $\gamma > 0$ and $R \in (0,1)$, there exists an efficient family of $[n, R\cdot n]_{a}$ $((1-R-\gamma) \cdot n/2, 2^{-n})$-robust secret sharing schemes $\{\RSS_n\}$ on shares of alphabet size $a = 2^{\tilde{O}(\gamma^{-2})}$.
\end{theorem}

Given any $[n,s]_a$ $(d,\epsilon)$-robust secret sharing scheme $\RSS$, and a $[[n,k]]_q$ private $(d,\varepsilon')$-approximate quantum code $Q$ defined by $(\Enc_r, \Dec_r)_{r \in [a]^s}$, we build an $(d,\varepsilon+\epsilon')$-approximate $[[n,k]]_{q \cdot a}$ code $\hat{Q}$ by secret sharing the private keys $r$ with $\RSS$ and appending the shares to the symbols of $Q$.

Formally, $\Enc$ for $\hat{Q}$ works as follows:
\begin{enumerate}
    \item Sample $r \from [a]^s$.
    \item Apply $\Enc$ to $\rho$, and call the resulting register $C = (C_1, \dots, C_n)$.
    \item Compute $\RSS.\Share(r)$, and store the (classical) result on register $S = (S_1, \dots, S_n)$.
    \item Use $(C,S)$ as the codestate, where register $i \in [n]$ is the $a \cdot q$-dimensional register $(C_i,S_i)$.
\end{enumerate}

To decode $\hat{Q}$, we recover the private keys using the reconstruction algorithm of $\RSS$ and then apply the decoder for $Q$ to the $C$ register:
\begin{enumerate}
    \item Measure register $S$ and apply $\RSS.\Reconstruct$ to the result, obtaining $\tilde{r}$.
    \item Apply $\Dec_{\tilde{r}}$ to $C$, and output the resulting decoded message.
\end{enumerate}

\begin{theorem}
\label{theorem:removing-privacy}
Given
\begin{enumerate}
    \item any $[n,s]_a$ $(d,\epsilon)$-robust secret sharing scheme $\RSS$, and
    \item any $(d,\varepsilon')$-approximate private quantum $[[n,k]]_q$ code $Q$ with key space $[a]^s$,
\end{enumerate}
the construction of $\hat{Q}$ above is an $(d, \varepsilon+\epsilon')$-approximate $[[n,k \cdot \log(q)/\log(q \cdot a)]]_{q \cdot a}$ quantum code.
\end{theorem}
\begin{proof}
    Let $\rho$ be any message. Partition the registers of $\Enc(\rho)$ into $C_A,C_{\bar{A}},S_A,S_{\bar{A}}$, where $C_A,C_{\bar{A}}$ are the private AQECC part and $S_A,S_{\bar{A}}$ are the RSS part, and the error channel acts only on $C_A,S_A$. Let $E$ denote the error channel. By the secrecy guarantee of the RSS against any subset of $d$ symbols, we have
    \begin{align*}
        & \Tr_S\bigg[E  \left(a^{-s} \sum_{r \in [a]^s} \Enc_r(\rho)_C \otimes \RSS.\Share(r)_S\right)\bigg] \\
        &= \Tr_{S_A}\bigg[  E\bigg(\Tr_{S_{\bar{A}}} a^{-s} \sum_{r \in [a]^s} \Enc_r(\rho)_C \otimes \RSS.\Share(r)_S \bigg) \bigg] \\
        &= \Tr_{S_A}\bigg[E\bigg(\Tr_{S_{\bar{A}}} a^{-2s} \sum_{r \in [a]^s,r' \in [a]^s} \Enc_r(\rho)_C \otimes \RSS.\Share(r')_S\bigg)\bigg] \\
        &= \Tr_S\bigg[E\bigg(a^{-2s} \sum_{r \in [a]^s,r' \in [a]^s} \Enc_r(\rho)_C \otimes \RSS.\Share(r')_S\bigg)\bigg]
    \end{align*}

    and thereby the adversary gains no information about the encryption key. By the recoverability guarantee of the RSS, $\Dec$ will use the correct private keys $r$ with probability $1-\epsilon$, so
    \begin{align*}
        \norm{\rho - \Dec \circ E \circ \Enc(\rho)}_1 &= \norm{\rho - \Dec \circ E \left(\frac{1}{a^s} \sum_{r \in [a]^s} \Enc_r(\rho) \otimes \RSS.\Share(r)\right)}_1 \\
        &\le \norm{\rho - \frac{1}{a^s} \sum_{r \in [a]^s} \Dec_r \circ \Tr_S \circ E (\Enc_r(\rho) \otimes \RSS.\Share(r))}_1 + \epsilon \\
        &= \norm{\rho - \frac{1}{a^s} \sum_{r \in [a]^s} \Dec_r \circ \Tr_S \circ E \left(\Enc_r(\rho) \otimes \frac{1}{a^s} \sum_{r' \in [a]^s} \RSS.\Share(r')\right)}_1 + \epsilon \\
        &\le \varepsilon' + \epsilon. \qedhere
    \end{align*}
\end{proof}

Using the private AQECC of \Cref{cor:private-aqec} and the robust secret sharing scheme of \cite{CDD+15}, we are able to instantiate \Cref{theorem:removing-privacy} to obtain the following result.

\begin{corollary}
\label{cor:aqeccmain}
For any constant $\gamma > 0$ and $R \in (0,1)$, there exists a family $\{Q_n\}$ of $((1-  R - \gamma)/2,2^{-\Omega(n)})$-approximate $[[n, R \cdot n]]_{q'}$ quantum codes on alphabets of size $q' = 2^{O(1/\gamma^5)}$. Furthermore, there is an efficient randomized algorithm to construct these codes with failure probability $2^{-\Omega(n)}$, as well as efficient encoding and decoding algorithms.
\end{corollary}

\begin{proof}
By \Cref{cor:private-aqec}, for any fixed $R', \gamma' > 0$ there exist $[[n,R' \cdot n]]_{q}$ private $((1-R'-\gamma')/2, 2^{-\Omega(n)})$-approximate quantum error correcting codes with $n$-bit keys and alphabet size $q=2^{\Theta(1/\gamma'^5)}$, with the desired efficiency and construction guarantees. By \Cref{theorem:rss-construction}, there exist $[n,r \cdot n]_{a}$ $((1-r-\gamma'') \cdot n/2, 2^{-n})$-robust secret sharing schemes, for $a = 2^{\tilde{O}(1/\gamma''^2)}$. 

So long as the private AQECC key length is smaller than the secret size for the RSS scheme of \Cref{theorem:rss-construction}, it can be secret shared into the RSS. This occurs if $n \leq n\cdot r\log a$, i.e., $r\geq 1/\log a$. From \cite{Guruswami2014OptimalRL}, we have $\log a\geq 1/\gamma''$, and thus we pick $r = \gamma''$. \Cref{theorem:removing-privacy} then yields an AQECC with decoding radius $\frac{1}{2}\min \big\{(1-R'-\gamma'),(1-r-\gamma'')\big\}$, and its rate is
\begin{equation*}
\frac{R' \cdot n \log q}{n \cdot (\log q + \log a)} = \frac{\Theta(1/\gamma'^5)}{\Theta(1/\gamma'^5) + \tilde{O}(1/\gamma''^2)} \cdot R',
\end{equation*}
and it only remains to pick $\gamma'$. By inspection, there is a function $f(\gamma'') = \Tilde{O}(\gamma''^{3/5})$ such that whenever $\gamma' \leq  f(\gamma'')$, the rate above is at least $R'\cdot (1-\gamma'')$. We can now divide into cases on $\gamma'', f(\gamma'')$ to determine the decoding radius. If $\gamma'' < f(\gamma'')$, pick $\gamma' = \gamma''$, and if $\gamma'' \geq f(\gamma'')$ we pick $\gamma' = f(\gamma'')$. In either case, the alphabet size is $q' = q\cdot a = 2^{O(1/\gamma''^5)}$, and the decoding radius is at least $\frac{1}{2}(1-\gamma'' - \max\{R', \gamma''\})$.

To conclude, let us set $\gamma'' = \gamma/4$, and $R' = R/(1-\gamma'')$. We may assume that $R' \le 1$ because $R\leq 1-\gamma = 1-4\gamma''$; otherwise the decoding radius in the theorem statement is $0$. The overall rate is then $R$ as intended. Since $\gamma''\leq 1/4$, we have $R'\leq R(1+2\gamma'')\leq R+\gamma/2$, and thus the resulting decoding radius is at least $(1-R-\gamma)/2$.
\end{proof}

\section{Distance Amplification and Alphabet Reduction for Quantum Codes}
\label{sec:ael}

In this section, we describe how we apply Alon-Edmonds-Luby (AEL) distance amplification \cite{AEL95} to quantum codes. This technique has seen extensive use in classical coding theory \cite{guruswami2001expander,Guruswami2002NearoptimalLC,guruswami2003linear,Guruswami2008ExplicitCA,hemenway2018linear,Kopparty2015HighrateLA,gopi2018locally,hemenway2019local}, as it allows for amplifying the distance of a code while reducing alphabet size and preserving properties such as decodability and local testability and correctability.

To the best of our knowledge, AEL amplification has not previously been used in the quantum setting. However, as shown below, this technique translates almost flawlessly to quantum codes. This observation immediately yields some new results that will be discussed below, such as efficiently decodable quantum codes approaching the quantum Singleton bound over constant sized alphabets. In fact, such codes can be made LDPC using recent asymptotically good quantum LDPC constructions \cite{panteleev2022asymptotically,leverrier2022quantum}, some of which permit efficient decoding \cite{leverrier2022efficient,dinur2022good,gu2022efficient,leverrier2022parallel}.

Both of our constructions of AQECCs approaching the quantum Singleton bound apply the AEL construction. Specifically, we instantiate the AQECC construction described in \Cref{sec:AQECC} using quantum list-decodable codes approaching the quantum Singleton bound with constant alphabet size that we will construct in \Cref{sec:frsael} by applying AEL to reduce the alphabet size of folded quantum Reed-Solomon codes. Then in \Cref{sec:aqeccdirect}, we provide another construction of AQECCs by applying AEL amplification directly with inner codes that are themselves appropriately chosen AQECCs.

At a high level, the AEL construction amplifies the distance of a code by concatenating with a small inner code of good distance, and then permuting the symbols of the concatenated code using an expander graph. If the outer code is chosen to have high rate and constant relative distance, then the resulting construction inherits the rate and distance of the inner code, up to a small loss. This procedure in fact works for arbitrary quantum codes; we do not even need to restrict attention to stabilizer codes below.

We first define the appropriate notion of expansion.

\begin{definition}
\label{def:expander}
  An $r$-regular bipartite graph $G=(L,R,E)$ with $|L|=|R|=n$ is said to be $\varepsilon$-pseudorandom if it holds for every $S\subseteq L$ and $T\subseteq R$ that
  \begin{equation*}
    \left||E(S,T)|-\frac{r|S||T|}{n}\right| \leq \varepsilon r\sqrt{|S||T|}.
  \end{equation*}
\end{definition}

By the well-known expander mixing lemma, the 2-lift of any $\lambda$-spectral expander graph is $\lambda$-pseudorandom. Thus in particular, the 2-lift of a $r$-regular Ramanujan graph is $\varepsilon$-pseudorandom for $\varepsilon=2\sqrt{r-1}/r$. It follows that there exist explicit families of $r$-regular $\varepsilon$-pseudorandom graphs for all $r\geq 4/\varepsilon^2$.

We now present the AEL construction in the quantum setting. We first present a basic version that amplifies distance without decreasing the alphabet size. We will then explain how the alphabet size can also be reduced. The analysis of the quantum version of the AEL construction in \Cref{prop:aelbasic} and \Cref{prop:aelred} below directly adapts the well-known analogous results for classical codes dating back to \cite{AEL95}. However, to the best of our knowledge the application to the quantum setting is new, as are the corollaries below providing quantum codes approaching the quantum Singleton bound that are efficiently unique-decodable up to half their distance.

\subsection{Distance Amplification without Alphabet Reduction}
\label{sec:aelbasic}
We first describe a simple version of AEL amplification that does not reduce alphabet size. Let $Q_{\text{out}}$ and $Q_{\text{in}}$ be $[[n_{\text{out}},k_{\text{out}},\Delta_{\text{out}}n_{\text{out}}]]_{q_{\text{out}}}$ and $[[n_{\text{in}},k_{\text{in}},\Delta_{\text{in}}n_{\text{in}}]]_{q_{\text{in}}}$ quantum codes respectively such that $q_{\text{out}}=q_{\text{in}}^{k_{\text{in}}}$. Let $Q_\diamond=Q_{\text{out}}\diamond Q_{\text{in}}$ denote the $[[n_{\text{out}} n_{\text{in}},k_{\text{out}}k_{\text{in}}]]_{q_{\text{in}}}$ concatenated code. The key step is now to permute the components of $Q_\diamond$ according to the edges of an $n_{\text{in}}$-regular bipartite expander, and to block the $n_{\text{out}}n_{\text{in}}$ $q_{\text{in}}$-ary components into $n_{\text{out}}$ $q_{\text{in}}^{n_{\text{in}}}$-ary components. Specifically, let $G=(L=[n_{\text{out}}],R=[n_{\text{out}}],E)$ be an $n_{\text{in}}$-regular $\varepsilon_0$-pseudorandom bipartite graph, and assume that the edges at each vertex have an arbitrary labeling by elements of $[n_{\text{in}}]$. Let $\pi_G:[n_{\text{out}}]\times[n_{\text{in}}]\rightarrow[n_{\text{out}}]\times[n_{\text{in}}]$ be the permutation that maps $(i,j)\in[n_{\text{out}}]\times[n_{\text{in}}]$ to the unique pair $(i',j')$ such that the $j$th outgoing edge in $G$ from vertex $i\in L$ lands on $i'\in R$, and $j'$ is the label that vertex $i'$ assigns to this edge $(i,i')$. Define $Q=\pi_G(Q_\diamond)$ to be the $[[n_{\text{out}},R_{\text{in}}k_{\text{out}}]]_{q_{\text{in}}^{n_{\text{in}}}}$ code obtained by applying the permutation $\pi_G$ to the components of $Q_\diamond$, and then regrouping each resulting block of $n_{\text{in}}$ components to a single $q_{\text{in}}^{n_{\text{in}}}$-ary symbol. Observe that $Q$ has the same rate $R=R_\diamond=R_{\text{out}}R_{\text{in}}$ as $Q_\diamond$.

\begin{proposition}
\label{prop:aelbasic}
  The code $Q$ defined above has relative distance $\Delta\geq\Delta_{\text{in}}-2\varepsilon_0\sqrt{\Delta_{\text{in}}/\Delta_{\text{out}}}$.
  
  More generally, for every $0<\alpha_{\text{out}},\alpha_{\text{in}}<1$, if a codeword of $Q$ experiences an error on $\alpha\leq\alpha_{\text{in}}-\varepsilon_0\sqrt{\alpha_{\text{in}}/\alpha_{\text{out}}}$ fraction of its ($q_{\text{in}}^{n_{\text{in}}}$-ary) components, then after applying $\pi_G^{-1}$, fewer than $\alpha_{\text{out}}$ fraction of the resulting inner code blocks experience errors on at least $\alpha_{\text{in}}$ fraction of their ($q_{\text{in}}$-ary) components.
\end{proposition}

\Cref{prop:aelbasic} follows directly from the following lemma, which captures the key property of the permutation $\pi_G$.

\begin{lemma}
\label{lem:expperm}
Let $G=(L=[n_{\text{out}}],R=[n_{\text{out}}],E)$ be an $n_{\text{in}}$-regular $\varepsilon_0$-pseudorandom bipartite graph, and define $\pi_G:[n_{\text{out}}]\times[n_{\text{in}}]\rightarrow[n_{\text{out}}]\times[n_{\text{in}}]$ as above. Then for every $0<\alpha_{\text{out}},\alpha_{\text{in}}<1$, and for every $T\subseteq R$ of size $|T|\leq(\alpha_{\text{in}}-\varepsilon_0\sqrt{\alpha_{\text{in}}/\alpha_{\text{out}}})\cdot n_{\text{out}}$, there are fewer than $\alpha_{\text{out}}\cdot n_{\text{out}}$ vertices $i\in L$ for which
\begin{equation}
\label{eq:inneroverload}
    |\{j\in[n_{\text{in}}]:\pi_G(i,j)\in T\times[n_{\text{in}}]\}| \geq \alpha_{\text{in}}\cdot n_{\text{in}}.
\end{equation}
\end{lemma}
\begin{proof}
  The proof is a direct application of the definition of $\varepsilon_0$-pseudorandomness. Let $S\subseteq L$ be the set of vertices for which \Cref{eq:inneroverload} holds. Then by definition
  \begin{align*}
      |E(S,T)|
      &\geq |S| \cdot \alpha_{\text{in}} n_{\text{in}}.
  \end{align*}
  Meanwhile by the $\varepsilon_0$-pseudorandomness of $G$,
  \begin{align*}
      |E(S,T)|
      &\leq \frac{n_{\text{in}}|S||T|}{n_{\text{out}}} + \varepsilon_0 n_{\text{in}} \sqrt{|S||T|} \\
      &< |S|\cdot\alpha_{\text{in}}n_{\text{in}} - |S|\cdot\varepsilon_0 n_{\text{in}}\sqrt{\frac{\alpha_{\text{in}}}{\alpha_{\text{out}}}} + |S|\cdot\varepsilon_0n_{\text{in}}\sqrt{\frac{\alpha_{\text{in}}}{|S|/n_{\text{out}}}}
  \end{align*}
  where the second inequality above applies the fact that $|T|\leq(\alpha_{\text{in}}-\varepsilon_0\sqrt{\alpha_{\text{in}}/\alpha_{\text{out}}})\cdot n_{\text{out}}$. The above inequalities give a contradiction whenever $|S|/n_{\text{out}}\geq \alpha_{\text{out}}$, so we must have $|S|/n_{\text{out}}<\alpha_{\text{out}}$, as desired.
\end{proof}

To apply \Cref{lem:expperm} to an AEL code $Q=\pi_G(Q_{\text{out}}\diamond Q_{\text{in}})$, we will typically choose $\alpha_{\text{out}}$ and $\alpha_{\text{in}}$ to be the decoding radii of $Q_{\text{out}}$ and $Q_{\text{in}}$ respectively. Then if $Q$ experiences errors on $\alpha_{\text{in}}-\varepsilon_0\sqrt{\alpha_{\text{in}}/\alpha_{\text{out}}}$ fraction of its components, we may apply the natural decoding algorithm $\Dec_{Q_{\text{out}}}\circ\Dec_{Q_{\text{in}}}^{\otimes n_{\text{out}}}\circ\pi_G^{-1}$, and \Cref{lem:expperm} guarantees that fewer than $\alpha_{\text{out}}$ fraction of the inner decodings fail, so that the outer decoding will succeed. The following proof of \Cref{prop:aelbasic} formalizes this idea for unique decoding for both the inner and outer codes. We will subsequently follow \cite{Guruswami2008ExplicitCA} in applying the same idea with list decoding for the inner code and list recovery for the outer code.

\begin{proof}[Proof of \Cref{prop:aelbasic}]
  The second statement in the proposition is equivalent to \Cref{lem:expperm}, so we just need to show the first statement. For this purpose, it suffices to show that $Q$ can be decoded from errors acting on at most $\kappa=\Delta_{\text{in}}/2-\varepsilon_0\sqrt{\Delta_{\text{in}}/\Delta_{\text{out}}}$ fraction of the components. Assume some error acts on some subset $T\subseteq[n_{\text{out}}]$ of the code components with $|T|\leq\kappa n_{\text{out}}$. Let $S\subseteq[n_{\text{out}}]$ be the set of blocks in $Q_\diamond=\pi_G^{-1}(Q)$ inside which at least $\Delta_{\text{in}}n_{\text{in}}/2$ of the qudits are mapped by $\pi_G$ to blocks in $T$. \Cref{lem:expperm} with $\alpha_{\text{in}}=\Delta_{\text{in}}/2$ and $\alpha_{\text{out}}=\Delta_{\text{out}}/2$ gives that $|S|<\Delta_{\text{out}}n_{\text{out}}/2$, so the natural decoding algorithm $\Dec_{Q_{\text{out}}}\circ\Dec_{Q_{\text{in}}}^{\otimes n_{\text{out}}}\circ\pi_G^{-1}$ is guaranteed to successfully correct the error, as all inner decodings in blocks outside of $S$ must succeed.
\end{proof}

Below we show how applying AEL amplification to asymptotically good quantum LDPC codes with a constant sized random inner code yields efficiently decodable quantum LDPC codes approaching the quantum Singleton bound over constant sized alphabets. To the best of our knowledge, such a family of quantum codes was not previously known even without the LDPC requirement.

\begin{corollary}
\label{cor:qLDPC}
For every $\gamma>0$ and $0<R<1$, there exists some $q=q(\gamma)$ such that there exists an explicit infinite family of quantum LDPC CSS codes of rate $R$, relative distance $\Delta=(1-R-\gamma)/2$, and alphabet size $q$ that are decodable in time $O_\gamma(n)$ up to errors on $\Delta/2$ fraction of the $n$ components.
\end{corollary}
\begin{proof}
We may assume that $R<1-\gamma$, as otherwise $\Delta\leq 0$ and the result holds trivially. We may also assume that $R\geq\gamma/2$, as otherwise we may simply construct an AQECC of larger rate $\gamma/2$ using the argument below for updated parameters $R'=\gamma/2$ and $\gamma'=\gamma/2$. By \cite{panteleev2022asymptotically,leverrier2022quantum,leverrier2022efficient,dinur2022good,gu2022efficient,leverrier2022parallel}, there exist explicit infinite families of $[[n_{\text{out}},k_{\text{out}}]]_2$ qLDPC CSS codes $Q_{\text{out}}$ of rate $1-\gamma/20$ over a binary alphabet that are decodable in linear time (for fixed $\gamma$) up to some constant fraction $\kappa=\kappa(\gamma)$ of errors. Here we do not need to let $\kappa$ or the decoding runtime depend directly on $R$ because $\gamma/2\leq R\leq 1-\gamma$ by assumption. 

Futhermore, by \Cref{claim:randomcss}, there exist sufficiently large $q_{\text{in}}=q_{\text{in}}(\gamma)$ and $n_0=n_0(\gamma)$ such that for every $n_{\text{in}}\geq n_0$, a random CSS code $Q_{\text{in}}$ of rate $R+\gamma/20$, block length $n_{\text{in}}$, and alphabet size $q_{\text{in}}$ will with positive probability have relative distance $\geq(1-R-\gamma/10)/2$.
Here we may assume that $q_{\text{in}}=2^m$ for some $m\in\mathbb{N}$, so the classical codes making up the random CSS code are linear over $\mathbb{F}_{q_{\text{in}}}$, and therefore are also $\mathbb{F}_2$-linear.

Now set $\gamma_0=\gamma\sqrt{\kappa}/40$, and set $n_{\text{in}}=\max\{n_0,4/\gamma_0^2\}$. Let $Q_\diamond=Q_{\text{out}}\diamond Q_{\text{in}}$ denote the concatenated code, where before concatenating we first block together groups of $mk_{\text{in}}$ qubits in $Q_{\text{out}}$ to obtain a code with the same rate and decoding radius over the alphabet $(\mathbb{F}_2^m)^{k_{\text{in}}}$. Finally, for some $n_{\text{in}}$-regular $\gamma_0$-pseudorandom bipartite graph $G$, let $Q=\pi_G(Q_\diamond)$. Such a graph exists because $n_{\text{in}}\geq 4/\gamma_0^2$ by assumption.

Thus we have constructed an infinite family of quantum codes $Q$ of rate $(1-\gamma/20)(R+\gamma/20)\geq R$ over the alphabet $\mathbb{F}_2^m$. By \Cref{prop:aelbasic}, these codes are decodable up to errors acting on $\Delta_{\text{in}}/2-2\gamma_0/\sqrt{\kappa}>(1-R-\gamma)/4$ fraction of the components. The proof of \Cref{prop:aelbasic} in fact implies that the natural decoding algorithm $\Dec_Q=\Dec_{Q_{\text{out}}}\circ\Dec_{Q_{\text{in}}}^{\otimes n_{\text{out}}}\circ\pi_G^{-1}$ decodes up to this radius. For fixed $\gamma$, this algorithm runs in time linear in the block length because $\Dec_{Q_{\text{out}}}$ by assumption runs in linear time, and $Q_{\text{in}}$ is a constant size CSS code and thus can be decoded in constant time via brute force. Also observe that for fixed $\gamma,R$, then $Q_{\text{out}}$ is an explicit LDPC CSS code, and $Q_{\text{in}}$ is a constant sized CSS code that can be constructed in constant time via a brute force search. Thus $Q_\diamond$ and therefore $Q$ is an explicit quantum LDPC CSS code.
\end{proof}

\subsection{Distance Amplification with Alphabet Reduction}
\label{sec:aelred}
We now show how the AEL construction can simultaneously amplify the distance and reduce the alphabet size of the outer code. Again let $Q_{\text{out}}$ and $Q_{\text{in}}$ be $[[n_{\text{out}},k_{\text{out}},\Delta_{\text{out}}n_{\text{out}}]]_{q_{\text{out}}}$ and $[[n_{\text{in}},k_{\text{in}},\Delta_{\text{in}}n_{\text{in}}]]_{q_{\text{in}}}$ quantum codes respectively such that $q_{\text{out}}=q_{\text{in}}^{k_{\text{in}}}$. Let $Q_\diamond=Q_{\text{out}}\diamond Q_{\text{in}}$ denote the $[[n_{\text{out}} n_{\text{in}},k_{\text{out}}k_{\text{in}}]]_{q_{\text{in}}}$ concatenated code. As in \Cref{sec:aelbasic}, we will permute the components of $Q_\diamond$ according to an appropriate expander graph and reblock into larger symbols. However, we will now reblock into smaller symbols than before.

Specifically, for some $r\leq n_{\text{in}}$, we partition the code components into groups of $r$ $q_{\text{in}}$-ary symbols, each of which we block into a $q_{\text{in}}^r$-ary symbol. To avoid rounding issues, for simplicity assume that $r$ divides $n_{\text{in}}$, and let $b=n_{\text{in}}/r$. Let $G=(L=[n_{\text{out}}b],R=[n_{\text{out}}b],E)$ be an $r$-regular $\varepsilon_0$-pseudorandom bipartite graph, and assume that the edges at each vertex have an arbitrary labeling by elements of $[r]$. Similarly as in \Cref{sec:aelbasic}, let $\pi_G:[n_{\text{out}}b]\times[r]\rightarrow[n_{\text{out}}b]\times[r]$ be the permutation that maps $(i,j)\in[n_{\text{out}}b]\times[r]$ to the unique pair $(i',j')$ such that the $j$th outgoing edge in $G$ from vertex $i\in L$ lands on $i'\in R$, and $j'$ is the label that vertex $i'$ assigns to this edge $(i,i')$. Define $Q=\pi_G(Q_\diamond)$ to be the $[[n_{\text{out}}b,k_{\text{out}}k_{\text{in}}/b]]_{q_{\text{in}}^r}$ code obtained by applying the permutation $\pi_G$ to the components of $Q_\diamond$, and then regrouping each consecutive block of $r$ $q_{\text{in}}$-ary components into a single $q_{\text{in}}^r$-ary component. As before, $Q$ has the same rate $R=R_\diamond=R_{\text{out}}R_{\text{in}}$ as $Q_\diamond$. However, now the alphabet size $q_{\text{in}}^r$ of $Q$ does not depend on the inner code's block length $n_{\text{in}}$.

\begin{proposition}
\label{prop:aelred}
The code $Q$ defined above has relative distance $\Delta\geq\Delta_{\text{in}}-6(\varepsilon_0/2\cdot\sqrt{\Delta_{\text{in}}/\Delta_{\text{out}}})^{2/3}$.

More generally, for every $0<\alpha_{\text{out}},\alpha_{\text{in}}<1$, if a codeword of $Q$ experiences an error on $\alpha\leq\alpha_{\text{in}}-3(\varepsilon_0/2\cdot\sqrt{\alpha_{\text{in}}/\alpha_{\text{out}}})^{2/3}$ fraction of its ($q_{\text{in}}^r$-ary) components, then after applying $\pi_G^{-1}$, fewer than $\alpha_{\text{out}}$ fraction of the resulting inner code blocks experience errors on at least $\alpha_{\text{in}}$ fraction of their ($q_{\text{in}}$-ary) components.


\end{proposition}
\begin{proof}
We will prove the second statement in the proposition, as the first statement then follows by letting $\alpha_{\text{out}}=\Delta_{\text{out}}/2$ and $\alpha_{\text{in}}=\Delta_{\text{in}}/2$ analogously as in \Cref{prop:aelbasic}. Let $\varepsilon'=(\varepsilon_0/2\cdot\sqrt{\alpha_{\text{in}}/\alpha_{\text{out}}})^{2/3}$. Then $Q$ experiences an error on $\alpha\leq\alpha_{\text{in}}-\varepsilon'-\varepsilon_0\sqrt{\alpha_{\text{in}}/\varepsilon'\alpha_{\text{out}}}$ fraction of its components, so \Cref{lem:expperm} implies that fewer than $\varepsilon'\alpha_{\text{out}}$ fraction of the components of $\pi_G^{-1}(Q)$ (each of which is itself a length-$r$ block of $q_{\text{in}}$-ary symbols) experience an error acting on at least $\alpha_{\text{in}}-\varepsilon'$ fraction of its symbols. Now recall that each inner code block of $\pi_G^{-1}(Q)$ consists of $b$ of these length-$r$ blocks with $br=n_{\text{in}}$. Thus for a given inner code block, if at most $\varepsilon'$ fraction of the $b$ length-$r$ blocks experience an error acting on at least $(\alpha_{\text{in}}-\varepsilon')r$ symbols, then the number of symbols in the entire inner block that experience an error is at most $\varepsilon'br+(1-\varepsilon')(\alpha_{\text{in}}-\varepsilon')br<\alpha_{\text{in}}br$. Thus an inner code block can only experience errors on $\geq\alpha_{\text{in}}$ fraction of its $n_{\text{in}}=br$ symbols if $>\varepsilon'$ fraction of its length-$r$ blocks experience errors on at least $\alpha_{\text{in}}-\varepsilon'$ fraction of their symbols. Therefore if at least $\alpha_{\text{out}}$ fraction of the inner code blocks experience errors on at least $\alpha_{\text{in}}$ fraction of their components, then more than $\varepsilon'\alpha_{\text{out}}$ fraction of the length-$r$ blocks in the entire code $\pi_G^{-1}(Q)$ must experience errors on $\geq\alpha_{\text{in}}-\varepsilon'$ fraction of their symbols. But as described above, \Cref{lem:expperm} implies that this latter statement cannot occur, so fewer than $\alpha_{\text{out}}$ fraction of the inner code blocks can experience errors on at least $\alpha_{\text{in}}$ fraction of their components, as desired.
\end{proof}

We will typically apply \Cref{prop:aelred} with an outer code of growing alphabet size $q_{\text{out}}=\poly(n_{\text{out}})$, such as quantum Reed-Solomon or folded Reed-Solomon, and an inner code of fixed alphabet size $q_{\text{in}}=O(1)$, such as a random CSS code. 

\section{Constant-Alphabet Quantum List Decodable Codes Near the Singleton Bound}

\label{sec:frsael}

In this section, we apply the AEL distance amplification and alphabet reduction construction described in \Cref{sec:ael} to the folded quantum Reed-Solomon codes described in \Cref{sec:rsfolding}. This construction is a quantum analogue of the classical construction of \cite{Guruswami2008ExplicitCA}. As in the classical case, the resulting quantum codes have constant alphabet size and are efficiently list-decodable with polynomial list size for a fraction of adversarial errors approaching the Singleton bound.

\begin{theorem}
\label{thm:frsael}
For every $\gamma>0$ and $0<R<1$, there exists an infinite family of quantum CSS codes $Q$ of rate $R$, relative distance at least $\delta=(1-R-\gamma)/2$, and alphabet size $q=2^{O(1/\gamma^5)}$ that are $(\delta,L=n^{O(1/\gamma^3)})$ list-decodable in time $n^{O(1/\gamma^4)}$, where $n$ denotes the block length of $Q$. Furthermore, there exists a randomized algorithm that constructs $Q$ in time $n^{O(1/\gamma^4)}$ with failure probability $2^{-\Omega(n)}$.
\end{theorem}

The proof of \Cref{thm:frsael} concatenates a folded quantum Reed-Solomon outer code (see \Cref{sec:rsfolding}) with a list-decodable random CSS inner code and then reduces the alphabet size using AEL. To list-decode such a concatenation, we typically want to use a {\it list-recoverable} outer code as defined below.

\begin{definition}
    A code $C\subset \Sigma^n$ is $(\eta, \ell, L)$-LR (list-recoverable) if $\forall S_1\cdots S_n\subset \Sigma$, with $|S_i|\leq \ell$, 
    \begin{equation}
        \bigg|\bigg\{ c\in  C: c_i\in S_i \text{ for at least }\eta \cdot n \text{ symbols }i\in [n]\bigg\}\bigg| \leq L 
    \end{equation}
\end{definition}

The following result of \cite{Guruswami2008ExplicitCA} shows that folded Reed-Solomon codes have this property.

\begin{theorem}[\cite{Guruswami2008ExplicitCA}]
\label{thm:frslr}
For every $0<R<1$, and $\gamma>0$, $\ell\geq 1$, and $m\geq 1$, the $m$-folded Reed-Solomon codes over $\mathbb{F}_q$ of block length $N$ are 
\begin{equation*}
    \left(\eta=R\cdot\left(1+\frac{s}{t}\right)\cdot\frac{m}{m-s+1}\cdot\left(\frac{\ell}{R}\right)^{1/(s+1)},\; \ell,\; L=(Nm)^{O(s)}\right)
\end{equation*}
list-recoverable in time $(Nm)^{O(s)}$ for all choices of integers $s\leq m$ and $t\geq 1$ satisfying $(s+t)(\ell/R)^{1/(s+1)}<q$.
\end{theorem}

In \Cref{thm:frslr}, we will typically have $q$ growing with the block length $N$ while $s,t,\ell,R$ are fixed, so the condition that $(s+t)(\ell/R)^{1/(s+1)}<q$ trivially holds. Then for $\gamma>0$ and $\ell=O_\gamma(1)$ we may choose sufficiently large $s,t,m=O_\gamma(1)$ to obtain list-recovery up to a radius $\eta=R+O(\gamma)$ lying within $O(\gamma)$ of the Singleton bound.

\begin{proof}[Proof of \Cref{thm:frsael}]
The code $Q$ is an instantiation of the AEL construction described in \Cref{sec:aelred} with the following parameters. Note that we may assume that $\gamma/2\leq R\leq 1-\gamma$, as if $R<\gamma/2$ we may simply instead construct a code of greater rate $\gamma/2$, and if $R>1-\gamma$ then the desired relative distance satisfies $\delta<0$, which is trivially achievable. For completeness below we give explicit constants, which we do not attempt to optimize.
\begin{itemize}
    \item Let $Q_{\text{out}}^0$ be a $[[n_0,k_0]]_{q_0}$ quantum Reed-Solomon code of rate $R_{\text{out}}=k_0/n_0=1-\gamma/20$ and relative distance $\delta_{\text{out}}=(1-R_{\text{out}})/2$. We may assume that $n_0=\Theta(q_0)$, for instance by choosing $\mathbb{F}_{q_0}$ to be a field of characteristic $2$ so that we may take $q_0/2\leq n_0\leq q_0$.
    \item Set
    \begin{itemize}
        \item $q_{\text{in}}=2^{200/\gamma}=2^{O(1/\gamma)}$
        \item $\ell=q_{\text{in}}^{100/\gamma}=2^{O(1/\gamma^2)}$
        \item $s=100/\gamma\cdot\log(\ell/\gamma)=O(1/\gamma^3)$
        \item $m=100s/\gamma=O(1/\gamma^4)$.
    \end{itemize}
    Let $Q_{\text{out}}$ be the $[[n_{\text{out}},k_{\text{out}}]]_{q_{\text{out}}}$ code defined to be the $m$-folding of $Q_{\text{out}}^0$ (see \Cref{sec:rsfolding}).
    \item Let $Q_{\text{in}}$ be a $[[n_{\text{in}},k_{\text{in}}]]_{q_{\text{in}}}$ CSS code of rate $R_{\text{in}}=R+\gamma/20$ and dimension $k_{\text{in}}=\log_{q_{\text{in}}}q_{\text{out}}$ that has relative distance $\geq\delta_{\text{in}}=(1-R_{\text{in}}-\gamma/10)/2$ and consists of two classical $(\delta_{\text{in}},\ell)$ list-decodable codes, so that $Q_{\text{in}}$ $(\delta_{\text{in}},\ell^2)$ list-decodable. By \Cref{claim:randomcss}, a random CSS code satisfies these properties with probability $1-q_{\text{in}}^{-\Omega(\gamma n)}$.
    \item Let $\gamma_0=\gamma^2/1000$, $r=4/\gamma_0^2=O(1/\gamma^4)$, and let $G$ be an $r$-regular $\gamma_0$-pseudorandom bipartite graph.
    \item As described in \Cref{sec:aelred}, let $Q_\diamond=Q_{\text{out}}\diamond Q_{\text{in}}$, and let $Q=\pi_G(Q_\diamond)$ be our final construction.
\end{itemize}

A direct application of \Cref{prop:aelred} implies that $Q$ has relative distance $\geq\delta$. We now show that $Q$ is efficiently $(\delta,L)$ list-decodable. For this purpose, observe that $Q_\diamond$ is the concatenation of the CSS codes $Q_{\text{out}}=\text{CSS}(C_1^{\text{out}},C_2^{\text{out}})$ and $Q_{\text{in}}=\text{CSS}(C_1^{\text{in}},C_2^{\text{in}})$, and thus by \Cref{claim:concatcss},
\begin{equation*}
    Q = \text{CSS}\left(C_1=\pi_G(C_1^{\text{out}}\diamond (C_1^{\text{in}}/{C_2^{\text{in}}}^\perp)),\; C_2=\pi_G(C_2^{\text{out}}\diamond(C_2^{\text{in}}/{C_1^{\text{in}}}^\perp))\right),
\end{equation*}
Therefore by \Cref{thm:CSS}, it suffices to show that $C_1/C_2^\perp$ and $C_2/C_1^\perp$ above are efficiently $(\delta,\sqrt{L})$ list-decodable in the sense of \Cref{def:cosetld}. By symmetry of the construction, it suffices to consider $C_1/C_2^\perp$, as $C_2/C_1^\perp$ is identical. Then our goal is to show that given a corrupted codeword $x$ of $C_1$, there exists an efficient algorithm that outputs a list of size $\leq\sqrt{L}$ containing a representative of every coset in $C_1/C_2^\perp$ that intersects the Hamming ball $B_{\delta n}(x)$ of radius $\delta n$ centered at $x$. Recall that then subtracting $x$ from each list element yields the desired list of stabilizer-distinct $X$ error operators for quantum list-decoding as stated in \Cref{def:QLD}; the $Z$ error operators are obtained by analogously list-decoding $C_2/C_1^\perp$.

The desired list-decoding algorithm is as follows. Let $\alpha_{\text{in}}=\delta_{\text{in}}$ and $\alpha_{\text{out}}=\gamma/50$. First, for each $i\in[n_{\text{out}}]$, we apply a brute force search to compute the list $S_i\subseteq[q_{\text{out}}]=[q_{\text{in}}^{k_{\text{in}}}]\cong C_1^{\text{in}}/{C_2^{\text{in}}}^\perp$ of all inner code messages whose encodings differ from the $i$th inner code block of $\pi_G^{-1}(x)$ in at most $\alpha_{\text{in}}$-fraction of positions. Then we apply the list-recovery algorithm given by \Cref{thm:frsael} to $S_1,\dots,S_{n_{\text{out}}}$, and output the (encodings into $C_1$ of the) returned list. Recall that when we re-encode the inner code, we apply the maps $\Enc_1^{\text{in}}$ and $\Enc_2^{\text{in}}$ described in \Cref{sec:concat} that send outer code symbols to cosets, and then we may select an arbitrary element from each coset. 

The brute force searches over the inner code together take $O(n_{\text{out}}\cdot q_{\text{out}}) \leq n^{O(m)}=n^{O(1/\gamma^4)}$ time, and the outer list-recovery takes $(n_{\text{out}}m)^{O(s)}\leq n^{O(1/\gamma^3)}$ time, so the overall runtime is $n^{O(1/\gamma^4)}$.

To show correctness, assume that $x$ is obtained by corrupting at most $\delta$-fraction of the symbols in some codeword $c\in C_1$. It suffices to show that the list returned by the above algorithm includes an element of the coset $c+C_2^\perp$. By \Cref{prop:aelred} and by the list-decodability of $C_1^{\text{in}}$, at most $\alpha_{\text{out}}$-fraction of the inner decodings return a list that does not contain the respective component of $c$. Thus by \Cref{thm:frslr}, the final list returned by the algorithm will contain $c$, up to addition by some element of $\pi_G(({C_2^{\text{in}}}^\perp)^{\oplus n_{\text{out}}})\subseteq C_2^\perp$ from the re-encoding step described above. Also observe that by \Cref{thm:frslr}, the returned list will have size at most $(n_{\text{out}}m)^{O(s)}\leq n^{O(1/\gamma^3)}$.

It remains to be shown that $Q$ can be constructed in time $n^{O(1/\gamma^4)}$ with failure probability $2^{-\Omega(n)}$. The folded Reed-Solomon code is explicit, so the construction algorithm simply needs to find an appropriate inner CSS code $Q_{\text{in}}$. By \Cref{claim:randomcss}, a random choice of $Q_{\text{in}}$ works with probability $>1/2$, so trying $n$ random codes and checking the distance and list-decodability of each one by brute force gives the desired construction algorithm.
\end{proof}

\section{AQEC without List Recovery}
\label{sec:aqeccdirect}
This section presents a construction of efficient AQECCs that avoid the need for list recovery. In fact, the only algorithmic ingredient this construction relies on is an efficient unique-decoder for an asymptotically good code, as the construction only performs list-decoding on a small inner code where a brute force search is sufficient.

The main idea of this construction is to concatenate an outer high-rate constant-distance unique-decodable QECC with an inner AQECC that permits exact unique decoding with high probability almost up to the Singleton bound. Specifically, for the inner code we use the concatenation of a quantum list-decodable code near the Singleton bound with a PTC, to ensure (inefficient) unique decodability near the Singleton bound with high probability for any fixed error. The inefficiency here is not a concern due to the small block length of the inner code. We then apply AEL distance amplification to boost the distance of the concatenated code almost to the Singleton bound.

We now show how to construct private AQECCs approaching the Singleton bound. \Cref{theorem:removing-privacy} then gives the desired non-private AQECCs. 

\begin{theorem}
  \label{thm:uniqueAQECC}
  For any given $\gamma>0$ and $0<R<1$, let $\delta_{\text{out}}>0$ be such that there is an infinite explicit family of stabilizer codes $Q_0$ of rate $1-\gamma/4$ and alphabet size $q_0$ that are unique-decodable from errors acting on $\delta_{\text{out}}$ fraction of the components in $f(n_0\log q_0)$ time for some increasing function $f:\mathbb{R}_+\rightarrow\mathbb{R}_+$, where $n_0$ denotes the block length.
  
  Then there exists an infinite explicit family of private $((1-R-\gamma)/2,2^{-\Omega(\delta_{\text{out}}n)})$-AQECCs $Q$ of rate $R$, alphabet size $q=\max\{q_0,2^{1/\gamma^2\delta_{\text{out}}}\}^{O(1/\gamma)}$, private key length $O(\gamma^2\cdot n\log q)$ bits, and decoding time $f(n\log q)+O(nq^2)$, where $n$ denotes the block length.
\end{theorem}


\begin{proof}
  We may assume that $R\leq 1-\gamma$, as otherwise the decoding radius in the theorem statement is $0$. We may also assume that $R\geq\gamma/2$, as otherwise we may simply construct an AQECC of larger rate $\gamma/2$ using the argument below for updated parameters $R'=\gamma/2$ and $\gamma'=\gamma/2$; such a change only affects the resulting code's parameters by constant factors in the $R'$ and $\gamma'$ coefficients.
  
  We construct $Q$ as follows. Below, we ignore rounding issues when choosing integer-valued parameters for ease of presentation; such rounding only affects the resulting bounds by constant factors. For completeness we give explicit constants, which we do not attempt to optimize.
  \begin{itemize}
      \item Let $R_{\text{in}}=R+26\gamma/100$, $\delta_{\text{in}}=1-R_{\text{in}}-\gamma/40$ $q_{\text{in}}=2^{100/\gamma}$, $\ell=q_{\text{in}}^{400/\gamma}=2^{O(1/\gamma^2)}$ $\gamma_0=\gamma\sqrt{\delta_{\text{out}}}/500$, $n_{\text{in}}=\max\{(2/\gamma)\log q_0/\log q_{\text{in}},4/\gamma_0^2\}=\max\{\Theta(\log q_0),\Theta(1/\gamma^2\delta_{\text{out}})\}$. Let $Q_{\text{in}}$ be an $[[n_{\text{in}},k_{\text{in}}]]_{q_{\text{in}}}$ CSS code of rate $R_{\text{in}}$ and relative distance $\geq\delta_{\text{in}}$ that is $(\delta_{\text{in}},\ell)$ quantum list-decodable. By \Cref{claim:randomcss}, a random CSS code satisfies these properties with probability $\geq 1-q_{\text{in}}^{-\Omega(\gamma n)}\geq 1/2$. Such a code can be found via a brute force search in time $q_{\text{in}}^{O(n_{\text{in}}^2)}=\max\{2^{O((\log q_0)^2/\gamma)},2^{O(1/\gamma^5\delta_{\text{out}}^2)}\}$, which is constant with respect to $n$.
      \item Let $P_{\text{in}}=\{P_{\text{in}}^\kappa:\kappa\in K\}$ be a $\delta_{\text{out}}/10\ell$-PTC of block length $n_P=k_{\text{in}}$, rate $R_P\geq 1-\gamma/100$, (so message length is $k_P=R_Pn_P$), alphabet size $q_{\text{in}}$, and private key set of size $|K|=q_{\text{in}}^\lambda$ for $\lambda=500(1/\gamma+\log(1/\delta_{\text{out}}))+(\gamma/100)\log n_P$. Such a PTC was constructed by \cite{barnum2002authentication}, as stated in \Cref{thm:ptc}.
      \item Let $q_{\text{out}}=q_{\text{in}}^{k_P}\geq q_{\text{in}}^{\gamma n_{\text{in}}/2}\geq q_0$, and let $Q_{\text{out}}$ be the $[[n_{\text{out}},k_{\text{out}},\delta_{\text{out}}]]_{q_{\text{out}}}$ code of rate $R_{\text{out}}=1-\gamma/4$ obtained by reblocking (that is, folding) the code $Q_0$ defined in the theorem statement to increase the alphabet size from $q_0\leq q_{\text{out}}$ to $q_{\text{out}}$.
      \item Let $q=q_{\text{in}}^{n_{\text{in}}}$, and for every $\vec{\kappa}\in K^n$, define the $[[n,k]]_q$ code $Q_\diamond^{\vec{\kappa}}=\bigotimes_{i=1}^n(Q_{\text{in}}\circ P_{\text{in}}^{\kappa_i})\circ Q_{\text{out}}$. That is, $Q_\diamond$ has encoding map $\Enc_{Q_\diamond^{\vec{\kappa}}}=\Enc_{Q_{\text{in}}}^{\otimes n}\circ\bigotimes_{i=1}^n\Enc_{P_{\text{in}}^{\kappa_i}}\circ\Enc_{Q_{\text{out}}}$. Here $\Enc_{Q_\diamond^{\vec{\kappa}}}$ takes as input a message state and the classical private key $\vec{\kappa}$, the latter of which consists of $n$ independent private keys for the $n$ respective calls to $\Enc_{P_{\text{in}}}$.
      \item Let $G$ be an $n_{\text{in}}$-regular $\gamma_0$-pseudorandom bipartite graph, and for $\vec{\kappa}\in K^n$ let $Q^{\vec{\kappa}}=\pi_G(Q_{\diamond}^{\vec{\kappa}})$, where $\pi_G:[n]\times[n_{\text{in}}]\rightarrow[n]\times[n_{\text{in}}]$ is the permutation specified by $G$ as described in \Cref{sec:aelbasic}.
  \end{itemize}
  
  We now show that the private code $Q=\{Q^{\vec{\kappa}}:\vec{\kappa}\in K^n\}$ defined above has the desired properties. By definition, $Q$ has rate $R_{\text{out}}R_PR_{\text{in}}\geq R$, alphabet size $q=2^{O(1/\gamma^3\delta_{\text{out}})}$, and private key length $n\cdot\log(q_{\text{in}}^\lambda)\leq O(n(1/\gamma)(1/\gamma+\log(1/\delta_{\text{out}})+\gamma\log n_{\text{in}}))\leq O(\gamma^2\cdot n\log q)$, where this final inequality uses the fact that $\delta_{\text{out}}\leq\gamma/8$ by the quantum Singleton bound for $Q_0$, and the fact that $q=q_{\text{in}}^{n_{\text{in}}}$ with $n_{\text{in}}\geq 4/\gamma_0^2=\Theta(1/\gamma^2\delta_{\text{out}})$.
  
  It remains to be shown that $Q$ has a decoding algorithm with the desired error and efficiency parameters. We define the following natural decoding algorithm $\Dec_Q^{\vec{\kappa}}$. Let $\Dec_{\text{in}}^\kappa=\Enc_{Q_{\text{in}\circ P_{\text{in}}^\kappa}}^\dagger\circ\tDec_{\text{in}}^\kappa$ be the brute force decoding algorithm for the composed stabilizer code $Q_{\text{in}}\circ P_{\text{in}}^\kappa$. Here $\tDec_{\text{in}}^\kappa$ first performs a syndrome measurement for $Q_{\text{in}}\circ P_{\text{in}}^\kappa$ and finds all stabilizer-distinct errors of weight $\leq\delta_{\text{in}}$ consistent with the syndrome. If there is a unique such error $E$ (up to multiplication by stabilizers), the algorithm applies $E^\dagger$ to the measured state to correct the error; otherwise, it fails and gives an arbitrary output. Let $\Dec_{Q_{\text{out}}}$ denote the unique-decoding algorithm for $Q_{\text{out}}$ provided by the theorem statement. Then we let $\Dec_{Q^{\vec{\kappa}}}=\Dec_{Q_{\text{out}}}\circ\bigotimes_{i=1}^n\Dec_{\text{in}}^{\kappa_i}\circ\pi_G^{-1}$.
  
  As $Q^{\vec{\kappa}}$ is the concatenation of stabilizer codes, the decoding algorithm described above can be implemented as a single syndrome measurement followed by a classical algorithm that computes a Pauli error, whose inverse can then be applied to the measured state to restore the initial code state with high probability. However, for the purpose of our analysis, it is helpful to assume that the inner code syndromes are measured before the outer code syndrome, as then we can apply \Cref{lemma:private-aqec} (and its composable generalization in \Cref{theorem:composabilityaqeccs}) to the inner codes.
  
  Each of the inner decodings $\Dec_{\text{in}}^{\kappa_i}$ runs in time $O(q^2)$ because the brute force search may simply iterate through all possible Pauli errors. Thus the overall runtime of $\Dec_Q$ is $f(n)+O(nq^2)$. 
  
  It remains to be shown that $Q$ is a private $(\delta,2^{-\Omega(\delta_{\text{out}}n)})$-AQECC for $\delta=(1-R-\gamma)/2$. For this purpose, we will show that for every error channel $\mathcal{A}$ acting on some subset $T\subseteq[n]$ of the code components with $|T|\leq\delta n$, it holds with probability $2^{-\Omega(\delta_{\text{out}}n)}$ over the uniformly random choice of key $\vec{\kappa}\in K^n$ and over the randomness from measurements in the decoding algorithm that for every message $\ket{\phi}$, we have 
  \begin{equation}
  \label{eq:hpdec}
    \Dec_{Q^{\vec{\kappa}}}(\mathcal{A}(\Enc_{Q^{\vec{\kappa}}}(\ket{\phi})))=\ket{\phi}.
  \end{equation}
  As a point of notation, above we are viewing the channel $\Dec_{Q^{\vec{\kappa}}}\circ\mathcal{A}\circ\Enc_{Q^{\vec{\kappa}}}$ as a procedure that outputs a classical ensemble of pure quantum states. Thus in the notation of density matrices, we would equivalently like to show that
  \begin{equation*}
  \Dec_{Q^{\vec{\kappa}}} \circ \mathcal{A} \circ \Enc_{Q^{\vec{\kappa}}}(\ket{\phi}\bra{\phi}) = (1-p_{\text{bad}})\ket{\phi}\bra{\phi}+p_{\text{bad}}\rho_{\text{bad}}
  \end{equation*}
  for some $0\leq p_{\text{bad}}\leq 2^{-\Omega(\delta_{\text{out}}n)}$ and some density matrix $\rho_{\text{bad}}$.
  
  Note that we allow the message $\ket{\phi}$ to be an arbitrary pure state on the joint system consisting of the message register and an arbitrarily large side register, to which none of the encoding, decoding, or error channels have access. Allowing $\ket{\phi}$ to be entangled with a side register in this way implies that our proof gives the desired diamond norm bound in \Cref{def:privateaqecc}. However, the side register effictively has no effect in the argument below, so we will not need to reference it further.
  
  Letting $\ket{\psi}=\pi_G^{-1}\Enc_{Q^{\vec{\kappa}}}(\ket{\phi})$ and $\mathcal{A}_G=\pi_G^{-1}\mathcal{A}\pi_G$, then~(\ref{eq:hpdec}) is equivalent to
  \begin{equation*}
      \Dec_{Q_{\text{out}}}\left(\bigotimes_{i=1}^n\Dec_{\text{in}}^{\kappa_i}\right)\mathcal{A}_G(\ket{\psi}) = \ket{\phi}.
  \end{equation*}
  
  Recall that we may decompose 
  \begin{equation*}
      \mathcal{A}_G(\ket{\psi}) = \sum_{\nu}A_\nu\ket{\psi}\otimes\ket{\nu},
  \end{equation*}
  where $\nu$ ranges over a set of orthonormal environment states, and the $A_\nu$ are some operators such that $\sum_\nu A_\nu^\dagger A_\nu=I$. Then it suffices to show for each $\nu$, with probability $\geq 1-2^{-\Omega(\delta_{\text{out}}n)}$ we have 
  \begin{equation}
  \label{eq:decgoal}
      \Dec_{Q_{\text{out}}}\left(\bigotimes_{i=1}^n\Dec_{\text{in}}^{\kappa_i}\right)A_\nu\ket{\psi} = \ket{\phi}.
  \end{equation}
  
  Define $S\subseteq[n]\times[n_{\text{in}}]$ by $S=\pi_G^{-1}(T\times[n_{\text{in}}])$. Because $\delta_{\text{in}}-\gamma_0/\sqrt{\delta_{\text{out}}/100}<\delta$, \Cref{lem:expperm} implies that the set
  \begin{equation*} 
      S' = \{i\in[n]:|S\cap(\{i\}\times[n_{\text{in}}])|\geq\delta_{\text{in}}\}.
  \end{equation*}
  has $|S'|/n\leq\delta_{\text{out}}/100$. That is, the $\geq 1-\delta_{\text{out}}/100$ fraction of inner blocks with $i\notin S'$ have error rate from $\mathcal{A}_G$ below the list decoding radius $\delta_{\text{in}}$ of $Q_{\text{in}}$.
  
  Fix $\nu$. Let $\mathcal{P}_S\subseteq\mathcal{P}_q^n$ denote the set of Paulis supported inside $S$. Then we may decompose $A_\nu=\sum_{E\in\mathcal{P}_S}a_EE$ for some coefficeints $a_e$.
  
  As the left hand side of \Cref{eq:decgoal} does not depend on the order in which the inner blocks are decoded, assume that we first sequentially apply $\tDec_{\text{in}}^{\kappa_i}$ to the inner code blocks $i\notin S'$ of $\mathcal{A}_G(\ket{\psi})$. Then for a given such $i\notin S'$, $\tDec_{\text{in}}^{\kappa_i}$ is applied to the $i$th block of $\left(\bigotimes_{j\notin S':j<i}\tDec_{\text{in}}^{\kappa_j}\right)A_\nu(\ket{\psi})$. We may let 
  \begin{equation*}
    A_\nu^{(i)}\ket{\psi} = \left(\bigotimes_{j\notin S':j\leq i}\tDec_{\text{in}}^{\kappa_j}\right)A_\nu(\ket{\psi}) = \sum_{E\in\mathcal{P}_S}a_E^{(i)}E\ket{\psi}    
  \end{equation*}
  for some coefficients $a_E^{(i)}$, which are random variables depending on the choices of $\kappa_j$ as well as on the outcomes of measurements in $\tDec_{\text{in}}^{\kappa_j}$ for $j\notin S'$ with $j\leq i$. Note that for a fixed $i$, the random variables $a_E^{(i)}$ over $E\in\mathcal{P}_S$ may be correlated.
  
  
  
  For every $i\notin S'$, since $\mathcal{A}_G$ and therefore $A_\nu$ acts on fewer than $\delta_{\text{in}}$ fraction of the components in block $i$, local indistinguishability guarantees (\Cref{claim:LI-composable-draft}) that the distribution of the coefficients $a_E^{(i-1)}$ does not depend on the message encoded by $Q_{\text{in}}$ in block $i$. Thus in particular, the distribution of $a_E^{(i-1)}$ does not depend on the choice of $\kappa_i\in K$.
  
  For an inner code block $i\in[n]$, if $A_\nu^{(j)}=\sum_{E\in\mathcal{P}_S}a_E^{(j)} E$ is such that $a_E^{(j)}=0$ for every $E$ with $i$th block $E_i$ not equal to the identity, we say that block $i$ is \textit{error-free} in $A_\nu^{(j)}$. Recall the definition of $\tDec_{\text{in}}^{\kappa}$ in \Cref{alg:algorithm1}. By \Cref{lemma:private-aqec} and its composable generalization in \Cref{theorem:composabilityaqeccs}, for each $i\notin S'$, conditioned on the values of $a_E^{(i-1)}$, it holds with probability $\geq 1-\delta_{\text{out}}/10$ over the randomness in $\kappa_i$ and in the syndrome measurement that $\tDec_{\text{in}}^{\kappa_{i}}$ successfully corrects the error on inner code block $i$, meaning that block $i$ is error-free in $A_\nu^{(i)}$. Furthermore, by \Cref{remark:envreg}, if block $i$ is error-free in $A_{\nu^{(i)}}$, then block $i$ is error-free in $A_\nu^{(j)}$ for all $j\geq i$, so that in particular block $i$ is error-free in $A_\nu^{(n)}$. It follows that
  \begin{align*}
      \Pr[A_\nu^{(n)}\text{ has }\geq\delta_{\text{out}}n\text{ blocks }i\notin S'\text{ that are not error-free}]
      &\leq \Pr[\operatorname{Bin}(n-|S'|,\delta_{\text{out}}/10)+|S'|\geq\delta_{\text{out}}n] \\
      &\leq e^{-\delta_{\text{out}}n/2},
  \end{align*}
  where the final inequality above follows from the Chernoff bound. Thus because by definition $\ket{\psi}=(\bigotimes_{i=1}^n\Enc_{Q_{\text{in}}\circ P_{\text{in}}^{\kappa_i}})\Enc_{Q_{\text{out}}}\ket{\phi}$, it follows that with probability $\geq 1-e^{-\delta_{\text{out}}n/2}$, the state
  \begin{equation*}
      \left(\bigotimes_{i=1}^n\Dec_{\text{in}}^{\kappa_i}\right)A_\nu\ket{\psi} = \left(\bigotimes_{i=1}^n\Enc_{Q_{\text{in}}\circ P_{\text{in}}^{\kappa_i}}^\dagger\right)\left(\bigotimes_{i\in S'}\tDec_{\text{in}}^{\kappa_i}\right)A_\nu^{(n)}\ket{\psi}
  \end{equation*}
  is equal to $A'\Enc_{Q_\text{out}}\ket{\phi}$ for some operator $A'$ that acts as the identity on $>1-\delta_{\text{out}}$ fraction of the components $i\in[n]$. In this case we are guaranteed that $\Dec_{Q_{\text{out}}}A'\Enc_{Q_{\text{out}}}\ket{\phi}=\ket{\phi}$, so \Cref{eq:decgoal} holds with probability $\geq 1-e^{-\delta_{\text{out}}n/2}$, as desired.

\end{proof}

The following corollary instantiates \Cref{thm:uniqueAQECC} with two examples of suitable outer codes, namely distance-amplified concatenated Reed-Solomon codes and asymptotically good quantum LDPC codes.

\begin{corollary}
  \label{cor:outercodes}
For every $\gamma>0$ and $0<R<1$, we have:
\begin{enumerate}
    \item There exists an infinite explicit family of private $((1-R-\gamma)/2,2^{-\Omega(\gamma n)})$-AQECCs of rate $R$, alphabet size $q=2^{O(\frac{1}{\gamma^4}\log\frac{1}{\gamma})}$, block length $n\geq q$, private key length $O(\gamma^2\cdot n\log q)$ bits, and decoding time $O(n\log q)^3$.
    \item For sufficiently small $\delta_{\text{out}}=\delta_{\text{out}}(\gamma)>0$, there exists an infinite explicit family of private LDPC $((1-R-\gamma)/2,2^{-\Omega(\delta_{\text{out}}n)})$-AQECCs $Q$ of rate $R$, alphabet size $q=2^{O(1/\gamma^3\delta_{\text{out}})}$, private key length $O(\gamma^2\cdot n\log q)$ bits, and decoding time $O_{\gamma}(n)$.
\end{enumerate}
\end{corollary}

For the first item in \Cref{cor:outercodes}, we will concatenate a Reed-Solomon outer code with an inner code from the ensemble of codes given below.

\begin{lemma}
  \label{lem:qwozencraft}
  For a prime power $q$ and positive integers $r,s$, let $C_1,C_2\subseteq\mathbb{F}_{q^r}^s$ be $(s-1)$-dimensional subspaces chosen uniformly at random subject to the constraint that $C_1^\perp\subseteq C_2$. Let $Q=\text{CSS}_{\mathbb{F}_q}(C_1,C_2)$ be the CSS code of block length $n=rs$, rate $R=1-2/s$, and local dimension $\mathbb{F}_q$ obtained by associating $\mathbb{F}_{q^r}\cong\mathbb{F}_q^r$. Then for every $\gamma>0$, it holds with probability $\geq 1-2q^{-\gamma n}$ that $Q$ has relative distance at least $H_q^{-1}(1/s-\gamma)$, where $H_q$ denotes the $q$-ary entropy function.
\end{lemma}
\begin{proof}
  By definition for $i=1,2$, the code $C_i$ contains any given $x\in\mathbb{F}_{q^r}^s\cong\mathbb{F}_q^{rs}$ with probability $(q^{r(s-1)}-1)/(q^{rs}-1)\leq q^{-r}$. Therefore the chance that $C_i$ contains some $x$ of Hamming weight $\leq H_q^{-1}(1/s-\gamma)$ is at most $q^{-r}\cdot|B_0(H_q^{-1}(1/s-\gamma))|\leq q^{-r}\cdot q^{rs(1/s-\gamma)}=q^{-rs\gamma}$. Union bounding over $i=1,2$, it follows that both $C_1,C_2$, and therefore $Q$, has relative distance at least $H_q^{-1}(1/s-\gamma)$ with probability $\geq 1-2q^{-rs\gamma}$.
\end{proof}

The codes in \Cref{lem:qwozencraft} may be viewed as a quantum analogue of the Wozencraft ensemble, as they provide a derandomization of random CSS codes that still approach the quantum Singleton bound over large alphabets.

\begin{lemma}
  \label{lem:RSwoz}
  For every $\gamma>0$, there exists an explicit infinite family of $[[n,k]]_q$ codes $Q$ of rate $R=k/n\geq 1-\gamma$, local dimension $q=2^{O(\frac{1}{\gamma^3}\log\frac{1}{\gamma})}$, and block length $n\geq q$ that are unique-decodable from errors acting on $\Omega(\gamma)$ fraction of the components in time $O(n\log q)^3$.
\end{lemma}
\begin{proof}
  We construct $Q$ by applying the AEL construction in \Cref{prop:aelred} to the concatenation of an outer Reed-Solomon code of rate $1-\gamma/2$ with an inner code of rate $1-\gamma/2$ from the ensemble in \Cref{lem:qwozencraft}. Specifically, let $Q_{\text{in}}=\text{CSS}_{\mathbb{F}_{q_{\text{in}}}}(C_1,C_2)$ be a $[[n_{\text{in}},k_{\text{in}},\Delta_{\text{in}}n_{\text{in}}]]_{q_{\text{in}}}$ code from the ensemble in \Cref{lem:qwozencraft} of rate $R_{\text{in}}\geq 1-\gamma/2$, local dimension $q_{\text{in}}=1/\gamma$, and relative distance $\Delta_{\text{in}}\geq H_{q_{\text{in}}}^{-1}(\gamma/8)\geq\Omega(\gamma)$. As described in \Cref{lem:qwozencraft}, each of $C_1,C_2$ is specified by a single parity check vector in $\mathbb{F}_{q_{\text{in}}^r}^s$ up to scalar multiplication by $\mathbb{F}_{q_{\text{in}}^r}$. Therefore there are at most $2q_{\text{in}}^{k_{\text{in}}}$ possible choices for each of $C_1,C_2$, so there are at most $(2q_{\text{in}}^{k_{\text{in}}})^2$ possible choices for $Q$. Computing the distance of each of $C_1,C_2$ by brute force takes $O(q_{\text{in}}^{k_{\text{in}}})$ time, so a brute force search to find $Q_{\text{in}}$ takes $O(q_{\text{in}}^{3k_{\text{in}}})$ time.

  Now let $Q_{\text{out}}$ be a $[[n_{\text{out}},k_{\text{out}}]]_{q_{\text{out}}}$ quantum Reed-Solomon code of rate $R_{\text{out}}=1-\gamma/2$, and block length and local dimension $n_{\text{out}}=q_{\text{out}}=q_{\text{in}}^{k_{\text{in}}}$. Then \Cref{prop:aelred} implies that we may choose $\gamma_0=\Theta(\gamma^{3/2})$, $r=4/\gamma_0^2$, and $G$ an $r$-regular $\gamma_0$-pseudorandom bipartite graph such that the $[[n,k]]_q$ code $Q=\Pi_G(Q_{\text{out}}\circ Q_{\text{in}})$ has local dimension $q=q_{\text{in}}^r=2^{O(\frac{1}{\gamma^3}\log\frac{1}{\gamma})}$ and is decodable from adversarial errors acting on $\Omega(\gamma)$ fraction of the components. Furthermore, the decoding algorithm simply decodes the inner codes and then decodes the outer code. Brute force decoding in all inner code blocks takes time $O(n_{\text{out}})^3\leq O(n^3)$, and decoding the Reed-Solomon outer code can be done in time $O(n_{\text{out}}\log q_{\text{out}})^3\leq O(n\log q)^3$. 
  Thus the overall decoding algorithm takes $O(n\log q)^3$ time as desired.

  Furthermore, $Q$ is explicit as we showed above that the inner code $Q_{\text{in}}$ can be constructed by a brute force search in time $O(q_{\text{out}}^3)\leq O(n^3)$, and the quantum Reed-Solomon code $Q_{\text{out}}$ is explicit.
\end{proof}

\begin{proof}[Proof of \Cref{cor:outercodes}]
  \begin{enumerate}
  \item \label{it:RW} The result follows from \Cref{thm:uniqueAQECC} by letting $Q_0$ be a code of rate $1-\gamma/4$ and alphabet size $q_0=2^{O(\frac{1}{\gamma^3}\log\frac{1}{\gamma})}$ that is unique-decodable from errors acting on $\delta_{\text{out}}=\Omega(\gamma)$ fraction of the components in time $O(n_0\log q_0)^3$, as given in \Cref{lem:RSwoz}.
  \item The result follows form \Cref{thm:uniqueAQECC} by letting $Q_0$ be a family of asymptotically good quantum LDPC codes of local dimension $q_0=2$ that permits linear-time decoding. For instance, the codes of \cite{panteleev2022asymptotically,leverrier2022quantum} have such a decoding algorithm \cite{leverrier2022efficient}.
  \end{enumerate}
\end{proof}

As described in \Cref{subsec:RSS}, the private key of a private AQECC can be incorporated into the code using a classical robust secret sharing scheme. Specifically, \Cref{theorem:rss-construction} and \Cref{theorem:removing-privacy} imply that the statement of \Cref{thm:uniqueAQECC} still holds when we remove the condition that the AQECC is private.

\begin{corollary}
  Define $\gamma,R,\delta_{\text{out}},Q_0,f$ as in \Cref{thm:uniqueAQECC}, and assume that $f(n)=\poly(n)$ so that $Q_0$ permits efficient decoding. Then there exists an infinite explicit family of $((1-R-\gamma)/2,2^{-\Omega(\delta_{\text{out}}n)})$-AQECCs $Q$ of rate $R$, alphabet size $q=\max\{q_0,2^{1/\gamma^2\delta_{\text{out}}}\}^{O(1/\gamma)}$ that has efficient encoding and decoding algorithms.
\end{corollary}
\begin{proof}
  The result follows directly by applying the private AQECCs of \Cref{thm:uniqueAQECC} with the robust secret sharing schemes of \Cref{theorem:rss-construction} in \Cref{theorem:removing-privacy}. Specifically, as in the proof of \Cref{thm:uniqueAQECC}, we may assume that $\gamma/2\leq R\leq 1-\gamma$. Let $Q_{\text{priv}}$ be the private $((1-R_{\text{priv}}-\gamma_{\text{priv}})/2,2^{-\Omega(\delta_{\text{out}}n)})$-AQECC from \Cref{thm:uniqueAQECC} of rate $R_{\text{priv}}=R+\gamma/2$, block length $n$, local dimension $q_{\text{priv}}=\max\{q_0,2^{1/\gamma^2\delta_{\text{out}}}\}^{\Theta(1/\gamma)}$, and with $\gamma_{\text{priv}}=c\gamma$ for a sufficiently small constant $0<c<1/2$ such that $Q_{\text{priv}}$ has private key length at most $\gamma^2/100\cdot n\log q_{\text{priv}}$. Let $\RSS$ be the $((1-R_{\text{RSS}}-\gamma_{\text{RSS}})n/2,2^{-n})$-robust secret sharing scheme from \Cref{theorem:rss-construction} of rate $R_{\text{RSS}}=\gamma/2$, block length (number shares) $n$, share size $a_{\text{RSS}}=q_{\text{priv}}^{\gamma/5}$, and with $\gamma_{\text{RSS}}=\gamma/2$. Note that we may choose $a_{\text{RSS}}=q_{\text{priv}}^{\gamma/5}$ because the quantum Singleton bound implies that $\delta_{\text{out}}\leq\gamma/8$, so $q_{\text{priv}}\geq 2^{\Omega(1/\gamma^4)}$, and thus $q_{\text{priv}}^{\gamma/5}\geq 2^{\Omega(1/\gamma^3)}\geq 2^{\tilde{O}(1/\gamma^2)}$ is at least as large as the share size given by \Cref{theorem:rss-construction}. 
  
  Then applying \Cref{theorem:removing-privacy} with $Q_{\text{priv}}$ and $\RSS$ yields an AQECC $Q$ with the desired parameters. Specifically, because by construction the number of bits in the RSS message is $nR_{\text{RSS}}\log a_{\text{RSS}}=\gamma^2/10\cdot n\log q_{\text{priv}}$, which is at least as large as the private key length of $Q_{\text{priv}}$, we may indeed apply \Cref{theorem:removing-privacy} to obtain a $((1-R-\gamma)/2,2^{-\Omega(\delta_{\text{out}}n)})$-AQECC of rate $R_{\text{priv}}/(1+\gamma/5)\geq R$ and alphabet size $q=q_{\text{priv}}^{1+\gamma/5}\leq\max\{q_0,2^{1/\gamma^2\delta_{\text{out}}}\}^{O(1/\gamma)}$.
\end{proof}

Thus in particular we immediately obtain the following non-private analogue of \Cref{it:RW} of \Cref{cor:outercodes}, again by letting $Q_0$ be the code from \Cref{lem:RSwoz}.

\begin{corollary}\label{corollary:aqeccswoLR}
  For every $\gamma>0$ and $0<R<1$, there exists an infinite explicit family of $((1-R-\gamma)/2,2^{-\Omega(\gamma n)})$-AQECCs of rate $R$, alphabet size $q=2^{O(\frac{1}{\gamma^4}\log\frac{1}{\gamma})}$ and block length $n\geq q$ that has efficient encoding and decoding algorithms.
\end{corollary}

\section*{Acknowledgements}

We thank Venkatesan Guruswami, Andrea Coladangelo, Omar Alrabiah, Debbie Leung, James Bartusek and Umesh Vazirani for fruitful discussions. We also thank Fermi Ma, Markus Grassl, and Andreas Winter for helpful comments on earlier drafts of the paper.

TB and LG acknowledge support by the National Science Foundation Graduate Research Fellowship under Grant No. DGE 2146752.

\newpage

\nocite{*}
\bibliographystyle{alphaurl}
\bibliography{references}

\newpage

\appendix

\newpage
\section{Properties of Random CSS Codes}
\label{section:randomCSSproperties}

In this section, we briefly discuss the list decodability and other properties of random CSS codes, which we apply in for the construction of concatenated codes. These properties follow from straightforward applications of their well known classical counterparts, and so for conciseness we simply refer to the necessary statements. That is, we construct random CSS codes from correlated random classical linear codes. Therefore, by a union bound, any property that holds with high probability for these classical codes also holds with high probability for the CSS code.

We generate a random $[[n, k]]_q$ CSS code as follows. First, sample  $k_1 = (n+k)/2$ random linearly independent vectors $g_1,\cdots g_{k_1}$ in $\mathbb{F}_q^n$. Let $G_1 = (g_1\cdots g_{k_1}) \in \mathbb{F}_q^{n\times k_1}$ be the matrix defined by the sampled vectors and $G_2 = (g_1\cdots g_{k_2})$ to be the first $k_2 = n-k_1 $ columns of $G_1$. We define $C_1$ to be the classical linear code defined by using $G_1$ as a generator matrix, and $C_2$ to be the classical linear code defined by using $G_2^T$ as the parity check matrix, such that $C_2^\perp\subset C_1$ by construction. We now consider the random Galois-Qudit CSS code $Q_R = $ CSS$(C_1, C_2)$, where $R$ is the rate $k/n$.

First, by assumption with high probability both $G_2, G_1$ are full rank. Thereby, we have $C_2, C_1$ are $[n, k_1]_q$ classical codes, and CSS$(C_1, C_2)$ is a $[[n, k]]_q$ QECC. A standard exercise in coding theory is to prove random generator matrices (and parity check matrices) of $k$ columns (resp. $n-k$ rows) have distances approaching the Gilbert-Varshamov bound of $n\cdot H_q^{-1}(1-\frac{k}{n})$ with exponentially high probability, where $H_q$ denotes the $q$-ary entropy function. A similar statement also holds for list-decodability, as shown in \cite[Theorem 5.3]{Guruswami2001ListDO} and stated below.

\begin{lemma}[Well known]
\label{lem:randomld}
Let $C\subseteq\mathbb{F}_q^n$ be a random classical linear code of rate $R$. Then for every $0<\delta<H_q^{-1}(1-R)$,
\begin{itemize}
    \item $C$ has relative distance $\geq\delta$ with probability $\geq 1-q^{-n(1-R-H_q(\delta))}$
    \item For every $\ell\geq 1$, $C$ is $(\delta,\ell)$ list-decodable with probability $\geq 1-q^{n(1-\log_q(\ell+1)(1-R-H_q(\delta)))}$.
\end{itemize}
\end{lemma}

In particular, \Cref{lem:randomld} implies that given $\epsilon>0$, we can choose $q=2^{O(1/\epsilon)}$, $\delta=1-R-O(\epsilon)$, and $\ell=q^{O(1/\epsilon)}=2^{O(1/\epsilon^2)}$, and the random linear code $C$ will be $(\delta,\ell)$ list-decodable and have relative distance $\geq\delta$ with probability $\geq 1-q^{-\Omega(\epsilon n)}$. Union bounding over the two classical codes making up a random CSS code, \Cref{lem:randomld} and \Cref{thm:CSS} immediately imply the following.

\begin{claim} \label{claim:randomcss}
    Let $Q$ be a random CSS code of rate $R$ and block length $n$ over $\mathbb{F}_q$. Then for every $0<\delta<H_q^{-1}((1-R)/2)$,
    \begin{itemize}
        \item $Q$ has relative distance $\geq\delta$ with probability $\geq 1-2\cdot q^{-n((1-R)/2-H_q(\delta))}$
        \item For every $\ell\geq 1$, $Q$ is $(\delta,\ell)$ list-decodable with probability $\geq 1-2\cdot q^{n(1-\log_q(\sqrt{\ell})((1-R)/2-H_q(\delta)))}$.
    \end{itemize}
\end{claim}

Similarly as in the classical case, given $\epsilon>0$, we can choose $q=2^{O(1/\epsilon)}$, $\delta=(1-R-O(\epsilon))/2$, and $\ell=q^{O(1/\epsilon)}=2^{O(1/\epsilon^2)}$, and the random CSS code $Q$ will be $(\delta,\ell)$ list-decodable and have relative distance $\geq\delta$ with probability $\geq 1-q^{-\Omega(\epsilon n)}$.




\section{List-Recoverable Quantum Codes}
\label{sec:listrecoverablequantumcodes}

In this section, we discuss a quantum analog of a certain extension to list-decoding, the problem of list-recovery. While, at a high level, list-decoding requests the codewords of a certain code $C$ around a given point, the problem of list-recovery requests the codewords $c\in C$ whose individual symbols $c_1, c_2...c_n$ come from a small set of `candidate' symbols:

\begin{definition}
\label{def:clr}
    A code $C\subset \Sigma^n$ is $(\eta, l, L)$-LR (list-recoverable) if $\forall S_1\cdots S_n\subset \Sigma$, with $|S_i|\leq l$, 
    \begin{equation}
        \bigg|\bigg\{ c\in  C: c_i\in S_i \text{ for at least }\eta \cdot n \text{ symbols }i\in [n]\bigg\}\bigg| \leq L 
    \end{equation}
\end{definition}

This would suggest that once one fixes sets of candidate symbols $S_1\cdots S_n$, there aren't many codewords $c\in C$ which agree with the subsets $S_i\subset \Sigma$ at most $i$. We emphasize that the special case $(\eta, 1, L)$-LR is equivalent to $(1-\eta, L)$-LD. This is since if $C$ is $(\eta, 1, L)$-LR, for any vector $y\in \Sigma^n$, one can pick $S_i = \{y_i\}$ (such that $|S_i|=1$) and the number of codewords $c\in C$ which agree with $y$ on at least $\eta$ locations is the number of codewords in a ball of radius $1-\eta$ around $y$. Finally, $C$ is said to be \textit{efficiently} list-recoverable if there exists an efficient algorithm to find these codewords.

Below, we discuss an analogous notion of list recoverability for the normalizer group of a stabilizer QECC. Informally, it corresponds to enumerating the logical operators of the quantum code whose single-qudit components are selected from candidate lists $\mathcal{E}_i\subset \mathcal{P}_q$ of single-qudit operators. 

\begin{definition}
    A $[[n, k]]_q$ stabilizer code $Q$ with stabilizer group $\mathcal{S}(Q)$ is $(\eta, l, L)$-QLR (Quantum List Recoverable) if, for all subsets $\mathcal{E}_1, \cdots, \mathcal{E}_n \subset \mathcal{P}_q$ of single-qudit operators with $|\mathcal{E}_i|\leq l$, any list $L(\{\mathcal{E}\})$ of logically-distinct operators in the normalizer group $N(Q)$ where
    \begin{equation}
        L(\{\mathcal{E}\}) = \bigg\{ A = \otimes_i A_i \in N(Q): A_i \in \mathcal{E}_i \text{ for at least } \eta\cdot n \text{ symbols }i\bigg\}, 
    \end{equation}

    \noindent has size bounded by $|L(\{\mathcal{E}\})|\leq L$.
\end{definition}

This means that there don't exist many logical operators which match the $\mathcal{E}_i$ at most locations $i \in [n]$. We expect this to naturally reduces to the QLD definition when $l= 1$. Indeed,

\begin{claim}
    If a stabilizer QECC $Q$ is $(\eta, 1, L)$-QLR, then it is $(1-\eta, L)$-QLD.
\end{claim}

\begin{proof}
    Fix $E\in \mathcal{P}_q^n$. Recall the set of operators $F\in \mathcal{P}_{q}^n$ such that $F^\dagger E\in N(Q)$ is a logical operator, is equivalent to the set $E A^\dagger$ for $A\in N(Q)$. Thus $F$ has weight $\leq \delta\cdot n$ if $A_i = E_i$ for at least $1-\delta$ locations. Let $\mathcal{E}_i = \{E_i\}$. Since $Q$ is $(\eta, 1, L)$-QLR, there are at most $L$ such choices of $A$.
\end{proof}

We note that the characterization above of quantum list-recovery can be cast in terms of an analogous classical question about syndromes of the stabilizer code.

\begin{claim}
    Let $Q$ be $(\eta, l, L)$-QLR, let $E\in \mathcal{P}_q^n$ be any error on $Q$ of syndrome $s$, and consider any set of subsets $\mathcal{E}_1, \cdots, \mathcal{E}_n \subset \mathcal{P}_q$ of single-qudit operators with $|\mathcal{E}_i|\leq l$. Then any set $\mathcal{L}_s(\{\mathcal{E}\})\subset \mathcal{P}_q^n$ of logically-distinct operators of syndrome $s$ and which agree with the $\mathcal{E}_i$ on at least $\eta\cdot n$ qudits $i$, has size bounded by $|\mathcal{L}_s(\{\mathcal{E}\})|\leq L$.
\end{claim}

\begin{proof}
    Given any such set $\mathcal{L}_s(\{\mathcal{E}\})$, consider the set $\mathcal{L} = \big\{ E^\dagger F: \text{ for } F\in \mathcal{L}_s(\{\mathcal{E}\})\big\}$ defined by shifting the elements of the former by $E$. Also, for each $i\in [n]$, `shift' each list of single qudit operators $\mathcal{E}_i' = (E^\dagger)_i\mathcal{E}_i$ by $E_i^\dagger$. Note that $|\mathcal{E}_i'|\leq l$, and $|\mathcal{L}_s(\{\mathcal{E}\})| = |\mathcal{L}|$. Since each $F\in \mathcal{L}_s(\{\mathcal{E}\})$ has the same syndrome as $E$, $E^\dagger F\in N(Q)$ and thus the elements of $\mathcal{L}$ are logically-distinct operators in $N(Q)$. Moreover, $(E^\dagger F)_i \in \mathcal{E}_i'$ for most $i$. By the definition of QLR, we must have $|\mathcal{L}|\leq L$, which concludes the proof. 
\end{proof}

Much like the case for quantum list decoding, we show that one can devise quantum codes with good list recovery properties from the CSS construction of two classically list recoverable codes. 

\begin{claim}
    If $C_1, C_2\subset (\mathbb{F}_q^m)^n$ are two $\mathbb{F}_q$-linear classical codes with $C_2^\perp \subset C_1$ and $C_1, C_2$ are $(\eta, l, L)$-LR, then the Galois-qudit CSS$(C_1, C_2)$ is $(\eta, l,L^2)$-QLR.
\end{claim}

\begin{proof}
    The elements in the normalizer of CSS$(C_1, C_2)$ are the operators $E_{\textbf{a}, \textbf{b}}$, $a\in C_1, b\in C_2$. Now, fix arbitrary sets $\mathcal{E}_1, \cdots, \mathcal{E}_n \subset \mathcal{P}_q^m$ of size $\leq l$. We can write each single qudit operator $E_{ij}$ in the $i$th set $\mathcal{E}_i = \{E_{ij}\}_{j\in [l]}$ into a Pauli basis: $E_{ij} = E_{a_{ij}, b_{ij}}, a_{ij}, b_{ij}\in (\mathbb{F}_q^m)$. We can thereby define sets of `X' symbols $S^X_i=  \{a_{ij}\}_{j\in [l]}$, and `Z' symbols $S^Z_i=  \{b_{ij}\}_{j\in [l]}$, for each location $i$.
    
    If $A = E_{\textbf{a}, \textbf{b}}\in N(S)$, and $A_i\in \mathcal{E}_i$ for most $i$, then we have that $a_i \in S^X_i$ and $b_i \in S^Z_i$ for most $i$. However, $a\in C_1$ is $(\eta, l, L)$-LR, so there exist at most $L$ such valid codewords of $C_1$. Analogously, there exist at most $L$ valid choices of $b\in C_2$. The set of of valid elements $A\in N(S)$ is at most the number of pairs of valid elements.
\end{proof}

By casting this notion of list recovery as a classical problem, we inherit the efficiency of classical list decoding algorithms.

\begin{claim}
    If $C_1, C_2\subset (\mathbb{F}_q^m)^n$ with $C_2^\perp\subset C_1$ are efficiently $(\eta, l, L)$-LR, then given any syndrome $s$ of the quantum code $Q=$CSS$(C_1, C_2)$ and subsets $\mathcal{E}_1, \cdots, \mathcal{E}_n \subset \mathcal{P}_q$ of single-qudit operators with $|\mathcal{E}_i|\leq l$, there exists a polynomial time classical algorithm to find a list of logically-distinct operators $F$ of syndrome $s$ and $F_i\in \mathcal{E}_i$ for at least $\eta\cdot n$ symbols $i\in [n]$.
\end{claim}

\begin{proof}
     If $s=s_x, s_z$ are the syndrome measurements of the X and Z checks, respectively, then let $E' = E_{e_x, e_z}$ be any Pauli operator with the same syndrome. Given the sets of single qudit operators $\mathcal{E}_i= \{E_{a_{ij}, b_{ij}}\}_{j\in [l]}$, then consider the sets of $q$-ary symbols
     \begin{equation}
         S^X_i = \{(e_x)_i+a_{ij}\}_{j\in  [l]} \text{ and } S^Z_i = \{ (e_z)_i+b_{ij}\}_{j\in  [l]} \text{ for }i\in [n]
     \end{equation}

     We call the classical list decoding algorithm for $C_1$ on the $S^X_1, S^X_2\cdots $ and $C_2$ on the $S^Z_i$, obtaining lists of codewords $(c_{1, 1}, c_{1,  2}\cdots c_{1,  L})\subset C_1$ and $(c_{2, 1}, c_{2,  2}\cdots c_{2,  L})\subset C_2$ respectively. We note that each `correction' $c_{1, k}-e_x$, for $k\in [L]$, defines a Pauli operator $X^{e_x-c_{1, k}}$ with the same syndrome as $e_x$. Moreover, $(e_x-c_{1, k})_i$ agrees with $\{a_{ij}\}_{l\in [n]}$ in at least $\eta\cdot n$ locations $i$. Thus a subset of the $L^2$ operators $E_{\textbf{e}_x-\textbf{c}_{1, k}, \textbf{e}_z-\textbf{c}_{2, k'}}, k, k'\in [L]$ satisfy the claim.
\end{proof}

By combining the above with known results on the list-recoverability of folded Reed-Solomon codes (the below is a corollary of \Cref{thm:frslr} of \cite{Guruswami2008ExplicitCA} stated previously), 

\begin{theorem}[\cite{Guruswami2008ExplicitCA}]
    For every integer $l$, rates $R\in (0,1)$, $\gamma \in (0, R)$, and prime $p$, there exists $m$-folded RS codes (\cref{def:FRS}) over fields of characteristic $p$ of rate $R$ and block length $n$ which are $(R+\gamma,l, L)$-LR, where $m = O(\frac{\log l}{(1-R)\gamma^2})$, $L = O(n/\gamma^2)^{O(\gamma^{-1}\log l/R)}$, and alphabet size $(n/\gamma^2)^{O(m)}$.
\end{theorem}

We arrive at an exact analog:

\begin{corollary}
For every integer $l$, rates $R\in (0,1)$, $\gamma \in (0, R)$, and prime $p$, there exists $m$-folded QRS codes (\cref{def:FQRS}) over fields of characteristic $p$ of rate $R$ and block length $n$ which are $(\frac{1+R}{2}+\gamma,l, L)$-QLR, where $m = O(\frac{\log l}{(1-R)\gamma^2})$, $L = O(n/\gamma^2)^{O(\gamma^{-1}\log l)}$, and local dimension $(n/\gamma^2)^{O(m)}$.
\end{corollary}

\subsection{List Decodable Codes from List Recovery and Distance Amplification}

In this section, we discuss an application of a classical technique to devise codes over smaller alphabets with better list-decoding properties, using the distance amplification techniques in \cite{AEL95}. To do so, we discuss a quantum adaptation of the approach in \cite{Guruswami2008ExplicitCA}, who concatenated Folded RS codes with list decodable inner codes, to reduce the large alphabet size of the RS construction, and approach list decoding capacity on constant size alphabets (with efficient decoding). The theorem below is the main result of this section, on the construction of list decodable quantum codes up to the Singleton bound from the concatenation and distance amplification of list recoverable and list decodable quantum codes.

\begin{theorem}\label{theorem:listrecoveryconstruction}
    Fix $R, \gamma >  0$. There exists a $[[N, RN, \geq N\cdot (1-R-\gamma)/2]]_d$ stabilizer code $Q$ which is $((1-R-\gamma)/2, N^{1/\gamma^{O(1)}})$-QLD on local dimension $d = 2^{O(1/\gamma^5)}$. Moreover, $Q$ can be list-decoded efficiently.
\end{theorem}

\subsubsection{The Construction}

We build the quantum code in the theorem above using the concatenation $Q_\diamond = Q_{out}\diamond Q_{in}$ of two stabilizer codes $Q_{out}, Q_{in}$. Informally, the `outer' quantum code $Q_{out}$ is chosen to be a high rate, large alphabet, quantum list recoverable code, of size $n$ and local dimension $q^m$. In particular, a $[[n, (1-4\gamma) n]]_{q^m}$ QECC which is $(1-\gamma, l, L)$-QLR. In turn, the `inner' code $Q_{in}$ has rate $R$ near the target rate and is near list decoding capacity, i.e. with decoding radius approaching the quantum singleton bound, $(1-R)/2$. Formally, $Q_{in}$ is a $[[n', m]]_q$ QECC and $((1-R)/2 - 4\cdot \gamma, l)$-QLD, where $m=n'(R + 5\gamma)$. Their concatenation $Q_\diamond= Q_{out}\diamond Q_{in}$ is a code on $nn'$ qudits of local dimension $q$ and rate $(R + 5\gamma)(1-4\gamma) \geq R$

To conclude the construction, the physical qudits of the resulting concatenated code are permuted and re-grouped, using the AEL alphabet-reducing distance amplification techniques discussed in \cref{sec:ael}. Intuitively, as before this permutation is used to `spread out' errors on the inner code, such that any error on the final code doesn't corrupt most inner blocks beyond repair. Formally, we pick an $\gamma_0$-pseudorandom $D$-biregular bipartite graph $G$, with $\gamma_0 = \gamma^2$ and degree $D = 10/\gamma_0^2 = \text{poly}(1/\gamma)$, and use it to define a permutation $\pi_G$. The final code is defined by $\pi_G(Q_\diamond)$, and is defined on qudits of local dimension $q^D$. 

\subsubsection{Analysis}

We present a lemma on the distance amplification of generic list recoverable and list decodable quantum codes, and then instantiate the lemma with a particular choice of quantum codes $Q_{out},Q_{in}$ to prove \Cref{theorem:listrecoveryconstruction}. 

\begin{lemma} \label{lemma:distampf}
    Fix $R$ and sufficiently small $\gamma> 0$. Let $Q_{out}$ be a $[[n, (1-4\gamma) n]]_{q^m}$ QECC and $(1-\gamma, l, L)$-QLR, $Q_{in}$ be a $[[n', m = n'(R+5\gamma)]]_q$ QECC and $((1-R)/2 -4\gamma, l)$-QLD, and $G$ be a $\gamma^2$-pseudorandom bipartite expander graph of size $N=nn'/D$ and degree $D = O(1/\gamma^4)$. Then $\pi_G(Q_{out}\diamond Q_{in})$ is $((1-R)/2 - 8\cdot \gamma, L)$-QLD. 
\end{lemma}

If $E_G=\pi_G(E) \in (\mathcal{P}_q^{D})^{nn'/D}$ is an operator on the code $\pi_G(Q_{\diamond})$, then we refer to $\pi_G^{-1}(E_G) = E\in  (\mathcal{P}_q^{n'})^{n}$ as the `unpermuted' operator $E$ acting on the concatenated code $Q_\diamond$. We emphasize here that we restrict our attention to a basis of operators which factorizes as a tensor product of operators on the individual qudits of local dimension $q$. Since $Q_\diamond$ is a concatenated code, recall that its syndrome measurement corresponds to the syndromes of the inner codes, as well as the syndrome measurement of the encoding of the stabilizers of the outer code, see \cref{fact:concatstabilizers}.

To prove the lemma above, it suffices to show that there aren't many low weight, logically-distinct, operators $F_G$ on $\pi_G(Q_{\diamond})$, which agree with both of the syndrome measurements. To do so, we divide the proof into two key steps. First, we measure the syndrome of the inner blocks, and argue that their list-decodability and the expansion of $G$ produces small lists of candidate `correction' operators acting on the inner blocks. Then, we measure the syndrome of the outer code, and use its list-recoverability to infer that there can't be many global operators matching the elements of these inner-block lists.

\begin{claim} \label{claim:innersyndrome}
    Let $s_1\cdots s_n$ be any syndrome measurements of the $n$ inner code blocks, and let each $\mathcal{L}_{s_i}$, for $i\in [n]$, correspond to a list of operators on $Q_{in}$ of relative weight $\leq (1-R)/2 - 4\cdot \gamma$ which match $s_i$ and are all distinct up to a stabilizer of $Q_{in}$. Let $O_G = \pi_G(O)$ be any operator of relative weight $\leq (1-R)/2 - 8\gamma$ on $\pi_G(Q_\diamond)$, where $O = \otimes_i O_i \in (\mathcal{P}_q^{n'})^n$, and $O_i \in \mathcal{P}_q^{n'}$ is supported on the $i$th inner block. Then, if $O$ has syndrome $s_i$ on the $i$th block for $i\in [n]$, 
    \begin{equation}
        O_i \text{ is stabilizer equivalent to some element  of } \mathcal{L}_{s_i} \text{ for } \geq 1-\gamma \text{ inner blocks }i\in [n]
    \end{equation}
\end{claim}

\begin{proof}
Let $O_G = \pi_G(O)$ be any operator of relative weight $\leq (1-R)/2 - 8\cdot \gamma$ on $\pi_G(Q_\diamond)$. From \cref{prop:aelred}, the expansion of $G$ tells us that at least $\geq 1-\gamma$ inner blocks have wt$(O_i) \leq n' \cdot ((1-R)/2 - 4\cdot \gamma)$. We refer to these inner blocks as `Good' blocks. Via the list decodability of $Q_{in}$, if $i$ is Good, then $O_i$ is equivalent up to a stabilizer of $Q_{in}$, to some element of $\mathcal{L}_{s_i}$. Moreover, $|\mathcal{L}_{s_i}|\leq l$.
\end{proof}

Now, let us measure the syndrome $s^{out}$ of the outer code, by means of the encoded stabilizers of $Q_{out}$. Please refer to \Cref{fact:concatstabilizers} for a description of the stabilizers of concatenated codes. From   \Cref{claim:innersyndrome} it suffices to argue that there aren't many operators $O$ on the concatenated code with outer syndrome $s^{out}$ and where $O_i\in \mathcal{L}_{s_i}$ for most $i$. This resembles a statement about list-recoverability, but about operators on the concatenated code, not the outer code. The claim below formalizes how to reduce that statement to the list recoverability of the outer code. 

\begin{claim}\label{claim:outersyndrome}
    There are at most $L$ operators $E\in (\mathcal{P}_q^{n'})^n$ which are distinct up to a stabilizer of the concatenated code $Q_{out}\diamond Q_{in}$ and have outer syndrome $s^{out}$, inner syndromes $s_1\cdots s_n$, where $E_G = \pi_G(E)$ is an operator on the distance amplified code of relative weight $\leq (1-R-16\cdot \gamma)/2$. 
\end{claim}

As an immediate corollary of the above, we recover \Cref{lemma:distampf}.

\begin{corollary}
    $\pi_G(Q_\diamond)$ is $((1-R-O(\gamma))/2, L)$-QLD. 
\end{corollary}

\begin{proof}[Proof of \Cref{claim:outersyndrome}] The proof proceeds by using the elements of the lists $\mathcal{L}_{s_i}$ to construct operators acting on each symbol of the outer code, for which we then invoke list-recoverability. Consider, WLOG, the first element $F_1^i \in \mathcal{P}_q^{n'}$ in the list $\mathcal{L}_{s_i}$ of the $i$th inner block. For each $i\in [n]$, we construct a list of inner code normalizers
\begin{equation}
    L_i = \big\{ (F_1^i)^\dagger F_j^i \in N(Q_{in}) : F_j^i \in \mathcal{L}_{s_i}, j\in [l]\big\}
\end{equation}

Note that each element of $L_i$ is $S(Q_{in})$-logically-equivalent to a logical operator $\bar{X}_{ij} \in N(Q_{in})/S(Q_{in})$, defined by quotienting out the stabilizer group of $Q_{in}$. Thus, one can derive from $L_i$ a list of operators $\mathcal{E}_i\subset \mathcal{P}_{q^m}$ acting on the $i$th qudit of the outer code:
\begin{equation}
    \mathcal{E}_i = \bigg\{ X_{ij} = \Enc_{in}^\dagger \bar{X}_{ij} \Enc_{in} \in \mathcal{P}_{q^m}: \bar{X}_{ij} \text{ stabilizer equivalent to some element of }L_i\bigg\}
\end{equation}\

Note that $|\mathcal{L}_{s_i}|\leq l\Rightarrow |\mathcal{E}_i|\leq l$, and that the identity operator is in $\mathcal{E}_i$ (since it was in $L_i$ as well). Let $s_1^{out}$ denote the syndrome of the operator $F_1 = \otimes_i^n F_1^i$ on the encoded stabilizers of the outer code. Now, we invoke the list recoverability of the outer code to find a list $\mathcal{L}_{out}\subset \mathcal{P}_{q^m}^n$ of at most $L$ $S(Q_{out})$-logically-distinct operators $E = \otimes E_i$ acting on the outer code where $E_i \in \mathcal{E}_i$ for $\geq 1-\gamma$ qudits $i\in [n]$, and $E$ has syndrome $s^{out}-s^{out}_1$. Each element of $\mathcal{L}_{out}$ can now be represented as an operator acting on the concatenated code, by re-encoding into $Q_{in}$, and to conclude we shift back the rotation by $F_1= \otimes_i^n F_1^i$:
\begin{equation}
    \mathcal{L}_{s^{out}, s^{in}} = \big\{ F_1  \big(\otimes_i^n \Enc_{in}\big) X \big(\otimes_i^n \Enc_{in}^\dagger\big) : X\in \mathcal{L}_{out} \big\}
\end{equation}

To conclude the proof, we argue correctness of the produced list. It suffices to prove that any operator $O$ with bounded weight on the distance-amplified code with syndromes $s^{out}, s^{in} = s_1\cdots s_n$ is $S(Q_\diamond)$-stabilizer equivalent to an element of $\mathcal{L}_{s^{out}, s^{in}}$. For the purposes of contradiction, assume otherwise and let one such operator be $O = \otimes O_i$ on $Q_\diamond$. By construction, $F_1^\dagger O$ has syndrome 0 on all inner code blocks, since $(F_1^i)^\dagger O_i\in N(Q_{in})$, and thereby $F_1^\dagger O$ is $S(Q_\diamond)$-stabilizer equivalent to some encoded operator $(\otimes_i^n \Enc_{in}) X (\otimes_i^n \Enc_{in}^\dagger)$, for $X\in \mathcal{P}_{q^m}^n$. Moreover, $X_i\in \mathcal{E}_i$ for at least $1-\gamma$ locations $i$, and the outer code syndrome of $X$ is the same as that of $F_1^\dagger O$. Since the outer code syndrome of $F_1^\dagger O$ is precisely $s_{out}-s_1^{out}$, $X$ must be $S(Q_{out})$-logically-equivalent to some element returned by the list recovery, concluding the proof. 
\end{proof}

To conclude this section, we make an explicit choice of codes and parameters. Let $q = 2^{\Theta(1/\gamma)}$, and pick $Q_{in}$ to be a random $q$-ary CSS code of rate $R + 5\gamma$ which is $((1-R)/2- 4\gamma, q^{\Theta(1/\gamma)})$-QLD, see \cref{claim:randomcss}. By picking $l = q^{\Theta(1/\gamma)}$ in \Cref{theorem:fqrs}, we let $Q_{out}$ be the $m$-Folded Quantum RS code of \Cref{def:FQRS} of rate $1-4\gamma$, and blocklength $n$, folding parameter $m = \Theta(\frac{1}{\gamma^4} \log q) = \Theta(\frac{1}{\gamma^5} )$, which is $(1-\gamma, q^{O(1/\gamma)}, n^{O(\log (1/\gamma)/\gamma^3)})$-QLR. Naturally $Q_{out}, Q_{in}$ are defined over a field of the same characteristic, and we let the alphabet size of $Q_{out}$ be any composite $Q = q^{m\cdot x}$ where $x = \Theta(\log n)$, such that the block-length of $Q_{in}$ is $\Theta(m\log n)$. The resulting alphabet size is $q^D = 2^{\Theta(1/\gamma^5)}$, and list size $n^{O(\log (1/\gamma)/\gamma^3)}$. We list decode the resulting code by brute force decoding the inner codes in time $N\cdot q^{O(m\log N)} = N^{O(\frac{1}{\gamma^5})}$, and list recover $Q_{out}$ in time $N^{O_\gamma(1)}$, resulting in an efficient decoding.

\section{Generalization of \Cref{lemma:private-aqec}}
\label{section:proof-claim-LI-composable}

In this section we prove a particular parallel composability property of the private AQECCs of \Cref{lemma:private-aqec} presented in  \Cref{sec:ael}. For simplicity, we simply discuss the main modifications from \Cref{lemma:private-aqec}. We later apply this composability to study the concatenation of AQECCs.

In our proof of \Cref{lemma:private-aqec}, we actually showed the slightly stronger statement that $Q_{LD} \circ Q_k$ uniquely decodes an arbitrary adversarial error with high probability over the choice of private key $k$. In this section, we show that this stronger property continues to hold when the adversary is allowed access to a side-register entangled with the encoded message.

Consider the following setting, where Alice and Eve hold a bipartite quantum state $\ket{\psi}_{AE}$, and Alice desires to send the $A$ register to a third party, Bob. To do so, Alice has access to an adversarial quantum communication channel $\mathcal{A}_{SE}$, where Eve can adversarially tamper with up to a $\delta$ fraction $S$ of the registers of the sent message, in addition to their half $E$ of $\psi_{AE}$. Before doing so, Alice picks a random key $k\in K$, and encodes her half of the state using the private AQECC $(\Enc_k)_{A\rightarrow Q}$, and then sends her half over the channel. After receiving the corrupted code-state (and with knowledge of $k$), Bob attempts to decode using the channel $\Dec_k$ defined in \Cref{alg:algorithm1}. The density matrix $\rho^k_{BE}$ produced by Bob can be written as
\begin{equation}
 \rho^k_{BE} =  \big((\Dec_k)_{Q\rightarrow B}\otimes \mathbb{I}_{E}\big) \circ\big(\mathbb{I}_{Q\setminus S}\otimes \mathcal{A}_{SE}\big)\circ\big((\Enc_k)_{A\rightarrow Q}\otimes \mathbb{I}_{E}\big) (\psi_{AE})
\end{equation}

\Cref{theorem:composabilityaqeccs} below is a generalization of \Cref{claim:lemma41-proof}, which states that with high probability over the choice of $k$ (and the syndrome measurement outcomes), the quantum state $\rho_{BE}^k$ Bob recovers is simply as if Eve had corrupted their register $E$ without touching the $A$ register at all: i.e. $\rho_{BE}^k\approx (\mathbb{I}_A\otimes \mathcal{A}_E)(\psi)_{AE}$, for some CP operator $\mathcal{A}_E$ acting only on $E$. 

\begin{theorem}\label{theorem:composabilityaqeccs}
    Let $(\Enc_k, \Dec_k)$ be the private $(\delta, \epsilon')$-AQECC of \Cref{lemma:private-aqec}, defined via the composition of a $[[n, m, d\geq \delta n]]$ $(\delta, L)$-QLD code $Q_{LD}$ and an $\epsilon$-PTC $\{Q_k\}$ with keys $K$. Let $\ket{\psi}_{AE}$ be a bipartite quantum state, and for $k\in K$ let $|\psi_k\rangle_{QE}= (\mathbb{I}_E\otimes (\Enc_k)_{A\rightarrow Q})\ket{\psi}_{AE}$ be the encoding of the $A$ half of $\psi_{AE}$ into the AQECC. If $\mathbb{I}_{Q\setminus S}\otimes \mathcal{A}_{SB}$ is an arbitrary CP error operator supported on $B$ and a set $S\subset Q$ of at most $\delta \cdot n$ code registers, then for a random choice of $k$ the decoding algorithm $\Dec_k$ of \Cref{alg:algorithm1} produces the density matrix
\begin{equation}
  \sum_s\sum_{k\in \text{Good}_{\mathcal{A}, s}} \frac{p_{\mathcal{A}, s}}{|K|}\ket{k}\bra{k}\otimes \ket{s}\bra{s} \otimes (\mathbb{I}_A\otimes \mathcal{A}_E^s) (\psi_{AE}) + \sum_s\sum_{k\in K\setminus \text{Good}_{\mathcal{A}, s}} \frac{p_{\mathcal{A}, s}}{|K|}\ket{k}\bra{k}\otimes \ket{s}\bra{s} \otimes (\rho_{k, s})_{BE}
\end{equation}

That is, for each syndrome $s$ of the stabilizer code $Q_{LD}$, $s$ is measured with probability $p_{\mathcal{A}, s}$, $\mathcal{A}_E^s$ is a CP operator acting only on the register $E$, $(\rho_{k, s})_{BE}$ are generic bipartite PSD matrices, and $\text{Good}_{s}\subset K$ is a large subset of the PTC keys of size $\geq  |K|(1-L\cdot \epsilon)$.     
\end{theorem}

In other words, even if the adversary can corrupt a sideregister $E$ entangled with the message, with probability at least $(1-L\cdot \epsilon)$ Bob recovers the $A$ half of the original bipartite state $\psi_{AE}$ uncorrupted. 

\begin{remark}
\label{remark:envreg}
We remark that the proof of \Cref{theorem:composabilityaqeccs} below holds even if the state $\ket{\psi}$ is entangled with a third environment register to which none of the encoder, decoder, or error operators have access.
\end{remark}

To prove the theorem above, we present a generalization of \Cref{fact:syndromelocind} and the intuition in \Cref{section:QLD}. Qualititatively, it formalizes that the local indistinguishability of the stabilizer code ensures Eve learns nothing about the private encryption key $k$, even if Eve has prior entanglement with the message. Thus, if Bob performs a syndrome measurement upon receiving the corrupted codestate, he will collapse the state into a mixture of Pauli errors acting on his half of the state, and generic operators acting on Eve's half. Together with \Cref{claim:ptcsyndrome-draft} and \Cref{claim:lemma41-proof}, \Cref{claim:LI-composable-draft} below directly implies \Cref{theorem:composabilityaqeccs}.

\begin{claim}\label{claim:LI-composable-draft}
    Let $Q$ be be a $[[n,m,d]]_q$ stabilizer code, let $\ket{\psi}_{AE}$ be any bipartite quantum state, and let $|\tilde{\psi}\rangle_{QE} = (\mathbb{I}_E\otimes \Enc_{A\rightarrow Q})\ket{\psi}_{AE}$ be the encoding of half of $\ket{\psi}_{AE}$ into the code $Q$. If $\mathbb{I}_{Q\setminus S}\otimes \mathcal{A}_{SE}$ is an arbitrary CP error operator supported on $E$ and a set $S\subset Q$ of fewer than $d$ code registers, then the syndrome measurement $(M_{Q}\otimes \mathbb{I}_E)$ collapses the corrupted state into the mixture
    \begin{equation}
        (M_{Q}\otimes \mathbb{I}_E)\circ (\mathbb{I}_{Q\setminus S}\otimes \mathcal{A}_{SE})(\tilde{\psi}_{QE}) = \sum_s p_{\mathcal{A}, s} \ket{s}\bra{s} \otimes (\mathbb{I}_{Q\setminus S}\otimes (\sigma_{\mathcal{A}, s})_S\otimes \mathcal{A}^s_E)(\tilde{\psi}_{QE}), 
    \end{equation}
    That is, the syndrome $s$ is measured with probability $p_{\mathcal{A}, s}$ and the resulting state is collapsed a single Pauli error $\sigma_{\mathcal{A}, s}$ on the code $Q$, and a CP operator $\mathcal{A}^s_E$ supported only on the non-encoded half $E$.
\end{claim}

If $\mathcal{A}_{SE}$ is CPTP, then $\sum_s p_{\mathcal{A}, s}=1$. We emphasize that the distribution over syndrome measurements, $\{p_{\mathcal{A}, s}\}$, doesn't depend on the encoded half $A$ of $\ket{\psi}_{AE}$, only on the choice of error operator $\mathcal{A}$ and $\rho_E = \text{Tr}_E[\psi_{AE}]$.

\begin{proof}
    Consider a decomposition of the CP map $\mathcal{A}=\mathcal{A}_{SE}$ into Krauss operators $\{A_\nu\}$. Let us further decompose each operator $ A_\nu = \sum_{\sigma = \sigma_S\otimes \sigma_E} c^\nu_\sigma \cdot \sigma_S\otimes \sigma_E$ into a Pauli basis. The POVM $(\mathbb{I}_E\otimes M_{Q})$ is comprised of projectors $\Pi_s$ for each syndrome $s$, which we observe to collapse $A_\nu$ into Pauli's on $S\subset Q$ of syndrome $s$:
    \begin{equation}
        \bigg(\mathbb{I}_{E}\otimes (\Pi_s)_Q\bigg) (\mathbb{I}_{Q\setminus S}\otimes A_\nu ) |\tilde{\psi}\rangle_{QE} = \sum_{\sigma_S \text{ of syndrome }s} \sum_{\sigma_E} c^\nu_{\sigma_S, \sigma_E} \sigma_S\otimes \sigma_E |\tilde{\psi}\rangle_{QE}
    \end{equation}

    However, each $\sigma_S$ in the sum above is supported on $|S|\leq d$ and they all have the same syndrome $s$, and thereby must all be stabilizer-equivalent to an operator $\sigma_{\mathcal{A}, s}$ supported on $S$. One can thereby re-arrange the above
    \begin{equation}
         = (\sigma_{\mathcal{A}, s})_S \otimes \sum_{\sigma_E} a^{s, \nu}_{\sigma_E} \sigma_E |\tilde{\psi}\rangle_{QE} =(\sigma_{\mathcal{A}, s})_S \otimes (A_{s, \nu})|\tilde{\psi}\rangle_{QE} , \text{ where } a^{s, \nu}_{\sigma_E} = \sum_{\sigma_S \text{ of syndrome }s} c^\nu_{\sigma_S, \sigma_E}  
    \end{equation}

    To conclude, we denote as $p_{\mathcal{A}, s} = \sum_\nu \langle \tilde{\psi}|\mathbb{I}\otimes \mathcal{A}_{s, \nu}^\dagger \mathcal{A}_{s, \nu}  |\tilde{\psi}\rangle_{QE} = \sum_{\nu} \sum_{\sigma_E, \sigma'_E }(a^{s, \nu}_{\sigma_E})^*a^{s, \nu}_{\sigma'_E}\langle \tilde{\psi}|\mathbb{I}\otimes \sigma_E^\dagger \sigma'_E  |\tilde{\psi}\rangle_{QE} =\sum_{\nu} \sum_{\sigma_E, \sigma'_E }(a^{s, \nu}_{\sigma_E})^*a^{s, \nu}_{\sigma'_E}\cdot  \text{Tr}[\sigma_E^\dagger \sigma'_E  \rho_E]$, which is only a function of $\rho_E$ and $\mathcal{A}$.


\end{proof}

\end{document}